\documentclass[pra, onecolumn,superscriptaddress,notitlepage, preprintnumbers,amsmath,amssymb,nobalancelastpage,longbibliography]{revtex4-1}

\usepackage{tikz}
\usetikzlibrary{shapes,arrows,positioning}
\usepackage{standalone}
\usepackage{amssymb,amsmath,amstext}
\usepackage{mathtools}
\usepackage{graphicx}
\usepackage{epstopdf}
\usepackage{color}
\usepackage{appendix}
\usepackage[T1]{fontenc}
\usepackage{bbold}
\usepackage{bbm} 
\usepackage{float}
\usepackage{latexsym}
\usepackage[colorlinks=true,citecolor=blue,linkcolor=magenta]{hyperref}
\usepackage[latin1]{inputenc}
\usepackage{bibunits}
\defaultbibliographystyle{unsrt}

\usepackage{graphicx, color, graphpap}% Include figure files
\usepackage{enumitem}
\usepackage{amssymb}
\usepackage{amsmath,amsthm}
\usepackage{float}
\usepackage{multirow}
\usepackage[colorlinks=true,citecolor=blue,linkcolor=magenta]{hyperref}
\usepackage[T1]{fontenc}
\usepackage{bbm}
\usepackage{thmtools,thm-restate}
\usepackage{verbatim}
\usepackage{mathtools}
\usepackage{xspace}
\usepackage{amssymb}
\long\def\ca#1\cb{} %Use for commenting out: \ca...\cb

\newcommand{\ket}[1]{|#1\rangle}               %ket
\newcommand{\bra}[1]{\langle #1|}              %bra
\newcommand{\dya}[1]{\ket{#1}\!\bra{#1}}
\newcommand{\matl}[3]{\langle #1|#2|#3\rangle} %matrix element

\newcommand{\QC}{\ensuremath{\mathsf{QC}}\xspace}

\newcommand{\HST}{\ensuremath{\mathsf{HST}}\xspace}
\newcommand{\LHST}{\ensuremath{\mathsf{LHST}}\xspace}

\newcommand{\LET}{\ensuremath{\mathsf{LET}}\xspace}
\newcommand{\LLET}{\ensuremath{\mathsf{LLET}}\xspace}

\newcommand{\DC}{\mathcal{D}}
\newcommand{\EC}{\mathcal{E}}

\newcommand{\HC}{\mathcal{H}}
\newcommand{\IC}{\mathcal{I}}

\newcommand{\NC}{\mathcal{N}}
\newcommand{\PC}{\mathcal{P}}

\newcommand{\UC}{\mathcal{U}}
\newcommand{\VC}{\mathcal{V}}
\newcommand{\WC}{\mathcal{W}}

\newcommand{\Tr}{{\rm Tr}}

\newcommand{\ave}[1]{\langle #1\rangle}               %average
\renewcommand{\geq}{\geqslant}
\renewcommand{\leq}{\leqslant}
\newcommand{\mte}[2]{\langle#1|#2|#1\rangle }
\newcommand{\mted}[3]{\langle#1|#2|#3\rangle }

\newcommand{\Ct}{\widetilde{C}}

\newcommand{\VB}{\mathbb{V}}

\renewcommand{\vec}[1]{\boldsymbol{#1}}  % Bold vectors instead of arrow vectors

\newcommand{\ad}{^\dagger}
\newcommand*{\id}{\openone}

\newcommand{\Pt}{\widetilde{P}}

\newcommand{\Ft}{\widetilde{F}}
\newcommand{\jj}{^{(j)}}
\newcommand{\opt}{\text{opt}}

{}%         Body font
{}%         Indent amount (empty = no indent, \parindent 
\newtheorem{theorem}{Theorem}
\newtheorem{theoremApp}{Theorem}
\newtheorem{lemma}{Lemma}
\newtheorem{corollary}{Corollary}
\newtheorem{corollaryApp}{Corollary}
\newtheorem{definition}{Definition}

\begin{document}
\title{Noise Resilience of Variational Quantum Compiling}

\author{Kunal Sharma} 
%\thanks{The first two authors contributed equally to this work.}
\address{Theoretical Division, Los Alamos National Laboratory, Los Alamos, NM 87545, USA}
\address{Hearne Institute for Theoretical Physics and Department of Physics and Astronomy, Louisiana State University, Baton Rouge, LA USA.}

\author{Sumeet Khatri} 
%\thanks{The first two authors contributed equally to this work.}
\address{Hearne Institute for Theoretical Physics and Department of Physics and Astronomy, Louisiana State University, Baton Rouge, LA USA.}

\author{M. Cerezo}
\address{Theoretical Division, Los Alamos National Laboratory, Los Alamos, NM 87545, USA}
\address{Center for Nonlinear Studies, Los Alamos National Laboratory, Los Alamos, NM, USA
}

\author{Patrick J. Coles}
\address{Theoretical Division, Los Alamos National Laboratory, Los Alamos, NM 87545, USA}

\begin{abstract}
Variational hybrid quantum-classical algorithms (VHQCAs) are near-term algorithms that leverage classical optimization to minimize a cost function, which is efficiently evaluated on a quantum computer. Recently VHQCAs have been proposed for quantum compiling, where a target unitary $U$ is compiled into a short-depth gate sequence $V$. In this work, we report on a surprising form of noise resilience for these algorithms. Namely, we find one often learns the correct gate sequence $V$ (i.e., the correct variational parameters) despite various sources of incoherent noise acting during the cost-evaluation circuit. Our main results are rigorous theorems stating that the optimal variational parameters are unaffected by a broad class of noise models, such as measurement noise, gate noise, and Pauli channel noise. Furthermore, our numerical implementations on IBM's noisy simulator demonstrate resilience when compiling the quantum Fourier transform, Toffoli gate, and W-state preparation. Hence, variational quantum compiling, due to its robustness, could be practically useful for noisy intermediate-scale quantum devices. Finally, we speculate that this noise resilience may be a general phenomenon that applies to other VHQCAs such as the variational quantum eigensolver.
\end{abstract}

\maketitle

\section{Introduction}

Obtaining accurate answers from near-term quantum computers is a challenge with major scientific and technological implications. In these so-called noisy intermediate-scale quantum (NISQ) computers~\cite{preskill2018quantum}, errors arise, for example, due to decoherence processes, gate noise, and measurement noise. Clearly, error mitigation techniques will be necessary to make use of NISQ devices. Several promising error mitigation strategies have recently emerged, including zero-noise extrapolation~\cite{temme2017error}, quasi-probability decomposition~\cite{temme2017error}, post-selection~\cite{linke2018measuring, subasi2018entanglement}, noise-aware compiling~\cite{murali2019noise}, and machine learning for circuit-depth compression~\cite{cincio2018learning}. Let us consider two other strategies for error mitigation in what follows.

Hybridizing a quantum algorithm by pushing some of the complexity onto a classical computer allows one to only run a portion of the computation on the (error-prone) quantum computer. Excellent examples of this strategy are variational hybrid quantum-classical algorithms (VHQCAs)~\cite{mcclean2016theory}. VHQCAs only employ a quantum computer to evaluate a cost function that depends on the parameters of a quantum gate sequence and then leverage a classical optimization routine to minimize the cost and hence train the parameters. The most famous VHQCA is the variational quantum eigensolver (VQE)~\cite{peruzzo2014VQE}, where the cost function is the energy for some Hamiltonian and hence the goal is to prepare the ground state. VHQCAs have been proposed for many other applications~\cite{farhi2014QAOA, johnson2017qvector, romero2017quantum, larose2018, arrasmith2019variational, cerezo2019variational, jones2019variational, yuan2018theory, li2017efficient, kokail2019self, Khatri2019quantumassisted, jones2018quantum, heya2018variational,carolan2019variational}.

Another strategy for error mitigation is to find quantum circuits or quantum algorithms that are inherently noise resilient. Circuits for quantum error correction~\cite{devitt2013quantum, fowler2012surface}, of course, have this property of inherent noise resilience, and in fact, such circuits are resilient to all types of noise on a subset of the qubits. More generally, one could ask whether a circuit is resilient to a particular kind of noise process. Hence, for every circuit, which aims to compute some quantity, one could ask what noise models do not affect the output of the circuit.

The two strategies just mentioned have an interesting intersection: researchers have observed that some VHQCAs have some inherent noise resilience. McClean et al.~\cite{mcclean2016theory} noted that coherent errors (e.g., systematic gate biases) can lead to a situation where the formal unitary $V(\vec{\alpha})$ specified by the parameters $\vec{\alpha}$ is different from the actual unitary that is physically implemented $\widetilde{V}(\vec{\alpha})$. This error is correctable if there exists a vector $\vec{\beta}$ such that one can physically implement the unitary $\widetilde{V}(\vec{\alpha} + \vec{\beta})$ within one's ansatz, with the condition that $\widetilde{V}(\vec{\alpha} + \vec{\beta}) = V(\vec{\alpha})$. If this condition is satisfied, then one could still physically achieve the minimum value of the cost function, where the minimum value would be associated with different parameters than one would have in the noiseless case. We refer to this kind of noise resilience as {\it Cost Value Resilience}, since the value of the cost function at the global minimum is unaffected by the noise. Cost Value Resilience is important, e.g., if one is interested in estimating the ground state energy of a Hamiltonian with VQE.

In this work, we report on a different kind of noise resilience for VHQCAs. Instead of considering Cost Value Resilience, we consider the case where the optimal parameters are noise resilient, which we call {\it Optimal Parameter Resilience}. While Cost Value Resilience is related to coherent noise, we find that Optimal Parameter Resilience holds for certain kinds of incoherent noise, such as decoherence processes and readout errors. For certain applications, obtaining the correct optimal parameters is more important than obtaining the correct value of the cost function.

Quantum compiling~\cite{chong2017programming,haner2018software,venturelli2018compiling} is one of these applications. Compiling refers to transforming a high-level algorithm into a low-level machine code. For quantum compiling, it is crucial to do this transformation optimally, i.e., to keep the low-level code as short as possible, since errors accumulate with circuit depth. VHQCAs offer a promising framework for (optimal) quantum compiling. Three recent works introduced VHQCAs for quantum compiling, henceforth referred to as variational quantum compiling (VQC)~\cite{Khatri2019quantumassisted, jones2018quantum, heya2018variational}. In VQC one trains the parameters $\vec{\alpha}$ of a short-depth gate sequence $V(\vec{\alpha})$ such that it is close to a target unitary $U$. Here, some distance measure between $V(\vec{\alpha})$ and $U$ serves as the cost function and is efficiently evaluated on a quantum computer, while a classical optimizer adjusts the parameters $\vec{\alpha}$ to minimize the cost. VQC could be an important tool for NISQ computing since it could optimally shrink the depth of quantum circuits. However, a potential issue is that one needs to put the target unitary $U$ on the NISQ device, and hence the target itself is noisy or defective. Furthermore, there are noise sources in other parts of the cost-evaluation circuit. All of these may lead to a defective optimal $V(\vec{\alpha})$, with the noise effectively compiled into $V(\vec{\alpha})$.

Addressing these concerns, our main results are rigorous theorems stating that many different types of noise during cost evaluation do not affect the optimal $V(\vec{\alpha})$. For example, we show that VQC is resilient to measurement noise (readout error). We also show resilience to incoherent gate noise and decoherence processes, such as Pauli channels and non-unital Pauli channels, acting at specific times during the cost-evaluation circuit. In addition to these analytical results, we implement VQC on IBM's noisy quantum simulator~\cite{cross2017open} (which simulates their quantum hardware) for several quantum gates: quantum Fourier transform, Toffoli, and W-state preparation. In each case, we observed significant noise resilience (even more resilience than what is explained by our theorems) such that we effectively learned the true optimal values of $\vec{\alpha}$ despite the noise.

Finally, we speculate that the resilience phenomenon that we demonstrate for VQC may be more general, potentially applying to other VHQCAs. For example, we discuss the potential for seeing this resilience for VQE, and as a warm-up for the reader, we give a simple example in the next section where VQE exhibits Optimal Parameter Resilience. We also establish in the Discussion section that VQC is a special case of VQE, and hence our main results can be viewed as being relevant to VQE.

\section{Warm-up: Simple VQE example}\label{sec:VQE}

Here we show that VQE~\cite{peruzzo2014VQE} exhibits Optimal Parameter Resilience (OPR) to uncorrelated measurement noise for a special class of Hamiltonians. VQE may exhibit OPR more generally, although the proof would certainly be more involved. Hence we consider here this special case for illustration  and leave the more general case for future work. 

Consider a Hamiltonian that is a sum of local Pauli operators
\begin{equation}
\label{eqn1}
    H = -\sum_{j=1}^n c^{(j)} \sigma^{(j)}_{{w(j)}} \,,
\end{equation}
where $\sigma^{(j)}_{w(j)} = U^{(j)}_{w(j)}  \sigma^{(j)}_{z} (U^{(j)}_{w(j)})\ad $ is a local operator on qubit $j$ that is unitarily equivalent to the Pauli $z$ operator $\sigma^{(j)}_{z}$.
Physically, this Hamiltonian arises for a system of $n$ non-interacting spin-1/2 particles in a non-uniform (i.e., $j$-dependent) magnetic field. Without loss of generality, one can take the $ c^{(j)}$ coefficients to be non-negative (i.e., absorb any negativity into the definition of the Pauli operator). The ground state $\ket{\psi_0}$ of $H$ has a tensor product form: $\ket{\psi_0} = \bigotimes_{j=1}^n \ket{{w(j)}_{+}}$, where $\ket{{w(j)}_{+}}$ is the eigenvector of $\sigma^{(j)}_{{w(j)}}$ with the $+1$ eigenvalue.

Now suppose there is measurement noise in the cost-evaluation circuit. In the ideal case, one measures
$\ave{H} = \sum_j c^{(j)} \ave{\sigma^{(j)}_{{w(j)}}} = \sum_j c^{(j)} \ave{U^{(j)}_{w(j)}  \sigma^{(j)}_{z} (U^{(j)}_{w(j)})\ad}$ by applying $(U^{(j)}_{w(j)})\ad$ on the $j$-th qubit and measuring it on the standard basis to estimate $\ave{\sigma^{(j)}_{{w(j)}}}$. Then, by performing classical post-processing we compute the weighted sum in $\ave{H}$. However, with measurement noise, the $\sigma^{(j)}_{z}$ operator gets replaced by $\widetilde{\sigma}^{(j)}_{z} = (p_{00}^{(j)} - p_{10}^{(j)})\dya{0} - (p_{11}^{(j)} - p_{01}^{(j)})\dya{1}$. Here, $p_{kl}^{(j)}$ is the probability to obtain the $k$ outcome when feeding in the $\ket{l}$ state on the $j$-th qubit. Hence, instead of measuring $\ave{\sigma^{(j)}_{{w(j)}}}$, one measures $\ave{\widetilde{\sigma}^{(j)}_{{w(j)}}}$ with $\widetilde{\sigma}^{(j)}_{{w(j)}} = U^{(j)}_{w(j)}  \widetilde{\sigma}^{(j)}_{z}  (U^{(j)}_{w(j)})\ad  $. In other words, the Hamiltonian $H$ gets replaced by an effective Hamiltonian: 
\begin{equation}
\label{eqn2}
  \widetilde{H} = -\sum_{j=1}^n c^{(j)} \widetilde{\sigma}^{(j)}_{{w(j)}}  \,. 
\end{equation}

The ground state of $\widetilde{H}$ is a tensor product of one-qubit states that are the eigenvectors of $\widetilde{\sigma}^{(j)}_{{w(j)}}$ with the largest eigenvalue. Suppose we assume that $p_{00}\jj+ p_{11}\jj > p_{01}\jj + p_{10}\jj $ for all $j$, which means that the probability of getting the correct outcome is greater than the probability for getting the wrong outcome. With this assumption, the largest eigenvalue of $\widetilde{\sigma}^{(j)}_z$ is associated with the $\ket{0}$ state, and hence the largest eigenvalue of $\widetilde{\sigma}^{(j)}_{{w(j)}}$ is associated with $\ket{{w(j)}_{+}}$. Therefore, despite the measurement noise, one still finds that the ground state is $\ket{\psi_0} = \bigotimes_{j=1}^n \ket{{w(j)}_{+}}$. This implies that one would still learn the correct optimal parameters of the state-preparation circuit if one implemented VQE for this Hamiltonian.

\section{Background: Variational Quantum Compiling}\label{sec:background}

Let us now move on to Variational Quantum Compiling (VQC). VQC was first introduced in Ref.~\cite{Khatri2019quantumassisted}, under the name of Quantum-Assisted Quantum Compiling (QAQC). Two later works further investigated VQC~\cite{jones2018quantum, heya2018variational} with slightly different approaches. Since we are attempting to unite these works~\cite{Khatri2019quantumassisted,jones2018quantum, heya2018variational} under one umbrella, we are proposing the name VQC (instead of QAQC) as a unifying term.

There are two overarching approaches to VQC. One is to compile the full unitary matrix $U$ by considering the action of $U$ on all input states (or an informationally complete set of states)~\cite{Khatri2019quantumassisted, heya2018variational}. The other is to compile only a particular column of the matrix $U$ by considering the action of $U$ on a fixed input state~\cite{Khatri2019quantumassisted,jones2018quantum}.
The benefit of the first approach is that it is fully general, applying even when one does not know what the input state to $U$ will be (for example, if $U$ occurs in the middle of one's quantum algorithm). The benefit of the second approach is that, when the input state is known, it could lead to a shorter-depth compilation since it does not require compilation of the entire unitary matrix.

\subsection{Full unitary matrix compiling}

Full Unitary Matrix Compiling (FUMC) was treated in detail in Ref.~\cite{Khatri2019quantumassisted}. This work introduced cost functions based on the entanglement fidelity and proposed quantum circuits to quantify the cost based on the overlap between maximally entangled states. A slightly different but equivalent approach was employed in Ref.~\cite{heya2018variational}. We focus on the approach of Ref.~\cite{Khatri2019quantumassisted} in what follows.

\begin{figure}[t]
    \centering
    \includegraphics[width=\columnwidth]{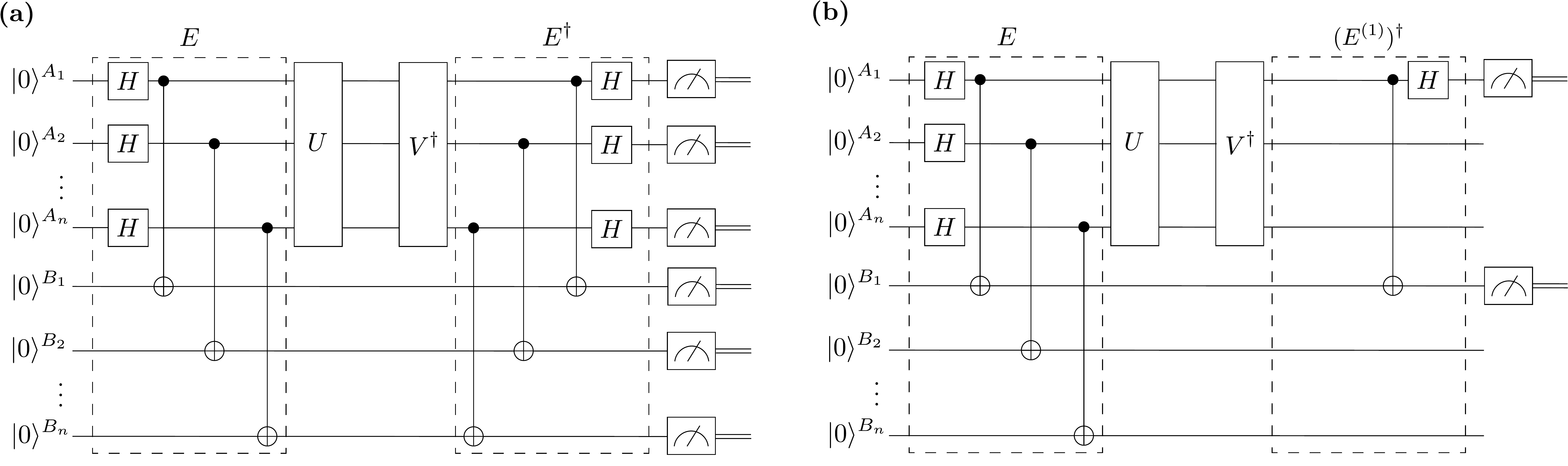}
    \caption{Circuits for cost evaluation in full unitary matrix compiling. (a) The Hilbert-Schmidt Test (HST). An entangling gate $E$, consisting of Hadamards and CNOTs, prepares a maximally entangled state between systems $A$ and $B$. Then a target unitary $U$ is applied on $A$, which is followed by a trainable unitary $V^{\dagger}$. Finally, a measurement in the Bell basis is performed by applying the adjoint of $E$, followed by a standard basis measurement. This circuit computes the Hilbert-Schmidt inner product between $U$ and $V$, as the probability to obtain the measurement outcome in which all $2n$ qubits are in the $\ket{0}$ state is  $F_{\HST}=(1/2^{2n})\vert \Tr(V^{\dagger}U)\vert^2$.   (b) The Local Hilbert-Schmidt Test (LHST), which is same as the HST circuit, except the disentangling gate $E^{\dagger}$ is applied only on one $A_jB_j$ pair of qubits (depicted here for the $A_1 B_1$ pair) and subsequently, the same two qubits are measured in the standard basis. The probability for the outcome associated with the $\ket{00}$ state is $F\jj_{\LHST}$ in \eqref{eqn6}.}
    \label{fig:HST-LHST-circuit}
\end{figure}

Two cost functions were considered in Ref.~\cite{Khatri2019quantumassisted}. One cost function $C_{\HST}$ quantifies the Hilbert-Schmidt inner product between the target unitary $U$ and the trainable gate sequence $V$, as follows:
\begin{equation}
\label{eqn4}
    C_{\HST}  = 1 - F_{\HST}\,, \quad \text{with} \quad  F_{\HST}=  |\Tr(V\ad U)|^2 / d^2\,,
\end{equation}
where $d = 2^n$ is the Hilbert-space dimension and $n$ is the number of qubits that $U$ acts on, and where we write $V$ instead of $V(\vec{\alpha})$ for simplicity. The circuit for computing $C_{\HST}$ is called the Hilbert-Schmidt Test (HST) and is shown in Fig.~\ref{fig:HST-LHST-circuit}(a). First, one prepares a maximally entangled state $\ket{\Phi}^{AB}$ by acting with a depth-two circuit $E$, then one applies $U$ followed by $V\ad$ on half of this maximally entangled state. Finally one measures the overlap with the original maximally entangled state $\ket{\Phi}^{AB}$ by applying $E\ad$ and quantifying the probability of the all-zeros measurement outcome. One can verify that this probability is equal to $F_{\HST}=|\Tr(V\ad U)|^2 / d^2$. This cost function is operationally meaningful since it is equivalent to the average fidelity $\overline{F}(U,V) = \int  |\mte{\psi}{V\ad U}|^2 d\psi$ between states acted upon by $U$ versus those acted upon by $V$, as follows \cite{horodecki1999general,nielsen2002simple}:
\begin{equation}
\label{eqn5}
    C_{\HST}  = \frac{d+1}{d}(1-\overline{F}(U,V))\,.
\end{equation}
Note that $C_{\HST}$ is faithful in that $C_{\HST} = 0$ iff $V = U$ (up to a global phase).

An alternative cost function~\cite{Khatri2019quantumassisted} is given by
\begin{equation}
\label{eqn6}
    C_{\LHST}  = 1 - F_{\LHST}\,, \quad \text{with} \quad  F_{\LHST}\  =  \frac{1}{n}\sum_{j=1}^n F\jj_{\LHST}\,,
\end{equation}
where $F\jj_{\LHST}$ is the probability of the $00$ measurement outcome in the Local Hilbert-Schmidt Test (LHST), which is the circuit shown in Fig.~\ref{fig:HST-LHST-circuit}(b). Note that $F_{\HST}$ is the entanglement fidelity for the quantum channel defined by $V\ad U$. On the other hand, $F\jj_{\LHST}$ is the entanglement fidelity for the quantum channel obtained from feeding into $V\ad U$ the maximally mixed state on $\overline{A}_{j}$ and then tracing over $\overline{A}_{j}$, where $\overline{A}_{j}$ consists of all qubits in $A$ other than $A_j$. As shown in Ref.~\cite{Khatri2019quantumassisted},
\begin{equation}
\label{eqn7}
    C_{\LHST}  \leq C_{\HST} \leq n C_{\LHST}\,,
\end{equation}
which implies that $C_{\LHST}$ is also a faithful cost function, i.e., $C_{\LHST} = 0$ iff $V = U$ (up to a global phase).

The overall cost function proposed by Ref.~\cite{Khatri2019quantumassisted} was a convex combination of $C_{\HST}$ and $C_{\LHST}$:
\begin{equation}
\label{eqn3}
    C(q) = q C_{\HST} + (1-q) C_{\LHST}\,.
\end{equation}
Here, $q$ is a free parameter with $0\leq q \leq 1$. The definition of $C(q)$ was motivated in Ref.~\cite{Khatri2019quantumassisted} by the fact that $C_{\HST}$ has a direct operational meaning (Eq.~\eqref{eqn5}) but it becomes difficult to train for large $n$ due to a vanishing gradient~\cite{cerezo2020cost}, whereas $C_{\LHST}$ is trainable but does not have a direct operational meaning. Hence one can take a weighted average of these two functions, where for small $n$ one can choose $q\approx 1$, while for large $n$ one can choose $q\approx 0$.

\subsection{Compiling with a fixed input state}

Fixed Input State Compiling (FISC) of a unitary matrix was introduced in \cite{jones2018quantum,Khatri2019quantumassisted} and treated in significant detail in \cite{jones2018quantum}. In this case, the goal is to train a gate sequence $V$ so that it has the same effect as a target unitary $U$ when acting on a given input state $\ket{\psi_0}$. For simplicity and due to its technological relevance, we will consider the case where $\ket{\psi_0} = \ket{\vec{0}}$ is the all-zero state, so that we are interested in training $V$ to satisfy  (up to a global phase):
\begin{equation}
    U\ket{\vec{0}}=V\ket{\vec{0}}\,,\quad\text{or equivalently}\quad W\ket{\vec{0}} = \ket{\vec{0}}\,,  \label{eq-SE-compiler}
\end{equation}
with $W = V\ad U$.  To quantify how far 
$W\ket{\vec{0}}$ is from the  state $\ket{\vec{0}}$, one can define the cost function
\begin{equation}
\label{eqn9} 
    C_{\LET}=1-G_{\LET}\,,
\end{equation}
where $G_{\LET}$ is the fidelity $F(\rho,\sigma)=\left(\Tr[\sqrt{\sqrt{\rho}\sigma\sqrt{\rho}}]\right)^2$ between these two states:
\begin{equation}
    G_{\LET} = 
    F(\dya{\vec{0}},W\dya{\vec{0}}W\ad) = |\matl{\vec{0}}{W}{\vec{0}}|^2=\Tr\left[P_{\vec{0}} W \dya{\vec{0}}W\ad \right]\,, \label{eq-f-SE}
\end{equation}
with $P_{\vec{0}}=\dya{\vec{0}}$ the projector onto the all-zero state. We employed the LET subscript here since we refer to the circuit used to quantify \eqref{eqn9} and \eqref{eq-f-SE} as the Loschmidt Echo Test (LET), shown in Fig.~\ref{fig:SE-LLE-circuit}(a). The Loschmidt echo~\cite{goussev2012loschmidt} refers to a forward and backward time evolution with the intent of recovering the initial state. This is analogous to the circuit in Fig.~\ref{fig:SE-LLE-circuit}(a) where one first evolves forward with $U$ and then attempts to undo that evolution with $V\ad$, to recover the initial state $\ket{\vec{0}}$. Hence the probability of the all-zero measurement outcome in Fig.~\ref{fig:SE-LLE-circuit}(a) is precisely $G_{\LET}$.

\begin{figure}[t]
    \centering
    \includegraphics[scale=0.3]{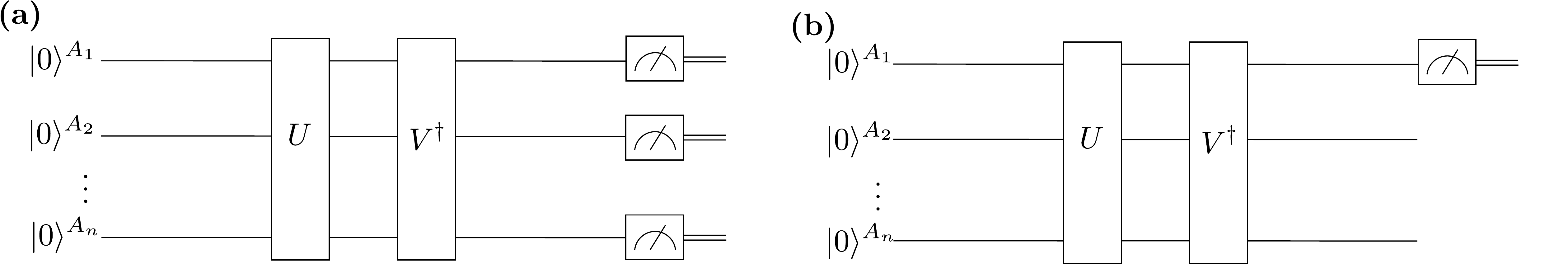}
    \caption{Circuits for cost evaluation in compiling with a fixed input state. (a) The Loschmidt Echo Test (LET). In this circuit, the probability of obtaining the measurement outcome in which all $n$ qubits are in the $\ket{0}$ state is   $G_{\LET}= |\matl{\vec{0}}{V\ad U}{\vec{0}}|^2$.  (b) The Local Loschmidt Echo Test (LLET), which is the same as the LET but only the $A_j$ qubit is measured. The probability that this qubit is in the $\ket{0}$ state is $G\jj_{\LLET}$ in \eqref{eq-LLET-func}.}
    \label{fig:SE-LLE-circuit}
\end{figure}

One can see that compiling with a fixed input state leads to more freedom and hence more solutions than full unitary matrix compiling. Note that $C_{\HST} =0$ iff $W = e^{i\phi}\id$ where $\phi$ is a global phase factor. On the other hand, $C_{\LET} =0$ iff $|\mted{\vec{0}}{W}{\vec{z}}| = |\mted{\vec{z}}{W}{\vec{0}}| = \delta_{\vec{z},\vec{0}}$ for all bit strings $\vec{z}$. Hence, for $W$ that achieve $C_{\LET} =0$, the $(n-1)\times(n-1)$ unitary principal submatrix of $W$ with matrix elements  $\mted{\vec{z}}{W}{\vec{z'}}$ (such that $\vec{z},\vec{z'}\neq\vec{0}$) remains completely arbitrary. This degeneracy of optima can simplify the optimization of $V$ as any of these optima will lead to $C_{\LET}=0$.

Analogous to the LHST cost for full unitary matrix compiling, one can define a cost function for fixed input state compiling that involves local observables:
\begin{equation}\label{eqn11}
    C_{\LLET}=1-G_{\LLET} = 1- \frac{1}{n}\sum_{j=1}^n G_{\LLET}^{(j)} \,, \qquad \text{with} \qquad G_{\LLET}^{(j)}=\Tr\left[\left(P_{0}^{A_j}\otimes\id^{\overline{A}_j}\right) W \dya{\vec{0}}W\ad \right]\,.
\end{equation}
Here, $P_{0}^{A_j}$ is the projector onto the zero state on the $A_j$ qubit, and $\id^{\overline{A}_j}$ denotes the identity on all qubits except $A_j$ and $n$ is the number of qubits. We call the circuit used to compute $C_{\LLET}$ the Local Loschmidt Echo Test (LLET), and this circuit is shown in Fig.~\ref{fig:SE-LLE-circuit}(b). Note that
\begin{equation}
    G_{\LLET}^{(j)}=\Tr_{A_j}\left[P_{0}^{A_j} \rho\jj \right]=\mte{0}{\rho\jj} = F(\dya{0},\rho\jj )\,, \label{eq-LLET-func}
\end{equation}
where $\rho\jj=\Tr_{\overline{A}_j}\left[W \dya{\vec{0}}W\ad \right]$. Hence $G\jj_{\LLET}$ corresponds to the probability of the zero outcome for the circuit in Fig.~\ref{fig:SE-LLE-circuit}(b). With a proof similar to that of \eqref{eqn7} one can show that
\begin{equation}
C_{\LLET}\leq C_{\LET}\leq n  C_{\LLET}\,,
\end{equation}
and hence $C_{\LLET}= 0$ iff $C_{\LET}=0$. Furthermore, one can define an overall cost function analogous to $C(q)$ in \eqref{eqn3},
\begin{equation}
\label{eqnCprimeQ}
C'(q) = q C_{\LET} + (1-q)C_{\LLET},
\end{equation}
which again is motivated by the fact that $C_{\LET}$ has a direct operational meaning but is difficult to train for large $n$, whereas the opposite is true for $C_{\LLET}$. Hence one can take $q\approx 1$ for small $n$ and $q\approx 0$ for large $n$.

\section{Noise Processes}\label{sec:noise}

In this work, we consider three different types of noise \cite{nielsen2010,wilde2017}: (1) decoherence noise, (2) gate noise, and (3) measurement noise. We now discuss how we mathematically model these three types of noise.

Let us start with decoherence. Physical models of decoherence often refer to $T_1$ and $T_2$ processes, which respectively pertain to thermal relaxation (energy dissipation) and dephasing (loss of phase coherence). These processes are typically modeled as local quantum channels acting independently on individual qubits. However, mathematically it is easier to deal with classes of quantum channels that act globally on sets of qubits (which can contain the independent local channels as a special case). In what follows, we define three types of global quantum channels: depolarizing noise, Pauli noise, and non-unital Pauli noise. It is worth noting that Pauli noise includes $T_2$ processes as a special case (i.e., the dephasing channel is a Pauli channel), and non-unital Pauli noise includes $T_1$ processes as a special case (i.e., the amplitude damping channel is a non-unital Pauli channel). Consider the following precise definitions.

\begin{definition}\label{def-GDN}
We define Depolarizing Noise (DN) as a Completely Positive Trace-Preserving (CPTP) map that maps an $n$-qubit state $\rho$ to the state $p\rho + (1-p) \id / (2^n)$.
\end{definition}

\begin{definition}\label{def-UPC}
We define Pauli Noise (PN) as a CPTP map $\PC$ whose superoperator is diagonal in the Pauli basis. In other words, its action on a Pauli operator $X^{\vec{l}} Z^{\vec{k}} \coloneqq X^{l_1}Z^{k_1}\otimes ... \otimes X^{l_n}Z^{k_n}$ is given by $\PC( X^{\vec{l}} Z^{\vec{k}} ) = c_{\vec{l}\vec{k}} X^{\vec{l}} Z^{\vec{k}}$, where $c_{\vec{0}\vec{0}}=1$. Furthermore, we assume that $c_{\vec{l}\vec{k}}  \geq 0$ for all $\vec{l}$ and $\vec{k}$, where $l_1, \dots, l_n, k_1, \dots, k_n \in \{0, 1\}$.
\end{definition}

\begin{definition}
We define Non-Unital Pauli Noise (NUPN) as a CPTP map $\PC_{\text{NU}}$ whose action on the identity is $\PC_{\text{NU}}(\id) = \id + \sum_{(\vec{l},\vec{k}) \neq (\vec{0},\vec{0}) } d_{\vec{l}\vec{k}}X^{\vec{l}}Z^{\vec{k}}$, and whose action on all other Pauli operators $X^{\vec{l}} Z^{\vec{k}}$ with $(\vec{l},\vec{k}) \neq (\vec{0},\vec{0})$ is given by $\PC_{\text{NU}}( X^{\vec{l}} Z^{\vec{k}} ) = c_{\vec{l}\vec{k}} X^{\vec{l}} Z^{\vec{k}}$. Furthermore, we assume that $c_{\vec{l}\vec{k}}  \geq 0$ for all $\vec{l}$ and $\vec{k}$.
\end{definition}

Next, we consider gate noise. While gate noise can involve coherent errors such as systematic gate bias, such errors are hardware-specific, and hence we focus on incoherent gate noise. We consider a simple model for gate noise in which every time a gate is implemented, a Pauli channel acts both before and after this gate. Furthermore, for generality, we allow these Pauli channels to act globally on all qubits, which serves as a model for cross-talk (where gates affect qubits on which they are intended to act trivially).

\begin{definition}
We define Pauli gate noise (PGN) as a simple noise model in which all gates are preceded and followed by global Pauli channels. In other words, for a gate $G$, instead of its action on a state $\rho$ being $G\rho G\ad$, we model its action as $\PC'(G \PC(\rho) G\ad)$ where $\PC$ and $\PC'$ are Pauli channels. Note that these Pauli channels act on all qubits, including qubits on which $G$ acts trivially. 
\end{definition}

Finally, we consider measurement noise, also known as readout error. For a single qubit, we model measurement noise as a classical bit-flip channel, where feeding in the standard basis state $\ket{l}$ leads to the $k$ outcome with probability $p_{kl}$. We allow for asymmetry in that one can have $p_{01}\neq p_{10}$, which is an important generality, e.g., when $T_1$ noise occurs during the measurement process. For multiple qubits, our measurement noise model is a tensor product of the aforementioned bit-flip channels, corresponding to uncorrelated measurement noise.

\begin{definition}\label{def-noisy-POVM}
We define Measurement Noise (MN) as a modification of the standard-basis POVM elements, which  are $\{P_0 = \dya{0}, P_1 =\dya{1}\}$ for a noiseless single qubit. With measurement noise, this POVM gets replaced by $\{\Pt_{0}, \Pt_{1} \}$, with $\Pt_{0} = p_{00} \dya{0} + p_{01} \dya{1}$ and $\Pt_{1} =  p_{10} \dya{0} + p_{11} \dya{1}$, where $p_{00}+ p_{10} = 1$, $p_{01}+ p_{11} = 1$, and $p_{kl}$ is the probability of getting the $k$ outcome given the $l$ input. Furthermore we assume that $p_{kk} > p_{kl}$ for $l\neq k$. Hence, for an $n$-qubit standard-basis measurement with measurement noise, we write the POVM element associated with the bit string $\vec{z} = (z_1,...,z_n)$ as 
\begin{equation}
\Pt_{\vec{z}} = \bigotimes_{j=1}^{n} \left( p_{z_j 0}^{(j)} \dya{0} + p_{z_j 1}^{(j)} \dya{1} \right)\,,    
\end{equation}
with $\sum_{z_j} p_{z_j 0}^{(j)} = 1$ and $\sum_{z_j} p_{z_j 1}^{(j)} = 1$, and we assume that $p_{z_jz_j}^{(j)} > p_{z_jl}^{(j)} $ for $l\neq z_j$. 
\end{definition}

\section{Main Results}

Before proceeding to the main results we first define two versions of Optimal Parameter Resilience (OPR), i.e., of  learning the correct gate sequence $V$ despite various sources of noise, which we refer to as strong-OPR and weak-OPR.

\begin{definition}\label{def:OPR}
Let $\VB_d$ be the set of $d\times d$ unitary matrices. Let $C_{\QC}(V)$ be a cost function of $V$ with $V\in \VB_d$, and suppose that $C_{\QC}(V)$ can be evaluated using a quantum circuit denoted $\QC$. Let $\Ct_{\QC}(V)$ denote the noisy version of $C_{\QC}(V)$, i.e., the corresponding function whenever the circuit $\QC$ is run in the presence of some noise process $\NC$. Let $\VB_d^{\opt}$ and $\widetilde{\VB}_d^{\opt}$ respectively denote the sets of unitaries that optimize $C_{\QC}(V)$ and $\Ct_{\QC}(V)$, i.e.,
\begin{align}
    \label{eqnOptimalSet}
    \mathbb{V}_d^{\opt} &= \{V' \in \VB_d : C_{\QC}(V') = \min_{V \in \VB_d} C_{\QC}(V)\}\,,\\
\widetilde{\VB}_d^{\opt} &= \{V' \in \VB_d : \Ct_{\QC}(V') = \min_{V \in \VB_d} \Ct_{\QC}(V)\}\,.
    \label{eqnOptimalNoisySet}
\end{align}
We say that $C_{\QC}(V)$ exhibits strong-OPR to $\NC$ if $ \widetilde{\VB}_d^{\opt} = \VB_d^{\opt}$. We say that $C_{\QC}(V)$ exhibits weak-OPR to $\NC$ if $ \widetilde{\VB}_d^{\opt} \subseteq \VB_d^{\opt}$.
\end{definition}

\begin{figure}[!t]
    \centering
    \includegraphics[width=\columnwidth]{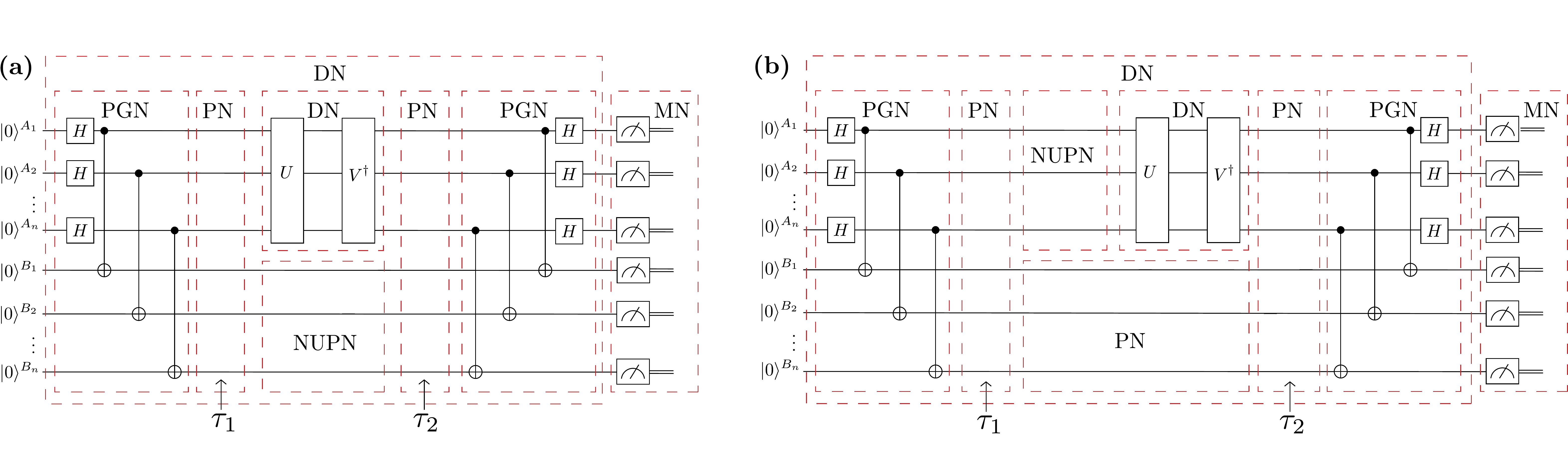}
\caption{Schematic diagram of: (a) Noise Model 1 of Definition \ref{def:HST-noise-model-1}, and (b) Noise Model 2 of Definition \ref{def:HST-noise-model-2}. The following acronyms are employed: Depolarizing Noise (DN), Pauli Gate Noise (PGN), Pauli Noise (PN), Non-Unital Pauli Noise (NUPN), and Measurement Noise (MN). Red dashed boxes indicate the time period and the qubits on which the noise process acts. Time $\tau_1$ ($\tau_2$) corresponds to the time immediately before (after) the action of the unitary $V\ad U$. While both panels show the HST, these noise models are also applicable to the LHST, provided one replaces $E\ad$ with $(E\jj)\ad$.}
    \label{fig:noisy-HST-circuit}
\end{figure}

\subsection{Noise Resilience of Full Unitary Matrix Compiling}

Let us begin with Full Unitary Matrix Compiling (FUMC). Figure~\ref{fig:noisy-HST-circuit} shows the two noise models that we will consider for FUMC. As shown in this figure, $\tau_1$ and $\tau_2$ are respectively defined as the times just before and just after the application of $V\ad U$. We note that the noise models considered in Fig.~\ref{fig:noisy-HST-circuit} capture fairly well the physical noise that is present in, e.g., superconducting-qubit quantum computers, with the exception that only depolarizing noise is allowed during the action of $V\ad U$. We make this simplification for ease of analysis, although our numerics in Sec.~\ref{secImplementations} relax this assumption.

Consider the following definition for the noise model depicted in Fig.~\ref{fig:noisy-HST-circuit}(a).

\begin{definition}\label{def:HST-noise-model-1}
We define  Noise Model 1 to be the following noise process during the HST circuit: (1) global depolarizing noise acting continuously throughout the circuit, (2) global Pauli noise at times $\tau_1$ and $\tau_2$, (3) global depolarizing noise on system A acting continuously in between $\tau_1$ and $\tau_2$, (4) global non-unital Pauli noise  on system $B$ acting continuously in between $\tau_1$ and $\tau_2$, (5) Pauli gate noise during $E$ and $E\ad$, and (6) measurement noise. We also use the term  Noise Model 1 when the same noise model acts during the LHST circuit, provided one replaces $E\ad$ with $(E\jj)\ad$.
\end{definition}

We now state our first main result. The proof of this result is given in Appendix~\ref{sec:proof-thm-1}, with some useful preliminaries and lemmas given in Appendices~\ref{sec:prem}--\ref{sec:meas-noise-FUMC}.

\begin{theorem}
\label{thm1}
The cost functions $C_{\HST}$ and $C_{\LHST}$ exhibit strong-OPR to  Noise Model 1 in Definition \ref{def:HST-noise-model-1}.
\end{theorem}

Note that this theorem also implies that $C(q) = q C_{\HST} + (1-q) C_{\LHST}$ exhibits strong-OPR to Noise Model 1, for all values of $q$. This is because the set $\VB_d^{\opt} = \widetilde{\VB}_d^{\opt}$ defined in \eqref{eqnOptimalSet} and \eqref{eqnOptimalNoisySet} is the same for $C_{\HST}$ and $C_{\LHST}$ functions. Hence this same set is optimal for $C(q)$.

Consider the implications of Theorem~\ref{thm1}. First, this theorem implies that FUMC is resilient to the measurement noise model in Definition~\ref{def-noisy-POVM}. Second, FUMC is completely resilient to Pauli gate noise during the entangling and disentangling gates, $E$ and $E\ad$. Note that this Pauli gate noise is global and hence accounts for cross talk. Third, FUMC is resilient to global depolarizing noise acting continuously throughout the circuit, as well as global Pauli noise acting at the specific times $\tau_1$ and $\tau_2$. Fourth, FUMC is resilient to depolarizing noise acting on system $A$ and non-unital Pauli noise acting on system $B$, provided that each of these process act (possibly continuously) during the time interval between $\tau_1$ and $\tau_2$. We emphasize that Pauli noise includes dephasing channels ($T_2$ noise) as a special case, while non-unital Pauli noise includes the amplitude damping channel ($T_1$ noise) as a special case. Importantly, Theorem~\ref{thm1} states that FUMC is resilient to the general case where all of these noise processes occur together.

We now state our second main result 
(proven in Appendix~\ref{sec:proof-thm-2}),
which deals with the noise model in Fig.~\ref{fig:noisy-HST-circuit}(b).

\begin{definition}\label{def:HST-noise-model-2}
We define Noise Model 2 to be the following noise process during the HST circuit: (1) global depolarizing noise acting continuously throughout the circuit, (2) global Pauli noise at times $\tau_1$ and $\tau_2$, (3) 
global non-unital Pauli noise on system $A$ at time $\tau_1$, (4) global depolarizing noise on system $A$ acting continuously in between $\tau_1$ and $\tau_2$, (5) global Pauli noise on system $B$ acting continuously in between $\tau_1$ and $\tau_2$, (6) Pauli gate noise during $E$ and $E\ad$, and (7) measurement noise. We also use the term  Noise Model 2 when the same noise model acts during the LHST circuit, provided one replaces $E\ad$ with $(E\jj)\ad$.
\end{definition}

\begin{theorem}
\label{thm2}
The cost functions $C_{\HST}$ and $C_{\LHST}$ exhibit strong-OPR to Noise Model 2 in Definition \ref{def:HST-noise-model-2}.
\end{theorem}

The implications of Theorem~\ref{thm2} are similar to those of Theorem~\ref{thm1}. The main difference is that Theorem~\ref{thm2} allows for non-unital Pauli noise on system $A$ at time $\tau_1$, at the expense of only allowing Pauli noise to act continuously on system $B$ between $\tau_1$ and $\tau_2$. The other aspects of the noise models treated by these two theorems are identical.

The above two theorems immediately imply several corollaries below. These corollaries establish resilience to noise models that are different and in some cases more general than the noise models previously considered, at the expense of possibly specializing the form of the unitary $W = V\ad U$. See Appendix~\ref{sec:proof-corol} for the proofs of all corollaries.

\begin{corollary}
\label{cor:Pauli-conjugation}
The cost functions $C_{\HST}$ and $C_{\LHST}$ exhibit strong-OPR to a noise model that includes the following: (1) all noise processes in  Noise Model 1, as well as
(2) a noise process during the implementation of $\mathcal{W} = \mathcal{W}_k \circ \cdots \circ \mathcal{W}_1 = \mathcal{V}^{\dagger}\circ\mathcal{U}$ (i.e., in the time interval between $\tau_1$ and $\tau_2$) in which global Pauli channels $\{\mathcal{P}^A_1, \dots, \mathcal{P}^A_k\}$ act on system $A$, such that the overall channel on $A$ is $\mathcal{P}^A_k\circ \mathcal{W}_k  \cdots \circ\mathcal{P}^A_1\circ \mathcal{W}_1$, provided that the following condition is satisfied:
\begin{align}\label{eq:pauli-channel-conjugation}
(\mathcal{P}^A_k\circ \mathcal{W}_k  \cdots \circ\mathcal{P}^A_1\circ \mathcal{W}_1)(\cdot) = (\mathcal{W}_k \circ \mathcal{W}_{k-1}\cdots \circ \mathcal{W}_1\circ \widehat{\mathcal{P}}^{A}) (\cdot)~.
\end{align}
Here $\widehat{\mathcal{P}}^{A}$ is also a Pauli channel, and the channels $\UC$, $\VC\ad$, and $\WC$ correspond to conjugating the state by the unitaries $U$, $V\ad$, and $W$, respectively.
\end{corollary}

The condition in \eqref{eq:pauli-channel-conjugation} implies that the overall channel consisting of global Pauli channels acting on system $A$ during the implementation of $\mathcal{W}$ is mathematically equivalent (although physically inequivalent) to a Pauli channel followed by $\mathcal{W}$. Therefore, Corollary \ref{cor:Pauli-conjugation} follows from Theorem \ref{thm1}. 

Consider the following implications of Corollary \ref{cor:Pauli-conjugation}. Unitaries corresponding to the Clifford group necessarily satisfy the condition in \eqref{eq:pauli-channel-conjugation}, as shown in Appendix~\ref{sec:prem}. Therefore, Corollary~\ref{cor:Clifford} below holds for any Clifford unitary $W$. Moreover, tensor-product unitaries satisfy this same condition provided that the noise is local depolarizing noise, and hence Corollary~\ref{cor:TensorProduct} below also follows from Corollary \ref{cor:Pauli-conjugation}.

\begin{corollary}
\label{cor:Clifford}
Let the $W=V\ad U$ gate sequence have the form $W = W^A_2 W^A_1$ with $W_1^A$ composed only of Clifford gates. Then the cost functions $C_{\HST}$ and $C_{\LHST}$ exhibit strong-OPR to a noise model that includes the following: (1) all noise processes in Noise Model~1, as well as  
(2) a noise process during the implementation of $\mathcal{W}^A_1 = \mathcal{W}_{1,k} \circ \cdots \circ \mathcal{W}_{1,1}$, in which global Pauli channels $\{\mathcal{P}^A_1, \dots, \mathcal{P}^A_k\}$ act on system $A$, such that the overall channel on $A$ is $\mathcal{P}^A_k\circ \mathcal{W}_{1,k}  \cdots \circ\mathcal{P}^A_1\circ \mathcal{W}_{1,1}$. 
\end{corollary}

\begin{corollary}
\label{cor:TensorProduct}
Let the $W=V\ad U$ gate sequence have the form $W = W_2^A W_1^A$ with $W_1^A =  W^{A'}_1 \otimes W^{A''}_1$ being a tensor product, i.e., $W$ is a tensor product up to a particular time. Then the cost functions $C_{\HST}$ and $C_{\LHST}$ exhibit strong-OPR to a noise model that includes the following: (1) all noise processes in Noise Model 1, as well as (2) a noise process during the implementations of $\WC_1^{A'} = \WC_{1,k}^{A'} \circ \cdots \circ \WC_{1,1}^{A'}$ and $\WC_1^{A''} = \WC_{1,l}^{A''} \circ \cdots \circ \WC_{1,1}^{A''}$ in which local depolarizing channels $\{\DC^{A'}_{1,1}, \dots ,\DC^{A'}_{1,k}\}$ and  $\{\DC^{A''}_{1,1}, \dots ,\DC^{A''}_{1,l}\}$ act on subsystems 
$A'$ and $A''$, respectively, such that the overall channel on $A=A'A''$ is $(\DC^{A'}_{1,k}\circ \WC^{A'}_{1,k}\dots \DC^{A'}_{1,1}\circ \WC^{A'}_{1,1})\otimes(\DC^{A''}_{1,l}\circ \WC^{A''}_{1,l}\dots \DC^{A''}_{1,1}\circ \WC^{A''}_{1,1})$. 
\end{corollary}

The following corollary follows from Theorem~\ref{thm2} and is analogous to Corollary~\ref{cor:Pauli-conjugation}.

\begin{corollary}
\label{cor:non-unital-Pauli-conjugation}
The cost functions $C_{\HST}$ and $C_{\LHST}$ exhibit strong-OPR to the following noise model: (1) all noise processes in  Noise Model 2, as well as 
(2) a noise process during the implementation of $\mathcal{W} = \mathcal{W}_k \circ \cdots \circ \mathcal{W}_1 = \mathcal{V}^{\dagger}\circ\mathcal{U}$ (i.e., in the time interval between $\tau_1$ and $\tau_2$) in which global non-unital Pauli channels $\{\mathcal{P}^A_{\text{NU},1}, \dots, \mathcal{P}^A_{\text{NU},k}\}$ act on system $A$  such that the overall channel on $A$ is $\mathcal{P}^A_{\text{NU},k}\circ \mathcal{W}_k  \cdots \circ\mathcal{P}^A_{\text{NU},1}\circ \mathcal{W}_1$, provided that the following condition is satisfied:
\begin{align}\label{eq:non-unital-pauli-channel-conjugation2}
(\mathcal{P}^A_{\text{NU},k}\circ \mathcal{W}_k  \cdots \circ\mathcal{P}^A_{\text{NU},1}\circ \mathcal{W}_1)(\cdot) = (\mathcal{W}_k \circ \mathcal{W}_{k-1}\cdots \mathcal{W}_1\circ \widehat{\mathcal{P}}^A_{\text{NU}}) (\cdot)~,
\end{align}
where $\widehat{\mathcal{P}}^A_{\text{NU}}$ is also a non-unital Pauli channel.
\end{corollary}

Finally, we present a simple corollary of Theorem~\ref{thm1} based on the ricochet property of the standard Bell state. Note that the noise model in the following corollary is fairly simple but nonetheless physically distinct from those considered in Fig.~\ref{fig:noisy-HST-circuit}, since it allows for global non-unital Pauli noise to occur during the implementation of $W$.

\begin{corollary}\label{cor:nupn-bw-tau1-2}
The cost functions $C_{\HST}$ exhibits strong-OPR to the following noise model: (1) global depolarizing noise acting continuously throughout the circuit, (2) global non-unital Pauli noise on system $A$ at a fixed time in between $\tau_1$ and $\tau_2$. 
\end{corollary}

\subsection{Noise Resilience of Fixed Input State Compiling}

Let us now consider Fixed Input State Compiling (FISC). Recall that the cost-evaluation circuits, shown in Fig.~\ref{fig:SE-LLE-circuit}, have less structure than the circuits in Fig.~\ref{fig:HST-LHST-circuit}. As a result, the noise model that we consider in the FISC case is simpler than the previously considered noise models. In particular, we define the following noise model, which is depicted in 
Fig.~\ref{fig:noisy-SE-LLE-circuit}. Note that, in this context, $\tau_1$ is defined as the time just before the application of $V\ad U$, and there is no need to consider a noisy quantum channel occurring after $V\ad U$ since the measurement occurs immediately after $V\ad U$.

\begin{definition}
\label{defFISCnoise}
We define Noise Model 3 to be the following noise process during the LET or the LLET: (1) global depolarizing noise acting continuously throughout the circuit, (2) global Pauli noise acting at time $\tau_1$, and (3) measurement noise. \end{definition}

We now state our main result for FISC, which is proven in Appendix~\ref{sec:proof-thm-3}.

\begin{figure}[t]
    \centering
    \includegraphics[width=0.75\columnwidth]{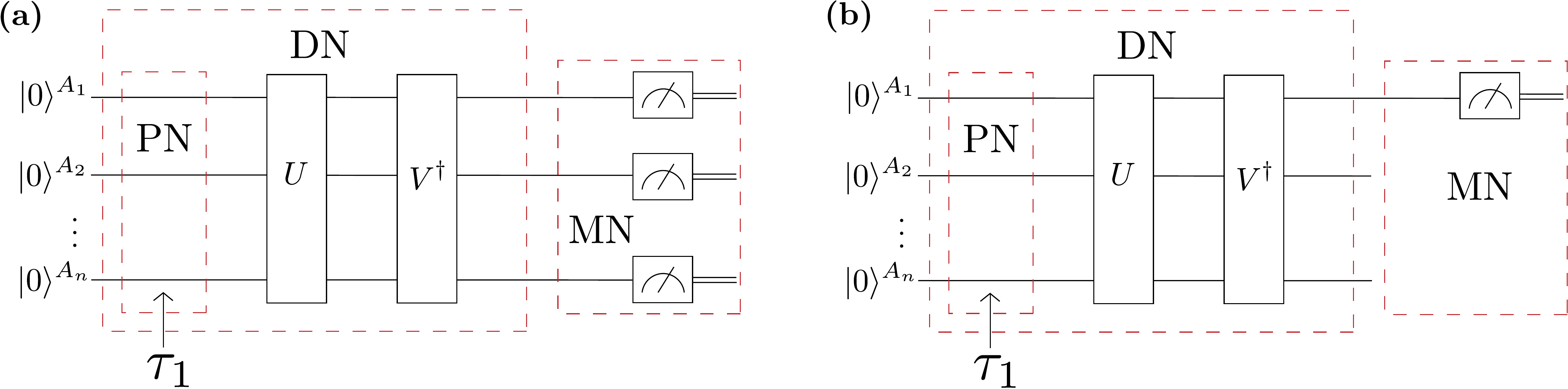}
    \caption{Schematic diagram of Noise Model 3 of Definition~\ref{defFISCnoise} for: (a) the LET circuit, and (b) the LLET circuit. Global depolarizing noise (DN) acts continuously throughout the circuit, global Pauli noise (PN) acts at time $\tau_1$, and measurement noise (MN) occurs during readout. }
    \label{fig:noisy-SE-LLE-circuit}
\end{figure}

\begin{theorem}
\label{thm3}
The cost functions $C_{\LET}$ and $C_{\LLET}$ exhibit weak-OPR, as defined in Definition~\ref{def:OPR}, to Noise Model 3 in Definition \ref{defFISCnoise}.
\end{theorem}

\noindent This theorem implies that FISC is resilient to the measurement noise model in Definition~\ref{def-noisy-POVM}. Furthermore, it is resilient to Pauli noise acting at $\tau_1$ and global depolarizing noise acting continuously throughout the circuit.

We remark that while FUMC exhibits strong-OPR for the noise models considered (see the previous section), here FISC exhibits weak-OPR instead. The latter arises from the fact that the optimal set of unitaries $\VB_d^{\opt}$ for FISC can be highly degenerate (i.e., can contain many unitaries) and the presence of noise could in general break such degeneracy. The ``weak'' term in weak-OPR is simply the fact that the number of global optima is possibly reduced by noise, not that the noise resilience itself is weak.
 Hence, weak-OPR should still be viewed as noise resilience, since the global optima in the presence of noise correspond to global optima in the noiseless case. This implies that training in the presence of noise will lead one to find the correct optimal parameters for $V(\vec{\alpha})$. 

Under certain conditions, Theorem~\ref{thm3} implies that $C'(q)$ defined in \eqref{eqnCprimeQ} will also exhibit weak-OPR to Noise Model 3. Let $\VB_{d,\LET}^{\opt}$ and  $\VB_{d,\LLET}^{\opt}$ denote the sets of unitaries that optimize $C_{\LET}$ and  $C_{\LLET}$, respectively. In the absence of noise we have $\VB_{d,\LET}^{\opt}=\VB_{d,\LLET}^{\opt}$, while in the presence of noise, Theorem \ref{thm3} implies $\widetilde{\VB}_{d,\LET}^{\opt}\subseteq\VB_{d,\LET}^{\opt}$ and $\widetilde{\VB}_{d,\LLET}^{\opt}\subseteq\VB_{d,\LLET}^{\opt}$. Hence, if $\widetilde{\VB}_{d,\LET}^{\opt}\cap\widetilde{\VB}_{d,\LLET}^{\opt}\neq \emptyset$, then for any value of $q$, $C'(q) = q C_{\LET} + (1-q) C_{\LLET}$ will also exhibit weak-OPR to Noise Model 3, where the unitaries that optimize $C'(q)$ in the noisy case belong to $\widetilde{\VB}_{d,\LET}^{\opt}\cap\widetilde{\VB}_{d,\LLET}^{\opt}$.

Theorem \ref{thm3} implies the following corollaries, which establish resilience to noise models that go beyond Noise Model~3 at the expense of specializing the form of $W$. Note that these corollaries are analogous to Corollaries~\ref{cor:Pauli-conjugation}--\ref{cor:TensorProduct}, and Corollary~\ref{cor:Pauli-conjugation-LET} implies Corollaries~\ref{cor:Clifford-FISC} and \ref{cor:TensorProduct-FISC}. See Appendix~\ref{sec:proof-corol} for the proofs.

\begin{corollary}
\label{cor:Pauli-conjugation-LET}
The cost functions $C_{\LET}$ and $C_{\LLET}$ exhibit weak-OPR to a noise model that includes the following: (1) all noise processes in  Noise Model 3, as well as
(2) a noise process during the implementation of $\WC = \mathcal{W}_k \circ \cdots \circ \mathcal{W}_1 = \mathcal{V}^{\dagger}\circ\mathcal{U}$ in which global Pauli channels $\{\mathcal{P}_1, \dots, \mathcal{P}_k\}$ act, such that the overall channel is $\mathcal{P}_k\circ \mathcal{W}_k  \cdots \circ\mathcal{P}_1\circ \mathcal{W}_1$, provided that the following condition is satisfied:
\begin{align}\label{eq:pauli-channel-conjugation-LET}
(\mathcal{P}_k\circ \mathcal{W}_k  \cdots \circ\mathcal{P}_1\circ \mathcal{W}_1)(\cdot) = (\mathcal{W}_k \circ \mathcal{W}_{k-1}\cdots \circ \mathcal{W}_1\circ \widehat{\mathcal{P}}) (\cdot)~,
\end{align}
where $\widehat{\mathcal{P}}$ is also a Pauli channel.
\end{corollary}

\begin{corollary}
\label{cor:Clifford-FISC}
Let the $W=V\ad U$ gate sequence have the form $W = W^A_2 W^A_1$ with $W_1^A$ composed only of Clifford gates. Then the cost functions $C_{\LET}$ and $C_{\LLET}$ exhibit weak-OPR to a noise model that includes the following: (1) all noise processes in  Noise Model 3, as well as  
(2) a noise process during the implementation of $\mathcal{W}^A_1 = \mathcal{W}_{1,k} \circ \cdots \circ \mathcal{W}_{1,1}$, in which global Pauli channels $\{\mathcal{P}^A_1, \dots, \mathcal{P}^A_k\}$ act on system $A$, such that the overall channel on $A$ is $\mathcal{P}^A_k\circ \mathcal{W}_{1,k}  \cdots \circ\mathcal{P}^A_1\circ \mathcal{W}_{1,1}$. 
\end{corollary}

\begin{corollary}
\label{cor:TensorProduct-FISC}
Let the $W=V\ad U$ gate sequence have the form $W = W_2^A W_1^A$ with $W_1^A =  W^{A'}_1 \otimes W^{A''}_1$ being a tensor product, i.e., $W$ is a tensor product up to a particular time. Then the cost functions $C_{\LET}$ and $C_{\LLET}$ exhibit weak-OPR to a noise model that includes the following: (1) all noise processes in Noise Model 3, as well as (2) a noise process during the implementations of $\WC_1^{A'} = \WC_{1,k}^{A'} \circ \cdots \circ \WC_{1,1}^{A'}$ and $\WC_1^{A''} = \WC_{1,l}^{A''} \circ \cdots \circ \WC_{1,1}^{A''}$ in which local depolarizing channels $\{\DC^{A'}_{1,1}, \dots ,\DC^{A'}_{1,k}\}$ and  $\{\DC^{A''}_{1,1}, \dots ,\DC^{A''}_{1,l}\}$ act on subsystems 
$A'$ and $A''$, respectively, such that the overall channel on $A=A'A''$ is $(\DC^{A'}_{1,k}\circ \WC^{A'}_{1,k}\dots \DC^{A'}_{1,1}\circ \WC^{A'}_{1,1})\otimes(\DC^{A''}_{1,l}\circ \WC^{A''}_{1,l}\dots \DC^{A''}_{1,1}\circ \WC^{A''}_{1,1})$. 
\end{corollary}

\section{Implementations}\label{secImplementations}

In this section, we present the results of implementing VQC on the following three-qubit unitaries:
the Toffoli gate, the three-qubit Quantum Fourier Transform (QFT), and a W-state preparation circuit. Each of these unitaries is of interest, e.g., the Toffoli gate when combined with the Hadamard gate provides a universal gate set for quantum computing~\cite{shi2003both}, the QFT is a subroutine in Shor's algorithm~\cite{shor1997factoring}, and W-state preparation is useful for the Quantum Approximate Optimization Algorithm~\cite{wang2019xy, farhi2014QAOA}. Figure \ref{fig:Circuits} shows gate sequences corresponding to these unitaries obtained from the literature. The Toffoli gate in Fig.~\ref{fig:Circuits}(a) is decomposed into a gate sequence that contains nine one-qubit gates and six CNOTs \cite{shende2009cnot}. For the QFT we employ its textbook circuit~\cite{nielsen2010} in Fig.~\ref{fig:Circuits}(b), while the circuit for W-state preparation in Fig.~\ref{fig:Circuits}(c) was derived from Refs.~\cite{bartschi2019deterministic,cruz2019efficient}.

Our VQC implementations were performed using IBM's noisy quantum simulator \cite{cross2017open} with  a noise model built from the reported noise parameters and connectivity of IBM's 14-qubit  Melbourne quantum computer \cite{gadi_aleksandrowicz_2019_2562111}. We remark that for VQC, we must have a target unitary $U$ that is written as a gate sequence in the native gate language and the native connectivity of the hardware. IBM's simulator for the Melbourne device has a square lattice connectivity and native gate alphabet of CNOTs, arbitrary rotation around $Z$ and $\pi/2$ rotation around X. Hence, transforming the gate sequences in Fig.~\ref{fig:Circuits} for the native device will typically add an overhead of additional gates. Therefore, the target gate sequences in our implementations actually correspond to IBM's compilation (with this overhead included) of the circuits in Fig.~\ref{fig:Circuits}. 

In IBM's noise model~\cite{cross2017open,qiskit}, one-qubit gate errors are modeled as a single-qubit depolarizing error followed by a thermal relaxation error, where thermal relaxation refers to both $T_1$ and $T_2$ channels. Similarly, two-qubit gate errors consist of a two-qubit depolarizing error followed by single-qubit thermal relaxation errors on each qubit. Finally, the noise model includes single-qubit readout errors.

\begin{figure}[t]
    \centering
    \includegraphics[width=\columnwidth]{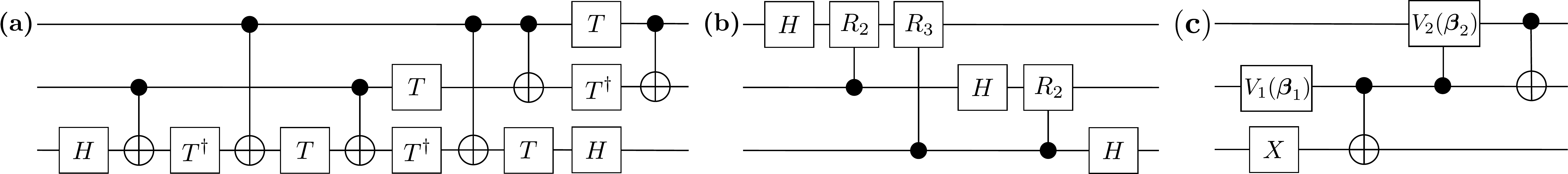}
    \caption{Quantum circuits for: (a) Toffoli Gate, (b) Three-qubit Quantum Fourier Transform, and (c) Three-qubit W-state preparation. Here,  $R_m$ stands for the  controlled phase gate with a phase shift of $\phi=e^{2\pi i/2^m}$, and $V_k(\vec{\beta}_k)$ is given by \eqref{eq-U3}. For the three-qubit W-state preparation circuit we have $\vec{\beta}_1=(2\arccos(\sqrt{1/3}),0,0)$ and $\vec{\beta}_2=(\pi/2,0,0)$.}
    \label{fig:Circuits}
\end{figure}

We employ two different ansatzes, shown in Fig.~\ref{fig:DCNOT}, and (as described below) we employ gradient-based optimization algorithms to train the gate sequence $V(\vec{\alpha})$. In Figs.~\ref{fig:Toffoli}--\ref{fig:Wstate}, we plot the results of implementing VQC with IBM's noisy simulator for the three-qubit gates in Fig.~\ref{fig:Circuits}. In each plot, we show the value of the noisy cost functions versus the number of iterations of the optimization algorithm. 
Additionally, we plot the corresponding value of the noiseless cost functions evaluated for the variational parameters $\vec{\alpha}$ obtained from the noisy optimization. These results allow us to verify if the parameters obtained from the noisy optimization are indeed minimizing the noiseless cost functions. Before discussing the results, we first give details for our ansatzes and optimization methods.

\subsection{Ansatzes and optimization methods}

As previously mentioned, to implement VQC we consider two ansatzes for the trainable unitary $V(\vec{\alpha})$. The building block of our ansatzes is a dressed CNOT gate, which is a two-qubit gate composed of a CNOT preceded and followed by single-qubit gates $V_{k}(\vec{\alpha}_k)$ acting on each qubit, as shown in Fig.~\ref{fig:DCNOT}(a). Each single-qubit gate  $V_{k}(\vec{\alpha}_k)$ is decomposed (up to a global phase) into three elementary rotations parameterized by three angles in the vector $\vec{\alpha}_k=(\alpha_{k,1},\alpha_{k,2},\alpha_{k,3})$~as 
\begin{equation}
    V_k(\vec{\alpha}_k)=e^{- i \alpha_{k,3} \sigma_y/2}e^{- i \alpha_{k,2} \sigma_z/2}e^{- i \alpha_{k,1} \sigma_{y}/2}\,.\label{eq-U3}
\end{equation}

Let us now introduce our ansatzes. We note that our two ansatzes are fairly similar to the ones introduced in Ref.~\cite{Khatri2019quantumassisted}. In our first ansatz, each layer is composed of $n$ dressed CNOTs, where $n$ is the number of qubits (in the special case of $n=2$ each layer consists of one dressed CNOT), with the precise structure defined as follows.

\begin{definition}\label{def:apa}
We define the alternating-pair ansatz as a layered ansatz in which each layer consists of (parameterized) dressed CNOT gates acting on alternating pairs of neighboring qubits as illustrated in Fig.~\ref{fig:DCNOT}(b). 
\end{definition}

We remark that it is useful to distinguish between a \textit{complete ansatz}, in which an exact compilation for $U$ is contained inside the ansatz, versus an \textit{incomplete ansatz}, where exact compilation is not possible. In general, a small number of layers can lead to an incomplete ansatz, where one can only reach approximate compilation. Hence, increasing the number of layers $l$ could allow one to obtain better compilations of $U$. Note however that while a large number of layers can achieve a complete ansatz, it can also be harder to train and can lead to a longer-depth circuit.

The alternating-pair ansatz may not lead to the optimal depth compilation for $U$, particularly in the complete ansatz case. Our second ansatz attempts to fix the issue of introducing unnecessary depth by having a structure that depends on $U$. 
\begin{definition}\label{def:uia}
We construct the target-inspired ansatz by taking the gate sequence for the target unitary $U$, expanding this gate sequence into single-qubit gates and CNOTs, removing all single-qubit gates that precede or follow a CNOT, and replacing each remaining CNOT in the gate sequence with a (parameterized) dressed CNOT. Finally, each remaining single-qubit gate is replaced by a parametrized single-qubit gate.
\end{definition}
\noindent As schematically depicted in Fig.~\ref{fig:DCNOT}(c), each layer is now composed of one dressed CNOT. This ansatz will always be complete since its structure is inspired by $U$. While this ansatz is not useful to compress the number of CNOTs in $V(\vec{\alpha})$, it is useful as a proof-of-concept to demonstrate OPR for complete ansatzes. We remark that a simple modification of this ansatz, where the placements of the dressed CNOTs are optimized over instead of fixed, would actually be useful for circuit-depth compression. Furthermore, we have implemented this dressed CNOT placement optimization, and we find 
that we obtain similar noise resilience results as those for the target-inspired ansatz.

\begin{figure}
    \centering
    \includegraphics[width=.8\columnwidth]{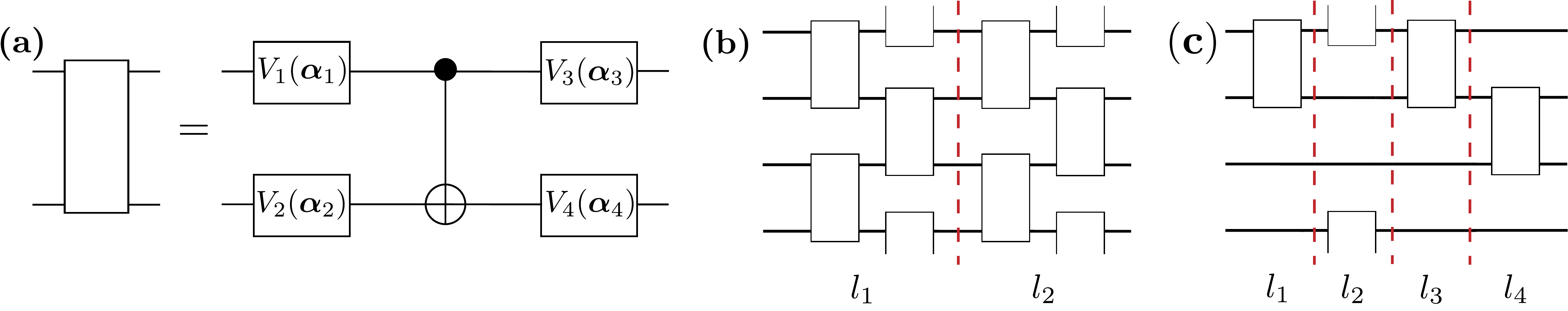}
    \caption{(a) The dressed CNOT  is composed of a CNOT preceded and followed by single-qubit gates $V_k(\vec{\alpha}_k)$, where $V_k(\vec{\alpha}_k)$ is given by~\eqref{eq-U3}. (b) Two layers of the alternating-pair ansatz in the case of four qubits. Each layer is composed of dressed CNOTs  acting on alternating pairs of neighboring qubits. (c) Schematic representation of the target-inspired ansatz. In this approach, the gate sequence of dressed CNOTs is obtained from  the gate sequence of the target unitary $U$.  }
    \label{fig:DCNOT}
\end{figure}

Let us now discuss the optimization methods.  As previously mentioned, the trainable gate sequence $V(\vec{\alpha})$ is a function of a set of parameters $\vec{\alpha}$ corresponding to the collection of the internal gate angles in each dressed CNOT. To optimize these parameters, we employ a gradient-descent approach. This approach exploits the fact that the gradient with respect to $\vec{\alpha}$ of  $C_\HST$, $C_\LHST$, $C_\LET$, and $C_\LLET$ can be computed by using the circuits for HST, LHST, LET, and LLET, respectively~\cite{mitarai2018quantum,Khatri2019quantumassisted}. We remark that we used different gradient-based approaches for the shallow and deep ansatz cases, since the latter requires a more sophisticated and efficient optimizer. 

Specifically, for the shallow ansatz cases where there are few parameters, we employ the simple gradient-based approach outlined in Ref.~\cite[Appendix 4]{Khatri2019quantumassisted}. In this approach, the number of shots $N$ per iteration is fixed. (We choose $N=50000$.) On the other hand, for deep ansatzes with larger numbers of parameters, we employ a more sophisticated gradient-based approach that improves efficiency by reducing the number of shots required~\cite{kubler2019adaptive}. This approach is the individual Coupled Adaptive Number of Shots (iCANS) algorithm of Ref.~\cite{kubler2019adaptive}, which is a
measurement-frugal method that often outperforms other optimizers in the presence of noise. The iCANS optimizer frugally adjusts the number of shots both for a given iteration and for a given partial derivative in a stochastic gradient descent. When employing iCANS, one sets as input: (1) the total number of shots employed during the optimization, and (2) the minimum number of shots (denoted $N_{\text{min}}$) employed to estimate the gradient for a given iteration. We set the latter to initially be $N_{\text{min}}=2$ and then later increase this to $N_{\text{min}}=250$, which empirically leads to good convergence.

\begin{figure}[t]
    \centering
    \includegraphics[width=\columnwidth]{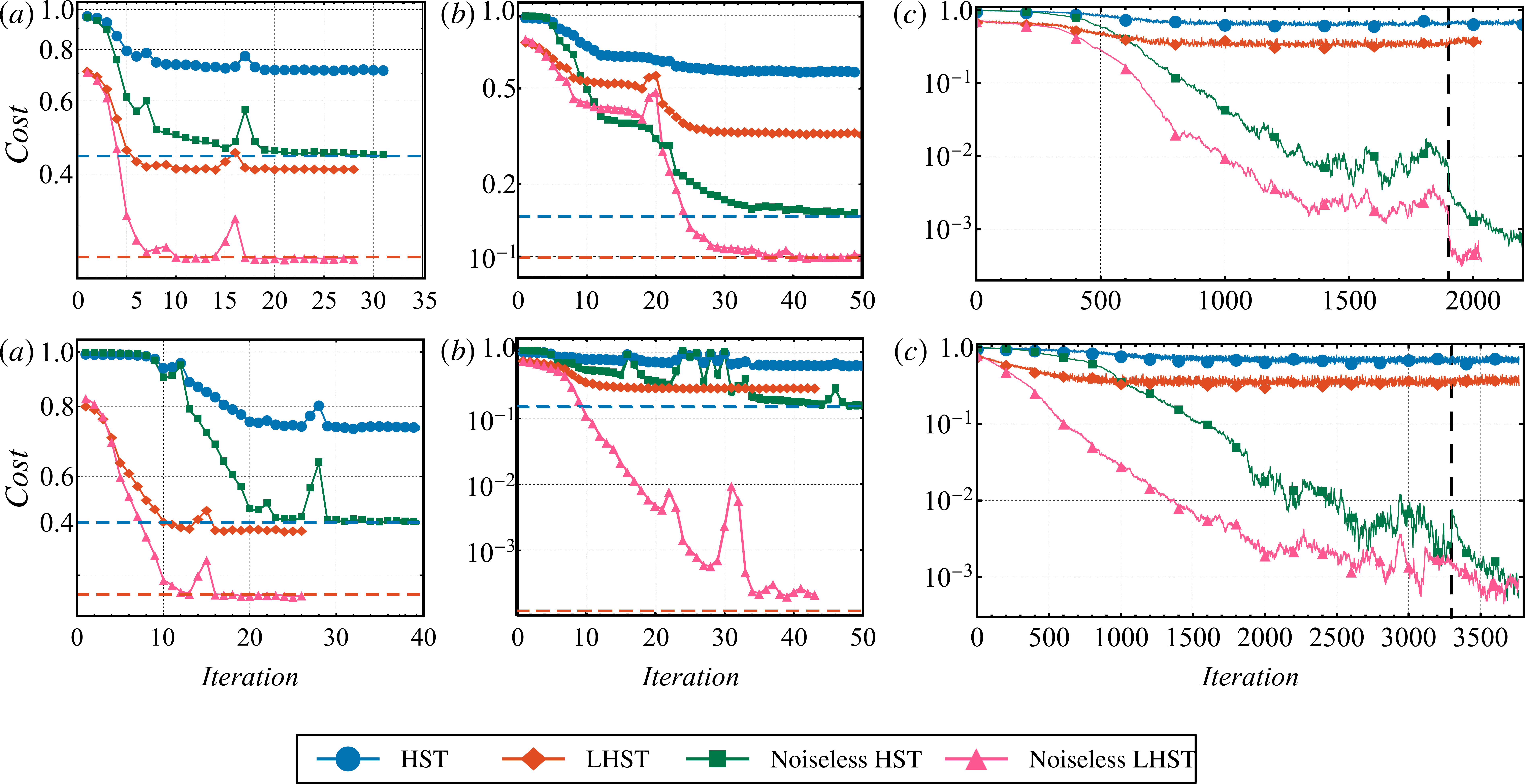}
    \caption{VQC implementations for the Toffoli gate (top) and three-qubit QFT (bottom). The ansatz for $V(\vec{\alpha})$ is: (a) one layer of the alternating-pair ansatz, (b) two layers of the alternating-pair ansatz, (c) the target-inspired ansatz.  The blue and green curves respectively plot the values of $\Ct_{\HST}$ and $\Ct_{\LHST}$ obtained by training $V(\vec{\alpha})$ in the presence of noise. The green and pink curves respectively plot the values of $C_{\HST}$ and $C_{\LHST}$ evaluated at the variational parameters $\vec{\alpha}$ obtained from the noisy optimization of $V(\vec{\alpha})$. 
    Curves are plotted as a function of the number of iterations in the gradient-descent algorithm, and the $y$-axis is in log-scale. The blue and red dashed lines in (a) and (b) correspond to the minimum value of $C_{\HST}$ and $C_{\LHST}$, respectively, determined by optimizing $V(\vec{\alpha})$ in a noise-free environment. Top: in both (a) and (b), the green and pink curves converge to the dashed blue and red lines, respectively. Bottom: While in (a) the green and pink curves converge to the dashed lines, in (b) the termination condition for the optimization algorithm was reached before the pink curve could achieve convergence.  The number of shots per iteration was $N=50000$ for (a) and (b). For (c)  we employed the iCANS optimizer~\cite{kubler2019adaptive}, where the total number of shots was $1.4\times 10^{7}$ and the minimum number of shots per iteration was initially $N_{\text{min}}=2$. The thick dashed vertical line in (c) indicates the point where we set $N_{\text{min}}=250$, which helped to further reduce the cost function.}
    \label{fig:Toffoli}
\end{figure}

\subsection{Toffoli gate}\label{sec:Toffoli-numerics}

The top panels in Figure~\ref{fig:Toffoli} show results of implementing VQC for the Toffoli gate.  Figure~\ref{fig:Toffoli}(top, a) corresponds to $V(\vec{\alpha})$ being given by a single layer of the alternating-pair ansatz of Definition \ref{def:apa}. Here, the noisy cost functions $\widetilde{C}_{\HST}$ and $\widetilde{C}_{\LHST}$ (blue and red curve, respectively) tend to decrease as the number of iterations increases and converge to non-zero values. We remark that the number of iterations can be different for $\widetilde{C}_{\HST}$ and $\widetilde{C}_{\LHST}$ since the termination condition of the optimization algorithm can be reached for a different number of iterations.

Figure~\ref{fig:Toffoli}(top, a) also depicts the cost functions $C_{\HST}$ and $C_{\LHST}$ evaluated for the variational parameters $\vec{\alpha}$ obtained from the noisy optimization (green and pink curve, respectively). These curves show that as the number of iterations increases, both $C_{\HST}$ and $C_{\LHST}$ tend to decrease too, indicating that the noisy training is {\it indirectly training the noiseless cost functions}, i.e., the adjustments to the parameters $\vec{\alpha}$ made by noisy training are reducing the noiseless cost functions. Note that $C_{\HST}$ and $C_{\LHST}$ do not converge to zero since a single layer of three dressed CNOTs forms an incomplete ansatz for the Toffoli gate. 

In order to determine if the algorithm is reaching the minimum value achievable with just one layer, we have also implemented VQC to compile the Toffoli gate in a noise-free simulation. The minimum values achieved for $C_{\HST}$ and $C_{\LHST}$ are shown as a blue and red dashed curve, respectively. Surprisingly, the cost functions evaluated with the parameters from the noisy training (green and pink curves) converge to the dashed lines. This suggests that the optimal parameters are noise resilient since noisy training reaches the minimum value obtained by noise-free training. As a caveat, however, we note that it is not clear whether the minima reached are global or local optima.  

Figure \ref{fig:Toffoli}(top, b) plots the VQC results for Toffoli with $V(\vec{\alpha})$ given by two layers of the alternating-pair ansatz. In this case,  $C_{\HST}$ and $C_{\LHST}$ converge to values which are smaller than the ones obtained in the one-layer case. The latter indicates that two layers allow for a more complete compilation of the Toffoli gate, albeit it appears that the ansatz is not yet complete. Note that both the decomposition of the Toffoli gate in Figure~\ref{fig:Circuits}, as well as two layers of the alternating-pair ansatz, consist of six CNOTs. However, the placement of the dressed CNOTs does not seem to be optimal. Finally, let us remark that the green and pink curves converge to the dashed blue and red lines, respectively. Hence, this once again shows that the optimal parameters are noise resilient. Similar to the previous case, it is not clear whether the minima reached are global or local minima.

Figure \ref{fig:Toffoli}(top, c) shows results for the target-inspired ansatz of Definition \ref{def:uia}. As the number of iterations increases, all curves tend to decrease,  with the green and pink curves converging to values   of the order of $10^{-4}$.  We remark that we have verified  that  $W=V\ad U\approx \id$ for the parameters obtained. In this case, we do not plot dashed blue and red curves since the ansatz is complete and the minimum of the noiseless cost functions is zero.

These results indicate that optimizing $V(\vec{\alpha})$ in the presence of noise yields the correct variational parameters $\vec{\alpha}$, which minimize the noiseless cost function.  Hence, both $C_{\HST}$ and $C_{\LHST}$ appear to exhibit OPR for the realistic noise model considered.

\begin{figure}
    \centering
    \includegraphics[width=.9\columnwidth]{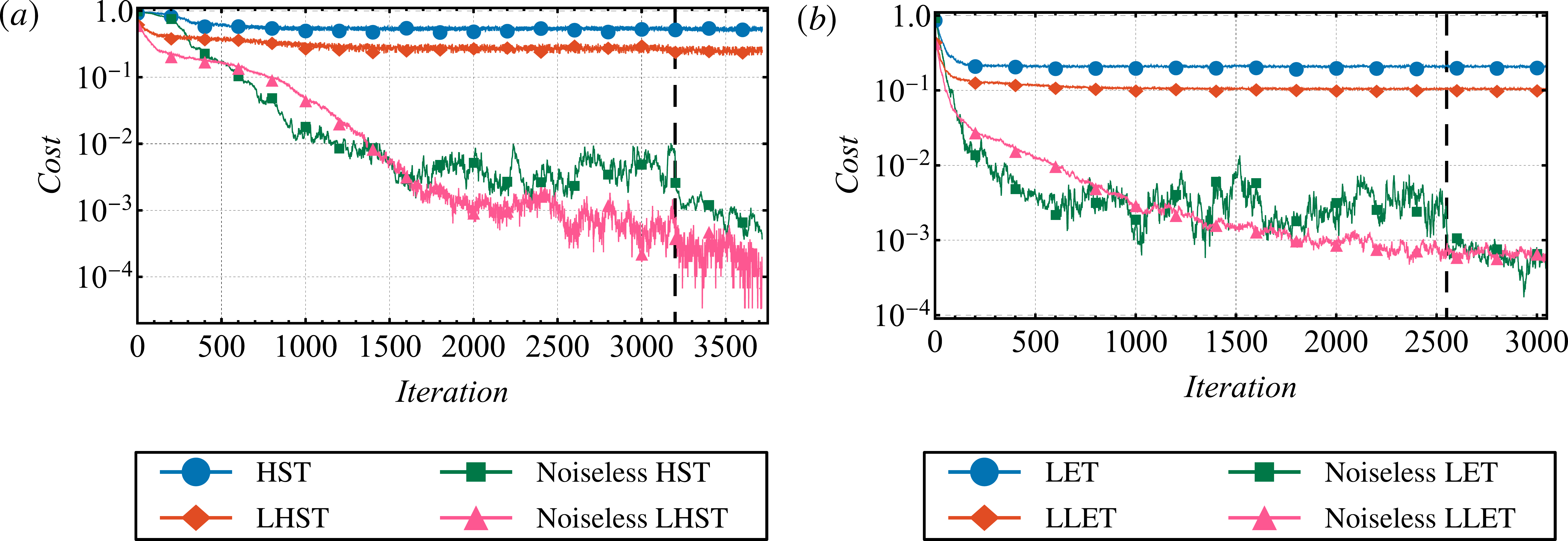}
    \caption{VQC implementations for the three-qubit W-state preparation circuit  for (a) the FUMC approach, and (b) the FISC approach. The trainable gate sequence $V(\vec{\alpha})$ is given by the target-inspired ansatz. In the left (right) panel the blue and green curves plot respectively the values of $\Ct_{\HST}$ ($\Ct_{\LET}$) and $\Ct_{\LHST}$ ($\Ct_{\LLET}$) obtained by noisy training of $V(\vec{\alpha})$. Similarly, in the left (right) panel the green and pink curves give respectively the values of $C_{\HST}$ ($\Ct_{\LET}$) and $C_{\LHST}$ ($\Ct_{\LLET}$) evaluated at the variational parameters $\vec{\alpha}$ obtained from the noisy optimization of $V(\vec{\alpha})$.  Curves are plotted as a function of the number of gradient-descent iterations, with the $y$-axis in log-scale. Via noisy training, the noiseless cost functions go down to $\sim 10^{-4}$. Initially we set $N_{\text{min}}=2$, and the thick dashed vertical lines shows the point where we increased this value to $N_{\text{min}}=250$. Increasing the minimum number of shots iCANS employs to compute each partial derivative leads to smaller cost function values in both cases.} 
    \label{fig:Wstate}
\end{figure}

\subsection{Quantum Fourier Transform}

We now discuss the VQC results for the three-qubit QFT. Figure~\ref{fig:Toffoli} shows the results for $V(\vec{\alpha})$ consisting of:  a single layer of the alternating-pair ansatz of Definition \ref{def:apa} (bottom, a), two layers of the alternating-pair ansatz (bottom, b), and the target-inspired ansatz of Definition \ref{def:uia} (bottom, c). As shown in these plots, most of the results for QFT are similar to the results for the Toffoli gate. In all cases the noiseless cost functions tended to decrease with iterations, indicating that noisy training indirectly trains the noiseless costs. 

For the one-layer case of Figure~\ref{fig:Toffoli} (bottom, a) the green and pink curves (noiseless cost functions evaluated at the parameters obtained from noisy training) converge to the value obtained by training in a noise-free environment (dashed curve). Here, the non-zero value of the dashed curve indicates that a one-layer ansatz is incomplete. 
This is in contrast to Figure~\ref{fig:Toffoli} (bottom, b), where the dashed red line of $C_{\LHST}$ is of the order of $10^{-4}$, implying that the ansatz is complete. Once again, in Figure~\ref{fig:Toffoli} (bottom, b), the green and pink curves approximately converge to the dashed lines (noiseless training), indicating noise resilience.  
Finally, Figure \ref{fig:Toffoli} (bottom, c), shows that that both $C_{\HST}$ and $C_{\LHST}$ appear to  exhibit OPR, as we can indirectly train the parameters in $V(\vec{\alpha})$ in the presence of noise.

\subsection{W-state preparation}

Finally, we discuss the results of implementing of VQC for both FUMC and FISC of a W-state preparation circuit. We remark here that we did not perform  FISC for the Toffoli gate and the QFT since those unitaries act trivially on the $\ket{\vec{0}}$ state. Moreover, we are only interested in comparing the FUMC and the FISC approach with a complete ansatz, meaning that we only considered the target-inspired ansatz of Definition \ref{def:uia}.

As shown in Fig.~\ref{fig:Wstate}, all cost functions $C_{\HST}$, $C_{\LHST}$, $C_{\LET}$, and $C_{\LLET}$ can be optimized indirectly via noisy training of $V(\vec{\alpha})$.  Both for FUMC and FISC the cost functions go down to $\sim 10^{-4}$, while for FUMC one can even reach values of $\sim 10^{-5}$ when employing the LHST.  Hence, our numerics indicate that $C_{\HST}$, $C_{\LHST}$, $C_{\LET}$, and $C_{\LLET}$ appear to exhibit OPR to IBM's realistic noise model.

\section{Discussion}

\subsection{VQC in the NISQ era}\label{secNISQera}

Our analytical and numerical results suggest that Variational Quantum Compiling (VQC) could be a useful tool for near-term noisy quantum computing. While there are several intended uses for VQC~\cite{Khatri2019quantumassisted}, the main purpose is for circuit-depth compression of quantum algorithms. This depth compression arises because VQC could achieve optimal compiling, whereas classical methods for quantum compiling either scale exponentially (if they are aiming at optimal compiling) or are sub-optimal when they are restricted to local (instead of global) compiling of the circuit.

Suppose one is able to achieve depth compression with VQC. This implies that the target unitary $U$ has a longer depth than the trained gate sequence $V(\vec{\alpha})$. Prior to our work, one may have been concerned that this depth compression might not reduce noise, because perhaps the noise occurring during $U$ is somehow compiled into the gate sequence $V(\vec{\alpha})$. However, our work shows that this is not the case. Despite various sources of incoherent noise (e.g., see the noise model in Fig.~\ref{fig:noisy-HST-circuit}), we find that one learns the correct optimal parameters $\vec{\alpha}$ for $V(\vec{\alpha})$. This means that, after performing VQC, if one was to implement the gate sequence $V(\vec{\alpha})$ instead of $U$, then one should see that $V(\vec{\alpha})$ really does achieve less noise than $U$, since the depth of $V(\vec{\alpha})$ is shorter.

\subsection{Summary of results}

In this work, we treated two different forms of VQC: Full Unitary Matrix Compiling (FUMC) and Fixed Input State Compiling (FISC). Our main analytical results were stated in Theorems~\ref{thm1}--\ref{thm3}. We found that both FUMC and FISC are resilient to measurement noise. In addition, they are both resilient to global depolarizing noise acting continuously throughout the circuit and global Pauli noise occurring just prior to the implementation of $W = V\ad U$. 

For FUMC, we were able to prove resilience to additional sources of noise, such as Pauli gate noise during the entangling and disentangling gates as well as non-unital Pauli noise occurring at particular times in the circuit. The fact that our noise resilience results are more extensive for FUMC than for FISC may simply be due to the fact that the cost-evaluation circuit for FUMC is more complicated than that for FISC. Hence it is possible that this additional resilience is needed to make the two approaches have similar levels of noise resilience. Alternatively, it could be possible that either FUMC or FISC is more noise resilient than the other, although this remains to be established. (Note that our numerics did not see a significant difference in the noise resilience of FUMC versus FISC.)

In addition, Corollaries~\ref{cor:Pauli-conjugation}--\ref{cor:TensorProduct-FISC} stated resilience results for noise models that go beyond the noise models considered in Theorems~\ref{thm1}--\ref{thm3}, at the expense of possibly specializing the form of the unitary $W = V\ad U$ (for example, to Clifford unitaries or tensor-product unitaries). In particular, these corollaries considered noise that occurs during the implementation of $W$, which is certainly practically relevant.

Our numerical results were presented in Figs.~\ref{fig:Toffoli}--\ref{fig:Wstate}. Generally speaking, these numerics agreed with our theoretical expectations and hinted at resilience beyond what is stated in our theorems, which we discuss in the next subsection. We emphasize that our implementations employed the noise model of IBM's 14-qubit Melbourne device, and hence this shows that VQC exhibits resilience for currently available hardware.

\subsection{Noise resilience beyond our theorems}

There are two senses in which VQC might exhibit resilience beyond the results stated in our theorems. The first sense is that VQC may be resilient to more general noise models than the ones we considered. The second sense is that VQC may be resilient even for the incomplete ansatz case, on which we elaborate below. Both of these possibilities appear to be supported by our numerical implementations. 

For evidence supporting the idea that VQC may be resilient to more general noise models, consider the following. The noise model associated with IBM's 14-qubit Melbourne device is more general than the noise models depicted in Figs.~\ref{fig:noisy-HST-circuit} and \ref{fig:noisy-SE-LLE-circuit}, and the unitaries we considered in Fig.~\ref{fig:Circuits} do not fall into the special cases (e.g., Clifford or tensor product) treated by Corollaries~\ref{cor:Pauli-conjugation}--\ref{cor:TensorProduct-FISC}. For example, IBM's noise model has non-unital Pauli noise associated with each gate and hence occurring throughout the implementation of $W=V\ad U$. Thus, our theorems and corollaries do not cover all of noise processes occurring in IBM's noise model. Despite this, we were able to reduce the noiseless cost (via noisy training) to  $\sim 10^{-4}$ for the Toffoli gate (Fig.~\ref{fig:Toffoli}(top, c)) and QFT (Fig.~\ref{fig:Toffoli}(bottom, c)), and to  $\sim 10^{-5}$ for W state preparation (Fig.~\ref{fig:Wstate}).

Naturally, our theorems and corollaries have a bias towards noise models that are mathematically easy to work with, such as Pauli noise or depolarizing noise, since this makes it easier to formulate proofs. It is therefore important for future work to attempt to show resilience beyond these noise models. 

As noted above, VQC may also have resilience beyond the complete ansatz case. Recall that we say an ansatz for $V(\vec{\alpha})$ is complete (incomplete) if it contains (does not contain) an exact compilation of $U$. Our theorems and corollaries are restricted to the complete ansatz case, whereas our numerics in Fig.~\ref{fig:Toffoli}  also consider the incomplete ansatz case. Interestingly, Fig.~\ref{fig:Toffoli}  showed that typically one can obtain the same value for the noiseless cost with either noisy or noiseless training. This surprising result suggests that perhaps the optimal values for $\vec{\alpha}$ may be resilient to noise even for the incomplete ansatz case, and future work should investigate this possibility. 

In addition, it will be important to investigate the effect of noise on the parameter landscape and parameter trainability~(e.g.,~\cite{gentini2019noise}). Our work indicates that the global optimum of VQC may not change with noise, but does not address the difficulty of finding this optimum.

\subsection{Coherent versus incoherent noise}

In the Introduction, we emphasized the distinction between OPR and Cost Value Resilience~\cite{mcclean2016theory}. The latter is relevant to coherent noise, whereas OPR is relevant to incoherent noise. Intuitively, we anticipate that coherent noise (e.g., systematic gate biases) in VQC will often shift the location of the global minimum in parameter space, and hence we expect coherent noise to have a non-trivial effect on the optimal parameters in VQC. Because of this intuition, we have focused our paper and our definition of OPR solely on incoherent noise. We remark that our definition of OPR, which is stated in terms of unitaries (rather than parameters), would need to be modified if one is interested in studying parameter resilience for coherent noise. However, as noted, we do not anticipate resilience to coherent noise to hold. We also remark that other strategies exist to correct coherent noise~\cite{DCG_khodjasteh2009dynamically}. Nevertheless, an interesting question for future work will be see whether OPR holds partially whenever both coherent and incoherent noise are present. In addition, it will be interesting to combine the ideas of OPR and Cost Value Resilience into a single framework.

\subsection{Noise resilience of VQE}

Finally, let us consider VHQCAs more generally. In particular, let us revisit the Variational Quantum Eigensolver (VQE) that we discussed in Sec.~\ref{sec:VQE}. As we now show, VQC is a special case of VQE. This idea was noted for FISC in Ref.~\cite{jones2018quantum}. However, the argument is more subtle for the FUMC case.

The key observation is that the various cost functions can be rewritten as the expectation values for some effective Hamiltonians:
\begin{align}
    C_{\LET}& =  \mte{\psi(\vec{\alpha})}{H_{\LET}} \,,\quad     C_{\LLET} =  \mte{\psi(\vec{\alpha})}{H_{\LLET}}\,,\notag \\
    C_{\HST} &= \mte{\chi(\vec{\alpha})}{H_{\HST}}\,,\quad     C_{\LHST} =  \mte{\chi(\vec{\alpha})}{H_{\LHST}}\,.
    \label{eqn:HamiltonianCosts}
\end{align}
Here $\ket{\psi(\vec{\alpha})} \in \HC^{A}$ and $\ket{\chi(\vec{\alpha})} \in \HC^{AB}$ are $n$-qubit and $2n$-qubit states, respectively, given by
\begin{align}
\ket{\psi(\vec{\alpha})} = V(\vec{\alpha})\ket{\vec{0}},\qquad \ket{\chi(\vec{\alpha})} = ( V(\vec{\alpha}) \otimes \id^B)\ket{\Phi}\,,
\end{align}
where $\HC^{X}$ denotes the Hilbert space of system $X$, and $\ket{\Phi} = E \ket{\vec{0}}$ is the standard maximally entangled state on $AB$. We remark that $\ket{\chi(\vec{\alpha})}$ is simply the Choi state associated with $V(\vec{\alpha})$. 

For the cost functions associated with FISC, the effective Hamiltonians are given by
\begin{align}
    H_{\LET}& =  \id^A - U\dya{\vec{0}}U\ad \,,\qquad    H_{\LLET} =  \id^A - \frac{1}{n}\sum_{j=1}^n U  (P_{0}^{A_j}\otimes\id^{\overline{A}_j}) U\ad\,,
        \label{eqn:HamiltonianFISC}
\end{align}
where $P_{0}^{A_j}$ is the projector onto the zero state of $A_j$. For the cost functions associated with FUMC, the effective Hamiltonians are given by
\begin{align}
    H_{\HST} &= \id^{AB} - (U\otimes \id^B)\dya{\Phi}(U\ad \otimes \id^B)\,,\quad     H_{\LHST} =  \id^{AB} - \frac{1}{n}\sum_{j=1}^n (U\otimes \id^B)(\dya{\Phi\jj}\otimes \id^{\overline{A}_j\overline{B}_j})(U\ad \otimes \id^B)\,, 
        \label{eqn:HamiltonianFUMC}
\end{align}
where $\ket{\Phi\jj}$ is the standard maximally entangled state on $A_jB_j$. With these Hamiltonians, one can verify that the expressions in \eqref{eqn:HamiltonianCosts} are equal to the original cost function definitions in Sec.~\ref{sec:background}. Hence, we have just shown that VQC is a special case of VQE, where the goal is to prepare the ground state of one of the Hamiltonians in \eqref{eqn:HamiltonianFISC} or \eqref{eqn:HamiltonianFUMC}.

The fact that VQC is a special case of VQE implies that, for specific Hamiltonians, VQE is noise resilient. Namely, we have shown that VQE exhibits OPR when the Hamiltonian has the form in either \eqref{eqn:HamiltonianFISC} or \eqref{eqn:HamiltonianFUMC}. This naturally points to the question of whether VQE is resilient more generally. It is therefore a very interesting direction for future research to extend our noise resilience to Hamiltonians other than the ones we considered.

\section{Conclusions}

In this work, we discovered a novel kind of noise resilience for Variational Hybrid Quantum-Classical Algorithms (VHQCAs). We introduced the idea of Optimal Parameter Resilience (OPR), where the variational parameters corresponding to the global optimum are unaffected by various types of incoherent noise. We showed that Variational Quantum Compiling (VQC) exhibits OPR. This paves the way for VQC to be used in the era of noisy intermediate-scale quantum computing as a tool for circuit-depth compression. Important future research directions include: (1) Extending our theorems to show resilience to more general noise models than the ones we considered (which our numerics suggest may be possible), (2) Exploring noise resilience for the incomplete ansatz case (which our numerics indicate may also be resilient), (3) Analyzing approximate noise resilience, (4) Studying the effect of noise on the parameter training process, and (5) Generalizing our resilience results to other Hamiltonians for the Variational Quantum Eigensolver and exploring resilience for other VHQCAs (for example, some evidence of noise resilience was recently reported in Ref.~\cite{bravo2019variational}).

\section{Acknowledgements}

We thank Lukasz Cincio and Mark M. Wilde for helpful discussions. KS acknowledges support from the U.S. Department of Energy (DOE) through a quantum computing program sponsored by the LANL Information Science \& Technology Institute. SK acknowledges support from the National Science Foundation and the National Science and Engineering Research Council of Canada Postgraduate Scholarship. MC was supported by the Center for Nonlinear Studies at Los Alamos National Laboratory (LANL). PJC acknowledges support from the LANL ASC Beyond Moore's Law project. MC and PJC also acknowledge support from the LDRD program at LANL. This work was also supported by the U.S. DOE, Office of Science, Office of Advanced Scientific Computing Research.

\bibliography{ref.bib}

\appendix

\section{Preliminaries}\label{sec:prem}

The main goal of the Appendix is to provide the proofs of Theorems~\ref{thm1}--\ref{thm3} and Corollaries~\ref{cor:Pauli-conjugation}--\ref{cor:TensorProduct-FISC}. For these proofs, we will need to first review some definitions and properties. We point readers to \cite{nielsen2010, wilde2017} for additional background.

\bigskip

\noindent \textbf{Pauli Basis.} In our proofs, we will work in the Pauli product basis, involving a tensor product of one-qubit Pauli operators. This is a natural basis to choose, given the qubit structure of quantum computers. Let
\begin{equation}
X^{\vec{l}}\coloneqq\sigma_x^{l_1}\otimes\sigma_x^{l_2}\otimes\dotsb\otimes\sigma_x^{l_n}\,,\qquad Z^{\vec{k}}\coloneqq\sigma_z^{k_1}\otimes\sigma_z^{k_2}\otimes\dotsb\otimes\sigma_z^{k_n}\,,
\end{equation}
where $l_1,l_2,\dotsc,l_n\in\{0,1\}$, $k_1,k_2,\dotsc,k_n\in\{0,1\}$, $\vec{l}= (l_1, \dots, l_n)$, and $\vec{k} = (k_1, \dots k_n)$. The following properties are satisfied by the Pauli operators: 
\begin{equation}\label{eq-pauli_prop}
X^{\vec{l}_1}X^{\vec{l}_2}=X^{\vec{l}_1\oplus\vec{l}_2},\quad Z^{\vec{k}_1}Z^{\vec{k}_2}=Z^{\vec{k}_1\oplus\vec{k}_2},\quad X^{\vec{l}}Z^{\vec{k}}=(-1)^{\vec{l}\cdot\vec{k}}Z^{\vec{k}}X^{\vec{l}},\quad \Tr[X^{\vec{l}}Z^{\vec{k}}]=2^n\delta_{\vec{l},\vec{0}},\delta_{\vec{k},\vec{0}},
\end{equation}
which follow from the properties of the single-qubit Pauli operators.

\noindent \textbf{Pauli group}. The Pauli group of $n$ qubits is 
$\mathbb{G}_n \coloneqq \{\pm 1, \pm i\} \times \{I, \sigma_x, \sigma_y, \sigma_z \}^{\otimes n}$.

\noindent \textbf{Clifford group}. The Clifford group on $n$ qubits is the set of unitaries that normalize the Pauli group, i.e., 
\begin{align}\label{def:clifford-group}
\mathbb{C}_n \coloneqq \{ U : U \mathbb{G}_n U^{\dagger} \in \mathbb{G}_n \}~.
\end{align}

\noindent \textbf{Maximally entangled states.}
In what follows, we consider the following maximally entangled states
$\ket{\Phi^+}\bra{\Phi^+} = \ket{\phi^+}\bra{\phi^+}^{\otimes n}$,
where $\ket{\phi^+} = (\ket{0,0}+\ket{1,1})/\sqrt{2}$. The aforementioned tensor product of maximally entangled states can be written in the Pauli basis as follows:
\begin{equation}\label{eq-Phi+}
\dya{\Phi^+}_{AB} = \frac{1}{2^{2n}}\sum_{\vec{l},\vec{k}} X^{\vec{l}}_AZ^{\vec{k}}_A\otimes X^{\vec{l}}_B Z^{\vec{k}}_B=\frac{1}{2^{2n}}\sum_{\vec{l},\vec{k}}Z^{\vec{k}}_AX^{\vec{l}}_A\otimes Z^{\vec{k}}_BX^{\vec{l}}_B.
\end{equation}

\noindent \textbf{All-zero state.} Noting that $\dya{0}=(\id+\sigma_z)/2$, then in the Pauli basis the all-zero state $\dya{\vec{0}}=\dya{0}^{\otimes n}$ is
\begin{equation}\label{eq-allzero-Pauli}
    \dya{\vec{0}} = \frac{1}{2^{n}}(\id+\sigma_z)^{\otimes n}=\frac{1}{2^{n}}\sum_{\vec{l}}Z^{\vec{l}}.
\end{equation}

\noindent \textbf{Pauli channels}. 
A Pauli noise channel corresponds to the action of random Pauli operators on a quantum state $\rho$ according to a probability distribution. Let $\mathcal{P}^A$ denote an $n$-qubit Pauli channel acting on system $A = A_1, \dots A_n$. Then the action of $\mathcal{P}^A$ on the state $\rho$ is given by 
\begin{equation}\label{eq:Pauli-channel}
\mathcal{P}^A(\rho)=\sum_{\vec{l},\vec{k}}p^A_{\vec{l},\vec{k}}X_A^{\vec{l}}Z_A^{\vec{k}}\rho(X_A^{\vec{l}}Z_A^{\vec{k}})^\dagger,
\end{equation}
where $0\leq p^A_{\vec{l},\vec{k}}\leq 1$, and $\sum_{\vec{l},\vec{k}}p^A_{\vec{l},\vec{k}}=1$. Using the properties in \eqref{eq-pauli_prop}, we find that 
\begin{align}
\mathcal{P}^A(X_A^{\vec{a}}Z_A^{\vec{b}})=\sum_{\vec{l},\vec{k}}p^A_{\vec{l},\vec{k}}X_A^{\vec{l}}Z_A^{\vec{k}}X_A^{\vec{a}}Z_A^{\vec{b}}Z_A^{\vec{k}}X_A^{\vec{l}}
=\sum_{\vec{l},\vec{k}}(-1)^{\vec{a}\cdot\vec{k}}(-1)^{\vec{b}\cdot\vec{l}} p^A_{\vec{l},\vec{k}}X_A^{\vec{a}}Z_A^{\vec{b}}
=p^A_{\vec{a},\vec{b}}X_A^{\vec{a}}Z_A^{\vec{b}},
\end{align}
where $p^A_{\vec{a},\vec{b}}\coloneqq \sum_{\vec{l},\vec{k}}(-1)^{\vec{a}\cdot\vec{k}}(-1)^{\vec{b}\cdot\vec{l}} p^A_{\vec{l},\vec{k}}$ and $-1\leq p^A_{\vec{a},\vec{b}}\leq 1$ for all $\vec{a},\vec{b}\in\{0,1\}^n$. Similarly, the action of a global Pauli channel $\mathcal{P}^{AB}$ acting on systems  $A=A_1\dotsb A_n$ and $B=B_1\dotsb B_n$, respectively, is defined as
\begin{equation}\label{eq:2n-qubit-Pauli-channel}
\mathcal{P}^{AB}(X_A^{\vec{a}_1}Z_A^{\vec{b}_1}\otimes X_B^{\vec{a}_2}Z_B^{\vec{b}_2})=p_{\vec{a}_1,\vec{a}_2,\vec{b}_1,\vec{b}_2}^{AB}X_A^{\vec{a}_1}Z_A^{\vec{b}_1}\otimes X_B^{\vec{a}_2}Z_B^{\vec{b}_2}.
\end{equation}

\noindent \textbf{Non-unital Pauli noise channels}.
The action of a non-unital Pauli channel $\mathcal{P}_{\text{NU}}$ on an $n$-qubit Pauli operators is
\begin{align}
\mathcal{P}_{\text{NU}}(X^{\vec{a}}Z^{\vec{b}})&=c_{\vec{a},\vec{b}}X^{\vec{a}}Z^{\vec{b}}\quad\forall~\vec{a}\neq\vec{0},\vec{b}\neq\vec{0},\label{def:non-un-Puali1}\\
   \mathcal{P}_{\text{NU}}(X^{\vec{0}}Z^{\vec{0}})&=\mathcal{P}_{\text{NU}}(\id)=\id+ \sum_{(\vec{a},\vec{b})\neq (\vec{0}, \vec{0})}d_{\vec{a},\vec{b}}X^{\vec{a}}Z^{\vec{b}}.\label{def:non-un-Puali2}
\end{align}

We now prove the following lemma based on Clifford unitaries and Pauli channels. 

\begin{lemma}\label{lem:clifford-conjugation}
Let $W$ be a Clifford unitary and let $\PC$ be a Pauli channel. Then for any state $\rho$, the following holds:
\begin{align}
    (\mathcal{W}\circ \PC)(\rho) = (\mathcal{Q}\circ \mathcal{W})(\rho)~,
\end{align}
where $\mathcal{Q}$ is another Pauli channel.
\end{lemma}
\begin{proof}
From \eqref{eq:Pauli-channel} it follows that
\begin{align}
\mathcal{W}\circ \PC(\rho) & = W\left(\sum_{\vec{l}, \vec{k}} p_{\vec{l}, \vec{k}} X^{\vec{l}}Z^{\vec{k}}\rho Z^{\vec{k}}X^{\vec{l}}\right)W^{\dagger} = \sum_{\vec{l}, \vec{k}} p_{\vec{l}, \vec{k}} (W X^{\vec{l}} Z^{\vec{k}}W^{\dagger}) (W\rho W^{\dagger}) (WZ^{\vec{k}}X^{\vec{l}}W^{\dagger})\\
& = \sum_{\vec{l}, \vec{k}} p_{\vec{l}, \vec{k}} X^{\vec{m}(\vec{l}, \vec{k})}Z^{\vec{n}(\vec{l}, \vec{k})} W\rho W^{\dagger} Z^{\vec{n}(\vec{l}, \vec{k} )} X^{\vec{m}(\vec{l}, \vec{k} )}\\
& = (\mathcal{Q}\circ \mathcal{W})(\rho)~.
\end{align}
The third equality follows from the definition of a Clifford unitary \eqref{def:clifford-group}, while the last equality follows from $\eqref{eq:Pauli-channel}$.
\end{proof}

\section{Noisy entangling and disentangling gates in FUMC}\label{sec:noisy-encoding}
For the proofs given in Appendices~\ref{sec:proof-thm-1}--\ref{sec:proof-corol}, we will make use of some properties of the noisy versions of entangling $E$ and disentangling $E\ad$ gates that appear in FUMC. Hence, it is helpful to first state these properties in this appendix. Recall that, for Pauli gate noise acting during $E$ or $E\ad$, we assume that global Pauli channels act before and after each Hadamard, as well as before and after each CNOT. This noise model incorporates the case when there could be correlated Pauli noise acting on different qubits during $E$ and $E\ad$. We note that the noisy entangling gate is the same for both the HST and the LHST. 
   
 Let $E = E^{AB}$ denote the ideal entangling gate, which can be split into a tensor product of two qubit entangling gates $E^{A_j B_j}$ as
  \begin{align}\label{eq:EAB-tensor-product}
      E^{AB}=E^{A_1B_1}\otimes E^{A_2B_2}\otimes\dotsb\otimes E^{A_nB_n}=\bigotimes_{j=1}^n E^{A_jB_j}. 
  \end{align}
Moreover, each $E^{A_j B_j}$ consists of a Hadamard gate acting on $A_j$ followed by a CNOT gate acting on both $A_j$ and $B_j$. In the quantum channel notation we write this as
$
    \mathcal{E}^{A_jB_j}=\mathcal{C}_X^{A_jB_j}\circ(\mathcal{H}^{A_j}\otimes\mathcal{I}^{B_j}),
$
where $\mathcal{H}^{A_j}$ are the quantum channels that implement the Hadamard gates and $\mathcal{C}_X^{A_jB_j}$ are the quantum channels that implement the CNOTs. The noisy version of $\mathcal{E}^{AB}$, which we denote by $\widetilde{\mathcal{E}}^{AB}$, is 
\begin{equation}\label{eq:noisy-encoding-channel}
    \widetilde{\mathcal{E}}^{AB}\coloneqq \bigotimes_{j=1}^{n}\mathcal{R}^{AB}_j\circ\mathcal{C}_X^{A_jB_j}\circ\bigotimes_{j=1}^{n}\mathcal{Q}_j^{AB}\circ(\mathcal{H}^{A_j}\otimes\mathcal{I}^{B_j})\circ \mathcal{P}^{AB}_j,
\end{equation}
where $\mathcal{P}^{AB}_j$, $\mathcal{Q}^{AB}_j$, and $\mathcal{R}^{AB}_j$ are $2n$-qubit global Pauli channels for all $i\in \{1, \dots, n\}$, as defined in \eqref{eq:2n-qubit-Pauli-channel}. Since both Hadamard and CNOT gates are Clifford unitaries, by using Lemma \ref{lem:clifford-conjugation} we find that   
\begin{equation}
    \widetilde{\mathcal{E}}^{AB}\coloneqq \mathcal{M}^{AB}\circ\bigotimes_{j=1}^{n}\mathcal{C}_X^{A_jB_j}\circ\bigotimes_{j=1}^{n}(\mathcal{H}^{A_j}\otimes\mathcal{I}^{B_j}),
\end{equation}
where $\mathcal{M}^{AB}$ is another Pauli channel. 

We now apply $\widetilde{\mathcal{E}}^{AB}$ on the all-zeros state $\dya{\vec{0, 0}}^{AB}$. Consider the following chain of equalities:
\begin{align}
\widetilde{\mathcal{E}}^{AB} (\dya{\vec{0,0}}^{AB})  = \widetilde{\mathcal{E}}^{AB} \bigg(\frac{1}{2^{2n}}  \sum_{\vec{a}, \vec{b}} Z_A^{\vec{a}} \otimes Z_B^{\vec{b}}\bigg)
 = \frac{1}{2^{2n}} \sum_{\vec{a}, \vec{b}} m^{AB}_{\vec{a}, \vec{a}, \vec{b}, \vec{b}}  X_A^{\vec{a}}Z^{\vec{b}}_A\otimes X^{\vec{a}}_BZ^{\vec{b}}_B~,\label{eq:global-noisy-E-on-0state}
\end{align}
where we used \eqref{eq-allzero-Pauli}, \eqref{eq:2n-qubit-Pauli-channel}, and the following identities for all  $j \in \{1, \dots n\}$: 
\begin{align}
  (\mathcal{H}^{A_j} \otimes \mathcal{I}^{B_j})(Z^{a_j}_{A_j} \otimes Z_{B_j}^{b_j}) =  X_{A_j}^{a_j}\otimes Z_{B_j}^{b_j} ~,\,\,
    ( \mathcal{C}_X^{A_jB_j})(X_{A_j}^{a_j}\otimes\id_B)=X_{A_j}^{a_j}\otimes X_{B_j}^{a_j}~,\,\,
     ( \mathcal{C}_X^{A_jB_j})(\id_{A_j}\otimes Z_{B_j}^{b_j})&=Z_{A_j}^{b_j}\otimes Z_{B_j}^{b_j}~.
\end{align}

The noisy disentangling channel for the HST is given by the adjoint of the noisy entangling channel, as defined in \eqref{eq:noisy-encoding-channel}. On the other hand, since in the LHST only two qubits $A_jB_j$ are measured for a given run of the experiment, the disentangling channel is applied only on the $A_jB_j$ pair. However, we assume that global Pauli channels act on $2n$ qubits before and after the  Hadamard and CNOT gate. For each $j\in \{1, \dots, n\}$, the disentangling channel is given by the adjoint of the following channel:
\begin{align}
\widetilde{\mathcal{E}}^{'AB}_j &\coloneqq \mathcal{R}^{AB}_j \circ (\mathcal{C}_X^{A_jB_j}\otimes \mathcal{I}^{\overline{A}_j \overline{B}_j})\circ Q^{AB}_j \circ (\mathcal{H}^{A_j}\otimes \mathcal{I}^{B_j}\otimes \mathcal{I}^{\overline{A}_j \overline{B}_j})\circ \mathcal{P}^{AB}_j~,\\
& = \mathcal{M}^{AB}_j\circ(\mathcal{C}_X^{A_jB_j}\otimes \mathcal{I}^{\overline{A}_j \overline{B}_j}) \circ (\mathcal{H}^{A_j}\otimes \mathcal{I}^{B_j}\otimes \mathcal{I}^{\overline{A}_j \overline{B}_j}),\label{def:adj-decoding-channel-lhst}
\end{align}
where $\mathcal{P}^{AB}_j$, $\mathcal{Q}^{AB}_j$,  $\mathcal{R}^{AB}_j$, and $\mathcal{M}^{AB}_j$ are $2n$-qubit global Pauli channels, as defined in \eqref{eq:2n-qubit-Pauli-channel}, and we used Lemma \ref{lem:clifford-conjugation}. We remark that the Pauli channels are defined with a  $j$  subscript in \eqref{def:adj-decoding-channel-lhst} to emphasize that for different runs of the experiment the Pauli channels that act could be different. 

From arguments similar to those used to derive \eqref{eq:global-noisy-E-on-0state}, we find that 
\begin{align}
\widetilde{\mathcal{E}}^{'AB}_j(\dya{0, 0}^{A_jB_j}\otimes \id_{\overline{A}_j \overline{B}_j}) = \frac{1}{2^{2}} \sum_{a_j, b_j = 0}^1 m^{AB}_{a_j,a_j, b_j, b_j} (X^{a_j}_{A_j}Z^{b_j}_{A_j} \otimes X^{a_j}_{B_j} Z^{b_j}_{B_j}\otimes \id_{\overline{A}_j \overline{B}_j})~.
\end{align}

\section{Measurement noise in FUMC}\label{sec:meas-noise-FUMC}

For the proofs given in Appendices~\ref{sec:proof-thm-1}--\ref{sec:proof-corol}, we will make use of some properties of measurement noise in FUMC. Hence, it is helpful to first state these properties in this appendix.

Let $P_{\vec{0}}$ denote the POVM element associated with getting the all-zeros outcome in the noiseless HST, which can be expressed as 
$    P_{\vec{0}}  \coloneqq  \dya{\vec{0}} = \bigotimes_{j=1}^{2n} \dya{0}.$ We consider the measurement noise as follows. For each qubit $j$, where $j\in \{1, \dots, 2n\}$, the ideal projector $\dya{0}$ 
gets replaced by $p^{(j)}_{00} \dya{0} +p^{(j)}_{01}\dya{1}$. Moreover, we assume that for all $j$ the following strict inequality holds:
$
p^{(j)}_{00} > p^{(j)}_{01}~.
$

Let $\widetilde{P}_{\vec{0}}$ denote the noisy POVM element. Then the following equalities hold: 
\begin{align}
    \widetilde{P}_{\vec{0}}
    &=\bigotimes_{j=1}^n\left(p^{A_j}_{00}\ket{0}\bra{0}^{A_j}+p^{A_j}_{01}\ket{1}\bra{1}^{A_j}\right)\otimes\bigotimes_{j=1}^n\left(p^{B_j}_{00}\ket{0}\bra{0}^{B_j}+p^{B_j}_{01}\ket{1}\bra{1}^{B_j}\right)\\
    &=\sum_{\vec{a},\vec{b}}p^A(\vec{a})p^B(\vec{b})\dya{\vec{a},\vec{b}}^{AB}\,,\label{def:meas-noise-HST}
\end{align}
  with $p^A(\vec{a})=(p^{A_1}_{01})^{a_1}\dotsb (p^{A_n}_{01})^{a_n}(p^{A_1}_{00})^{1-a_1}\dotsb(p^{A_n}_{00})^{1-a_n}$ and $p^B(\vec{b})=(p^{B_1}_{01})^{b_1}\dotsb(p^{B_n}_{01})^{b_n}(p^{B_1}_{00})^{1-b_1}\dotsb(p^{B_n}_{00})^{1-b_n}$.
 
 \subsection{Effective noisy measurement operator for the HST} \label{sec:eff-noisy-meas-op-HST}
 In the noiseless HST, the measurement is preceded by the  disentangling unitary $(E^{AB})^{\dagger}$, where $E^{AB}$ is defined in \eqref{eq:EAB-tensor-product}. In the Heisenberg picture, this corresponds to the evolution of the measurement operator with respect to the unitary $E^{AB}$.  We now derive the effective noisy POVM element as the evolution of $\widetilde{P}_{\vec{0}}$ under the noisy entangling channel $\widetilde{\mathcal{E}}^{AB}$ (defined in Section \ref{sec:noisy-encoding}). 
 
Using \eqref{eq-allzero-Pauli}, $\dya{\vec{a}, \vec{b}}^{AB}$ can be expressed as follows: 
\begin{align}
    \dya{\vec{a},\vec{b}}^{AB}
    &=(X_A^{\vec{a}}\otimes X_B^{\vec{b}})\left(\frac{1}{2^{2n}}\sum_{\vec{l},\vec{k}} Z_A^{\vec{l}}\otimes Z_B^{\vec{k}}\right)(X_A^{\vec{a}}\otimes X_B^{\vec{b}})=\frac{1}{2^{2n}}\sum_{\vec{l},\vec{k}}(-1)^{\vec{a}\cdot\vec{l}}(-1)^{\vec{b}\cdot\vec{k}} Z_A^{\vec{l}}\otimes Z_B^{\vec{k}}~,\label{eq:vec-ab-Puali-op}
\end{align}
 where we used the properties of the Pauli operators as defined in \eqref{eq-pauli_prop}. Then, from \eqref{eq:global-noisy-E-on-0state} and  the linearity of quantum channels,  it follows that 
 \begin{align}
 \widetilde{\mathcal{E}}^{AB}(\dya{\vec{a},\vec{b}}^{AB}) &= \frac{1}{2^{2n}}  \sum_{\vec{l}, \vec{k}} m^{AB}_{\vec{l}, \vec{l}, \vec{k}, \vec{k}}  (-1)^{\vec{a}\cdot\vec{l}}(-1)^{\vec{b}\cdot\vec{k}}  X_A^{\vec{l}}Z^{\vec{k}}_A\otimes X^{\vec{l}}_BZ^{\vec{k}}_B~. \label{eq:EN-on-vec-ab}%\\
% & = \frac{1}{2^{2n}} (Z_A^{\vec{a}}X_A^{\vec{b}}\otimes\id_B)\left(\sum_{\vec{l}, \vec{k}}p^{AB}_{\vec{0}, \vec{0}, \vec{l}, \vec{k}} q^{AB}_{\vec{l}, \vec{0}, \vec{0}, \vec{k}} r^{AB}_{\vec{l}, \vec{l}, \vec{k}, \vec{k}} X_A^{\vec{l}}Z^{\vec{k}}_A\otimes X^{\vec{l}}_BZ^{\vec{k}}_B\right)(X_A^{\vec{b}}Z_A^{\vec{a}}\otimes\id_B)\\
 %& =  (Z_A^{\vec{a}}X_A^{\vec{b}}\otimes\id_B) \widetilde{\mathcal{E}}^{AB}(\dya{\vec{0},\vec{0}}) (X_A^{\vec{b}}Z_A^{\vec{a}}\otimes\id_B)~. \label{eq:EN-on-vec-ab}
 \end{align}
Therefore, from \eqref{def:meas-noise-HST} and \eqref{eq:EN-on-vec-ab} it follows that 
\begin{align}
\widetilde{\mathcal{E}}^{AB}(\widetilde{P}_{\vec{0}})% &= \sum_{\vec{a},\vec{b}}p^A(\vec{a})p^B(\vec{b}) (Z_A^{\vec{a}}X_A^{\vec{b}}\otimes\id_B) \widetilde{\mathcal{E}}^{AB}(\dya{\vec{0},\vec{0}}) (X_A^{\vec{b}}Z_A^{\vec{a}}\otimes\id_B)\\
%& = (\widehat{\mathcal{P}}^A\otimes\mathcal{I}_B)(\widetilde{\mathcal{E}}^{AB}(\dya{\vec{0},\vec{0}}_{AB}))\label{eq:EN-on-P0N}\\
%&
= \frac{1}{2^{2n}} \sum_{\vec{a}, \vec{b}} m^{AB}_{\vec{a}, \vec{a}, \vec{b}, \vec{b}}  \widehat{p}^{A}_{\vec{a}, \vec{b}}Z^{\vec{b}}_AX_A^{\vec{a}}\otimes Z^{\vec{b}}_BX^{\vec{a}}_B\,,\label{eq:EN-on-P0N1}
\end{align}
where  $\widehat{p}^{A}_{\vec{a}, \vec{b}} =\sum_{\vec{l}, \vec{k}} (-1)^{\vec{a}\cdot\vec{l}} (-1)^{\vec{b}\cdot \vec{k}} p^A(\vec{l}) p^B(\vec{k})$, and $p^A(\vec{l})$ and $p^B(\vec{k})$ are probability distributions as in \eqref{def:meas-noise-HST}. 
  
 \subsection{Effective noisy measurement operator for the LHST}
 In the LHST, a noisy measurement on two qubits $A_jB_j$ is preceded by the disentangling unitary $(E^{A_j B_j})^{\dagger}$ acting on the same two qubits. Similar to Section \ref{sec:eff-noisy-meas-op-HST}, we now derive the effective POVM element as the evolution of  the operator $Q^{(j)}_{00}$ (defined below) under the adjoint of the noisy disentangling channel, as defined in \eqref{def:adj-decoding-channel-lhst}. The noisy POVM for the qubits $A_jB_j$ is given by 
\begin{align}
\widetilde{Q}^{(j)}_{00} = \sum_{a', b'=0}^1 p^{A_j}(a')p^{B_j}(b') \ket{a', b'}\bra{a',b'
}^{A_j B_j} ,
\end{align}
which follows from \eqref{def:meas-noise-HST}. Moreover, the overall noisy POVM for the LHST is defined as 
\begin{align}
\widetilde{Q}_{00} = \frac{1}{n}\sum_{j=1}^n \widetilde{Q}^{(j)}_{00} \otimes \id_{\overline{A}_j \overline{B}_j}~.
\end{align}
 By using arguments similar to those used in \eqref{eq:vec-ab-Puali-op}, \eqref{eq:EN-on-vec-ab}, and \eqref{eq:EN-on-P0N1}, we find that 
 \begin{align}
 \widetilde{\mathcal{E}}^{'A B}_j(\widetilde{Q}^{(j)}_{00} \otimes \id_{\overline{A}_j \overline{B}_j}) 
 & = \frac{1}{2^2}\sum_{a_j, b_j} m^{AB}_{a_j,a_j,b_j, b_j}\widehat{p}^{A_j}_{a_j, b_j} Z^{b_j}_A X^{a_j}_{A} \otimes Z^{b_j}_{B_j}X^{a_j}_{B_j}\otimes \id_{\overline{A}_j \overline{B}_j} ~,
 \end{align}
 where $ \widetilde{\mathcal{E}}^{'A B}_j$ is given by \eqref{def:adj-decoding-channel-lhst} and  $\widehat{p}^{A_j}_{a_j, b_j} = \sum_{a',b'=0}^1 (-1)^{a_j\cdot a' } (-1)^{b_j \cdot b'} p^{A_j}(a') p^{B_j}(b')$. 

Therefore, the overall effective noisy POVM for the LHST is defined as
\begin{align}\label{def:eff-meas-op-LHST}
\widetilde{\mathcal{E}}^{'AB}(\widetilde{Q}_{00}) 
& =  \frac{1}{2^{2}} \frac{1}{n} \sum_{j=1}^n \sum_{a_j, b_j = 0}^1 m^{AB}_{a_j,a_j,b_j, b_j}  \widehat{p}^{A_j}_{a_j, b_j}Z^{b_j}_A X^{a_j}_{A} \otimes Z^{b_j}_{B_j}X^{a_j}_{B_j} \otimes \id_{\overline{A}_j \overline{B}_j}\,.
\end{align}

\section{Proof of Theorem \ref{thm1}}\label{sec:proof-thm-1}

Before providing a proof of Theorem \ref{thm1}, we prove the following lemma.
\begin{lemma}
\label{lemmaGDN}
Let $C_{\QC}(V)$ be a cost function of $V$ with $V\in \VB_d$, and $\VB_d$ the set of $d\times d$ unitary matrices. Additionally suppose that $C_{\QC}(V)$ can be evaluated using a quantum circuit denoted $\QC$ as follows:
\begin{align}
C_{\QC}(V) \coloneqq \Tr[\Lambda \EC_{V}(\rho)],
\end{align}
where $\rho$ is a quantum state, $\Lambda$ denotes a POVM element and $\EC_{V}$ denotes the noisy unital quantum channel describing the evolution of the state throughout the computation, which depends on the unitary $V$. Then $\Ct_{\QC}(V)$ exhibits strong-OPR to a noise model composed of $\EC_{V}$ and a global depolarizing channels acting continuously throughout the computation.

\end{lemma}

\begin{proof}
Without loss of generality let us  decompose $\EC_{V}$ as $k$ noisy unital quantum channels: $\mathcal{E}_{V}= \mathcal{E}^k_{V}\circ\ldots\circ\mathcal{E}^1_{V}$. 
In the presence of global depolarizing noise acting throughout the computation, the cost function can now be expressed as
\begin{equation}
\widetilde{C}_{\QC}(V)=\Tr\left[\Lambda (\mathcal{D}^{k+1}\circ\mathcal{E}^k_{V}\circ\ldots\circ\mathcal{D}^2\circ\mathcal{E}^1_{V}\circ\mathcal{D}^1)(\rho)\right]\,,
\end{equation} 
where we have interleaved the channels $\mathcal{E}^i_{V}$ with global depolarizing channels $\mathcal{D}^i$. From Definition \ref{def-GDN} and from the fact that $\EC^i_{V}(\id) = \id$, it follows that
\begin{align}
\widetilde{C}_{\QC}(V)&=\Tr\left[\Lambda(\mathcal{D}^{k+1}\circ\mathcal{E}^k_{V}\circ\ldots\circ\mathcal{D}^2\circ\mathcal{E}^1_{V}\circ\mathcal{D}^1)(\rho)\right]= p\Tr\left[\Lambda(\mathcal{E}^k_{V}\circ\ldots\mathcal{E}^2_{V}\circ\mathcal{E}^1_{V})(\rho)\right] + (1-p)\Tr\left[\Lambda \id\right]/2^n\\
&=p C_{\QC}(V) + (1-p)/2^n
\label{eq-GDn-branch2}
\end{align}
where $p=p_{k+1}\ldots p_1$. Let $\VB_d^{\opt}$  denote the sets of unitaries that optimize $C_{\QC}(V)$ i.e.,
\begin{align}
    \mathbb{V}_d^{\opt} &= \{V' \in \VB_d : C_{\QC}(V') = \min_{V \in \VB_d} C_{\QC}(V)\}\,.
\end{align}
Then, from \eqref{eq-GDn-branch2} we have that 
any unitary in $\mathbb{V}_d^{\opt}$ will also optimize $\widetilde{C}_{\QC}(V)$. Hence $\Ct_{\QC}(V)$ exhibits strong-OPR to a noise model composed of $\EC_{V}$ and a global depolarizing channels acting  throughout the computation.

\end{proof}
By means of Lemma \ref{lemmaGDN} we know that if we show that a quantity exhibits OPR to a noise model $\mathcal{N}$ which does not include global depolarizing noise acting continuously throughout the computation, then said quantity will also exhibit OPR if we include global depolarizing noise to $\mathcal{N}$.  

We now provide a proof for Theorem \ref{thm1}.
\begin{theoremApp}
The cost functions $C_{\HST}$ and $C_{\LHST}$ exhibit strong-OPR to  Noise Model 1 in Definition \ref{def:HST-noise-model-1}.
\end{theoremApp}
\begin{proof}
We begin by breaking up the HST circuit into three time intervals. In the first time interval, the noisy entangling channel $\widetilde{\mathcal{E}}^{AB}$ is applied. In the second time interval, the quantum channel  $\mathcal{V}^{\dagger} \circ \mathcal{U}$ implements the unitaries $U$ and $V\ad$. Finally, in the third time interval $(\widetilde{\mathcal{E}}^{AB})^\dagger$ is applied. We assume that the global depolarizing noise occurs on systems $AB$ during all three time intervals and the global depolarizing noise occurs on system $A$ during the implementation of $\VC\ad\circ \UC$.  Moreover, suppose that two different global Pauli channels $\mathcal{Q}^{AB}$ and $\widehat{\mathcal{Q}}^{AB}$ act at times $\tau_1$ and $\tau_2$, respectively, and global non-unital Pauli channels act continuously on system $B$ in between $\tau_1$ and $\tau_2$. 

Let $\rho^{(0)}$ denotes the initial state of the HST circuit and is given by $\rho^{(0)} = \dya{\vec{0}, \vec{0}}^{AB}$. At $\tau_1$ the state is
\begin{align}\label{eq:state-tau1-hst-thm1}
\rho^{(1)} = \mathcal{Q}^{AB}(\mathcal{D}_{p^{(1,k)}}^{AB} \circ \widetilde{\mathcal{E}}^{AB}_k \dots \mathcal{D}_{p^{(1,1)}}^{AB} \circ \widetilde{\mathcal{E}}^{AB}_1 (\rho^{(0)}))~,
\end{align}
where we have broken up the $\tau_1$  into $k$ time increments and $\widetilde{\mathcal{E}}^{AB}_k \circ \dots \widetilde{\mathcal{E}}^{AB}_1$ is the channel that implements the noisy entangling channel $\widetilde{\mathcal{E}}^{AB}$, as defined in \eqref{eq:noisy-encoding-channel}. Moreover, each $\widetilde{\mathcal{E}}^{AB}_i$ is followed by a global depolarizing channel $\mathcal{D}_{p^{(1,i)}}^{AB}$, where $p^{(r,s)}$ denotes the depolarizing probability for the $s$-th time increment of the $r$-th time interval. Then $\rho^{(1)}$ reduces to 
\begin{align}
\rho^{(1)} &=\mathcal{Q}^{AB} \left( \DC_{p^{(1,k)}}^{AB}\circ \widetilde{\EC}^{AB}_{k}... \widetilde{\EC}^{AB}_{2}(p^{(1,1)}\widetilde{\EC}^{AB}_1(\rho^{(0)}) + (1-p^{(1,1)}) \id / 2^{2n}  )\right)\\
    &= p^{(1)}\mathcal{Q}^{AB}\circ\widetilde{\EC}^{AB}(\rho^{(0)}) + (1-p^{(1)}) \id / d \label{eq:state-tau1-hst-thm1-1}  = p^{(1)} [ \frac{1}{2^{2n}} \sum_{\vec{a}, \vec{b}} \beta^{AB}_{\vec{a}, \vec{b}}X_A^{\vec{a}}Z^{\vec{b}}_A\otimes X^{\vec{a}}_BZ^{\vec{b}}_B ] + (1-p^{(1)}) \id / 2^{2n} ~,
\end{align}
where $p^{(1)} = p^{(1,1)} ... p^{(1,k)}$. The second equality follows from Lemma~\ref{lemmaGDN} as $\widetilde{\mathcal{E}}^{AB}$ consists of only unitary  and Pauli channels, and thus each $\widetilde{\mathcal{E}}^{AB}_i$ is a unital channel, where $i\in \{1, \dots, k \}$.
The last equality follows from \eqref{eq:global-noisy-E-on-0state} and \eqref{eq:2n-qubit-Pauli-channel}, where $\beta^{AB}_{\vec{a}, \vec{b}} = m^{AB}_{\vec{a}, \vec{a}, \vec{b}, \vec{b}}  q^{AB}_{\vec{a}, \vec{a}, \vec{b}, \vec{b}}$.

Similarly, the state at $\tau_2$ is given by
\begin{align}\label{eq:rho2-hst-thm1}
\rho^{(2)} = \widehat{\mathcal{Q}}^{AB}( \mathcal{D}^{AB}_{p^{(2,l)}}\circ\DC_{s^{(2,l)}}^A \circ (\mathcal{W}_l \otimes \mathcal{P}^B_{\text{NU},l}) \dots \mathcal{D}^{AB}_{p^{(2, 1)}}\circ \DC_{s^{(2,1)}}^A \circ (\mathcal{W}_1 \otimes \mathcal{P}_{\text{NU},1}^B)(\rho^{(1)}) ). 
\end{align}
We first find the action of the channel $\mathcal{W}_1 \otimes \mathcal{P}_{\text{NU},1}^B$ on $\rho^{(1)}$. Consider that 
\begin{align}
(\mathcal{W}_1 \otimes  \mathcal{P}_{\text{NU},1}^B)(\rho^{(1)}) 
&= \frac{1}{2^{2n}} (\mathcal{W}_1 \otimes \mathcal{P}_{\text{NU},1}^B) \Big[p^{(1)} \Big(\sum_{(\vec{a}, \vec{b}) \neq (\vec{0}, \vec{0}  )} \beta^{AB}_{\vec{a}, \vec{b}} X_A^{\vec{a}}Z^{\vec{b}}_A\otimes X^{\vec{a}}_BZ^{\vec{b}}_B \Big) + \id_{AB}   \Big]\\
& = \frac{1}{2^{2n}} \Big[p^{(1)} \Big(\sum_{(\vec{a}, \vec{b}) \neq (\vec{0}, \vec{0}  )} \beta^{AB}_{\vec{a}, \vec{b}} c^{(1)}_{\vec{a}, \vec{b}}W_1X_A^{\vec{a}}Z^{\vec{b}}_AW_1^{\dagger}\otimes X^{\vec{a}}_BZ^{\vec{b}}_B \Big) + \id_{AB} + \sum_{(\vec{g}, \vec{h}) \neq (\vec{0}, \vec{0})} d^{(1)}_{\vec{g}, \vec{h}} \id_A \otimes X_B^{\vec{g}} Z_B^{\vec{h}} \Big]~,
\end{align}
where we used the definition of a non-unital Pauli channel from \eqref{def:non-un-Puali1} and \eqref{def:non-un-Puali2}. 
We note that the terms that are independent of $W_i$ do not affect the global optima. Therefore, the only relevant term in \eqref{eq:rho2-hst-thm1} is
\begin{align}
\widetilde{\rho}^{(2)} = \frac{p^{(2)}s^{(2)}p^{(1)}}{2^{2n}} \widehat{\mathcal{Q}}^{AB}\left(\sum_{(\vec{a}, \vec{b}) \neq (\vec{0}, \vec{0}) } \beta^{AB}_{\vec{a}, \vec{b}}  \big(\prod_{i =1}^m c^{(i)}_{\vec{a}, \vec{b}}\big)  WX_A^{\vec{a}}Z_A^{\vec{b}}W^{\dagger}\otimes X^{\vec{a}}_BZ^{\vec{b}}_B \right)~,
\end{align}
where $p^{(2)} = p^{(2,1)}\dots p^{(2, l)}$ and $s^{(2)} = s^{(2,1)}\dots s^{(2, l)}$, and where we have used  \eqref{def:non-un-Puali1} and Lemma~\ref{lemmaGDN}. 

Finally, the relevant term after the action of the noisy disentangling channel is 
\begin{align}
\widetilde{\rho}^{(3)} &= \mathcal{D}^{AB}_{p^{(3,m)}} \circ (\widetilde{\mathcal{E}}^{AB}_{m})^{\dagger} \dots \mathcal{D}^{AB}_{p^{(3,1)}} \circ (\widetilde{\mathcal{E}}^{AB}_1)^\dagger (\widetilde{\rho}^{(2)}) = p^{(3)} (\widetilde{\mathcal{E}}^{AB} )^{\dagger}(\widetilde{\rho}^{(2)}) + (1-p^{(3)}) \id/2^{2n}~,
\end{align}
where $p^{(3)} = p^{(3,m)}\dots p^{(3, 1)}$. The last equality follows from the fact that the channel $(\widetilde{\mathcal{E}}^{AB} )^{\dagger}$ consists of unitary channels and Pauli channels, and thus each $(\widetilde{\mathcal{E}}^{AB}_i)^{\dagger}$ is a unital channel. Therefore, the term that decides the global optima in the HST is given by 
\begin{align}
\sigma^{(3)} = (\widetilde{\mathcal{E}}^{AB} )^{\dagger}\circ \widehat{\mathcal{Q}}^{AB}\left(\sum_{(\vec{a}, \vec{b}) \neq (\vec{0}, \vec{0})} \beta^{AB}_{\vec{a}, \vec{b}}  \big(\prod_{i =1}^m c^{(i)}_{\vec{a}, \vec{b}}\big)  WX_A^{\vec{a}}Z_A^{\vec{b}}W^{\dagger}\otimes X^{\vec{a}}_BZ^{\vec{b}}_B\right)~,
\end{align}
where we have omitted the scaling factors. 
Let 
$
\Ft_{\HST}(V) \propto f(V) \coloneqq \Tr\left[ \widetilde{P}_{\vec{0}} \sigma^{(3)}\right]$.
Then
\begin{align}
f(V) &= \Tr\left[(\widehat{\mathcal{Q}}^{AB}\circ\widetilde{\mathcal{E}}^{AB})  (\widetilde{P}_{\vec{0}})  \left(\sum_{(\vec{a}, \vec{b}) \neq (\vec{0}, \vec{0})} \beta^{AB}_{\vec{a}, \vec{b}}  \big(\prod_{i =1}^m c^{(i)}_{\vec{a}, \vec{b}}\big)  WX_A^{\vec{a}}Z_A^{\vec{b}}W^{\dagger}\otimes X^{\vec{a}}_BZ^{\vec{b}}_B\right)\right] \label{def:falpha-hst-thm1} \\
& = \Tr\left[\sum_{\substack{(\vec{a}, \vec{b}) \neq (\vec{0}, \vec{0})\\\widetilde{\vec{a}},\widetilde{\vec{b}}}} \kappa^{AB}_{\vec{a}, \widetilde{\vec{a}}, \vec{b}, \widetilde{\vec{b}}}  Z^{\widetilde{\vec{b}}}_A X_A^{\widetilde{\vec{a}}} W X^{\vec{a}}_A Z^{\vec{b}}_A W^{\dagger} \otimes Z^{\widetilde{\vec{b}}}_B X^{\widetilde{\vec{a}}}_B X^{\vec{a}}_B Z^{\vec{b}}_B \right]\\
& = \Tr_A\left[ \sum_{(\vec{a}, \vec{b}) \neq (\vec{0}, \vec{0})} \kappa^{AB}_{\vec{a}, \vec{a}, \vec{b}, \vec{b}} Z^{\vec{b}}_A X^{\vec{a}}_A W X^{\vec{a}}_A Z^{\vec{b}}_A W^{\dagger  } \right]~.
\end{align}
The second equality follows from \eqref{eq:EN-on-P0N1}, where we set  $
\kappa^{AB}_{\vec{a}, \widetilde{\vec{a}}, \vec{b}, \widetilde{\vec{b}}} \coloneqq (1/2^{2n})\widetilde{m}^{AB}_{\widetilde{\vec{a}}, \widetilde{\vec{a}}, \widetilde{\vec{b}}, \widetilde{\vec{b}}}  \widehat{p}^A_{\widetilde{\vec{a}}, \widetilde{\vec{b}}} 
\widehat{q}^{AB}_{\widetilde{\vec{a}}, \widetilde{\vec{a}}, \widetilde{\vec{b}}, \widetilde{\vec{b}}} \beta^{AB}_{\vec{a}, \vec{b}}  \big(\prod_{i =1}^m c^{(i)}_{\vec{a}, \vec{b}}\big)$.
The last equality follows from \eqref{eq-pauli_prop}. Let $\VB_d^{\opt}$ denote the sets of unitaries that optimize $F_{\HST}(V)$ (and hence $C_{\HST}(V)$) such that
\begin{align}
\mathbb{V}_d^{\opt} &= \{V' \in \VB_d : W = (V')\ad U= e^{i \phi} \id \,, \quad \text{for some} \quad \phi\in [0,2\pi]\}.\label{eq-setHST}
\end{align}
We remark that this set of unitaries also optimizes $F_{\LHST}(V)$ (and hence $C_{\LHST}(V)$). 
Then, for $V' \in \VB_d$ we find
$
f(V')=\sum_{(\vec{a}, \vec{b}) \neq (\vec{0}, \vec{0})} \kappa^{AB}_{\vec{a}, \vec{a}, \vec{b}, \vec{b}}$.
Let 
\begin{equation}
    T(V) \coloneqq \sum_{(\vec{a}, \vec{b}) \neq (\vec{0}, \vec{0})}\sqrt{\kappa_{\vec{a},\vec{a},\vec{b},\vec{b}}^{AB}} X_A^{\vec{a}}Z_A^{\vec{b}}W^\dagger \otimes\ket{\vec{a},\vec{b}},\quad S(V)\coloneqq \sum_{(\vec{a}',\vec{b}') \neq (\vec{0}, \vec{0})}\sqrt{\kappa_{\vec{a}',\vec{a}',\vec{b}',\vec{b}'}^{AB}} W^\dagger X_A^{\vec{a}'}Z_A^{\vec{b}'}\otimes\ket{\vec{a}',\vec{b}'}.
\end{equation}
Consider the following inequality: 
\begin{align}
f(V) = \vert \left< S(V) , T(V) \right> \vert
 \leq  \sqrt{\Tr(S(V)^{\dagger} S(V))} \sqrt{\Tr(T(V)^{\dagger} T(V))}
 = \sum_{(\vec{a}, \vec{b}) \neq (\vec{0}, \vec{0})} \kappa^{AB}_{\vec{a}, \vec{a}, \vec{b}, \vec{b}} \label{eq:arg-cs-ineq2},
\end{align}
where we used the Cauchy-Schwarz inequality. Moreover, note that the inequality in \eqref{eq:arg-cs-ineq2} is saturated for any matrix $V' \in \VB_d$ if we assume that the coefficients $\kappa_{\vec{a},\vec{a},\vec{b},\vec{b}}^{AB}$ characterizing the noise satisfy $\kappa_{\vec{a},\vec{a},\vec{b},\vec{b}}^{AB}\geq 0$. Therefore, the set of unitaries that optimize $F_{\HST}(V)$ (and hence $C_{\HST}(V)$) is $\widetilde{\VB}_d^{\opt}=\mathbb{V}_d^{\opt}$. According to Definition \ref{def:OPR}, the latter means that $C_{\HST}$ exhibits strong-OPR to Noise Model 1 in Definition \ref{def:HST-noise-model-1}.

We now show that the cost function $C_{\LHST}$ exhibits strong-OPR to Noise Model 1. The LHST corresponds to the optimization of the following function:
\begin{align}
\Ft_{\LHST}(V) \propto g(V) = \Tr\left[(\widehat{\mathcal{Q}}^{AB}\circ\widetilde{\mathcal{E}}^{'AB})(\widetilde{Q}_{00})  \left(\sum_{(\vec{a}, \vec{b}) \neq (\vec{0}, \vec{0})} \beta^{AB}_{\vec{a}, \vec{b}}  \big(\prod_{i =1}^m c^{(i)}_{\vec{a}, \vec{b}}\big) WX_A^{\vec{a}}Z_A^{\vec{b}}W^{\dagger}\otimes X^{\vec{a}}_BZ^{\vec{b}}_B \right)\right]~,
\end{align}
where we replaced the disentangling and measurement channels in \eqref{def:falpha-hst-thm1} with \eqref{def:eff-meas-op-LHST}. Consider the following:
\begin{align}
g(V) &= \Tr\Bigg[\left(\frac{1}{2^{2}} \frac{1}{n} \sum_{j=1}^n \sum_{a'_j, b'_j = 0}^1 \widetilde{m}^{AB}_{a'_j,a'_j,b'_j ,b'_j}  \widehat{p}^{A_j}_{a'_j, b'_j}\widehat{q}_{a'_j, a'_j, b'_j, b'_j}  Z^{b'_j}_{A_j}X^{a'_j}_{A_j} \otimes Z^{b'_j}_{B_j} X^{a'_j}_{B_j} \otimes \id_{\overline{A}_j \overline{B}_j} \right)\nonumber\\
&\qquad \qquad \qquad \qquad \qquad \qquad \times  \left(\sum_{(\vec{a}, \vec{b}) \neq (\vec{0}, \vec{0})} \beta^{AB}_{\vec{a}, \vec{b}}  \big(\prod_{i =1}^m c^{(i)}_{\vec{a}, \vec{b}}\big)  WX_A^{\vec{a}}Z_A^{\vec{b}}W^{\dagger}\otimes X^{\vec{a}}_BZ^{\vec{b}}_B\right)\Bigg]\\
& =\Tr\left[\sum_{j=1}^n \sum_{(\vec{a}, \vec{b}) \neq (\vec{0}, \vec{0})} \sum_{a'_j, b'_j = 0}^1  \xi^{(j)}_{\vec{a}, a'_j, \vec{b}, b'_j} (Z^{b'_j}_{A_j} X^{a'_j}_{A_j} \otimes \id_{\overline{A}_j})WX^{\vec{a}}_AZ^{\vec{b}}_AW^{\dagger} \otimes Z^{b'_j}_{B_j}X^{a'_j}_{B_j} X^{a_j}_{B_j} Z^{b_j}_{B_j} X^{\overline{a}_j}_{\overline{B}_j}Z^{\overline{b}_j}_{\overline{B}_j} \right] \label{eq:f1alpha-thm2-1}\\
& = \Tr_A\left[\sum_{j=1}^n  \sum_{(\vec{a}, \vec{b}) \neq (\vec{0}, \vec{0})} \sum_{a'_j, b'_j = 0}^1  \xi^{(j)}_{\vec{a}, a'_j, \vec{b}, b'_j} (Z^{b'_j}_{A_j} X^{a'_j}_{A_j} \otimes \id_{\overline{A}_j})WX^{\vec{a}}_AZ^{\vec{b}}_AW^{\dagger} \otimes \Tr_{B_j}\left(Z^{b'_j}_{B_j}X^{a'_j}_{B_j} X^{a_j}_{B_j} Z^{b_j}_{B_j}\right) \Tr_{\overline{B}_j}\left( X^{\overline{a}_j}_{\overline{B}_j}Z^{\overline{b}_j}_{\overline{B}_j}\right) \right]\label{eq-midstep}\\
& = \Tr_A\left[ \sum_{j=1}^n \sum_{ (a_j, b_j)\neq (0,0) }  \xi^{(j)}_{a_j, a_j, b_j, b_j} (Z^{b_j}_{A_j}X^{a_j}_{A_j} \otimes \id_{\overline{A}_j}) (W (X^{a_j}_{A_j} Z^{b_j}_{A_j} \otimes \id_{\overline{A}_j} )W^{\dagger} )  \right]\\
& \leq \sum_{j=1}^n \sum_{(a_j, b_j) \neq (0,0)} \xi^{(j)} _{a_j, a_j, b_j, b_j}\,, \label{eq-proofCHST}
\end{align}
where in \eqref{eq-midstep} we have split $\Tr_B$ into a contribution from qubit $B_j$ and a contribution on all qubits except $B_j$, and where $
\xi^{(j)}_{\vec{a}, a'_j, \vec{b}, b'_j} = (1/4n) \widetilde{m}^{A,B}_{a'_j,a'_j,b'_j b'_j}  \widehat{p}^{A_j}_{a'_j, b'_j}\widehat{q}_{a'_j, a'_j, b'_j, b'_j} \beta^{AB}_{\vec{a}, \vec{b}}  \left(\prod_{i =1}^m c^{(i)}_{\vec{a}, \vec{b}}\right)$. 
The first equality is derived from \eqref{def:eff-meas-op-LHST}, while the inequality follows from the arguments similar to \eqref{eq:arg-cs-ineq2}.

Here we remark that the inequality \eqref{eq-proofCHST} is saturated for any unitary matrix in the set of unitaries that optimize $F_{\HST}(V)$ (and hence $C_{L\HST}(V)$) given by \eqref{eq-setHST}. Hence,  $C_{\LHST}$ exhibits strong-OPR to Noise Model 1 in Definition \ref{def:HST-noise-model-1} if we assume that the coefficients $ \xi^{(j)}_{a_j, a_j, b_j, b_j}$ characterizing the noise satisfy $ \xi^{(j)}_{a_j, a_j, b_j, b_j}\geq 0$.
\end{proof}

\section{Proof of Theorem \ref{thm2}}\label{sec:proof-thm-2}

\begin{theoremApp}
The cost functions $C_{\HST}$ and $C_{\LHST}$ exhibit strong-OPR to Noise Model 2 in Definition \ref{def:HST-noise-model-2}.
\end{theoremApp}
\begin{proof}
We break up the HST circuit into three time intervals similar to  Section \ref{sec:proof-thm-1}. We again assume that the global depolarizing noise occurs on system $AB$ during all three time intervals and the global depolarizing noise occurs on system $A$ during the implementation of $\VC\ad\circ \UC$. Moreover, suppose that a global Pauli channel $\mathcal{Q}^{AB}$ followed by a global non-unital Pauli channel $\mathcal{P}_{\text{NU}}^{A}$ acts at time $\tau_1$. Furthermore, a global pauli channel $\widehat{\mathcal{Q}}^{AB}$ acts at time $\tau_2$, while a 
global Pauli channel acts continuously on the system $B$ in between $\tau_1$ and $\tau_2$. 

The state at $\tau_1$ is given by
\begin{align}
\rho^{(1)} &= p^{(1)}\mathcal{P}^{A}_{\text{NU}}\circ \mathcal{Q}^{AB}\circ\widetilde{\EC}^{AB}(\rho^{(0)}) + (1-p^{(1)})\mathcal{P}^{A}_{\text{NU}}( \id / 2^{2n}) \\
& = p^{(1)} \left[\frac{1}{2^{2n}} \sum_{(\vec{a}, \vec{b})\neq (\vec{0}, \vec{0})} \widetilde{\beta}^{AB}_{\vec{a}, \vec{b}}X_A^{\vec{a}}Z^{\vec{b}}_A\otimes X^{\vec{a}}_BZ^{\vec{b}}_B \right] +  \frac{1}{ 2^{2n}}\id + \frac{1}{2^{2n}}\sum_{(\vec{g}, \vec{h})\neq (\vec{0}, \vec{0})}d_{\vec{g}, \vec{h}} X^{\vec{g}}_A Z^{\vec{h}}_A\otimes \id_B  ~.
\end{align}
The first equality follows from arguments similar to those used  to derive \eqref{eq:state-tau1-hst-thm1}--\eqref{eq:state-tau1-hst-thm1-1}.
The last equality follows from \eqref{eq:global-noisy-E-on-0state}, \eqref{def:non-un-Puali1}, and \eqref{def:non-un-Puali2}, where 
$\widetilde{\beta}^{AB}_{\vec{a}, \vec{b}} = m^{AB}_{\vec{a}, \vec{a}, \vec{b}, \vec{b}}  q^{AB}_{\vec{a}, \vec{a}, \vec{b}, \vec{b}}c_{\vec{a}, \vec{b}}$.  

At $\tau_2$ the state is 
\begin{align}\label{eq:rho2-hst-thm3}
\rho^{(2)} =  \widehat{\mathcal{Q}}^{AB}( \mathcal{D}^{AB}_{p^{(2,l)}}\circ\DC_{s^{(2,l)}}^A \circ (\mathcal{W}_l \otimes \widehat{\mathcal{P}}^{B}_l) \dots \mathcal{D}^{AB}_{p^{(2, 1)}}\circ \DC_{s^{(2,1)}}^A \circ (\mathcal{W}_1 \otimes \widehat{\mathcal{P}}_1^{B})(\rho^{(1)}) )~.
\end{align}
The term that depends on $W$ in \eqref{eq:rho2-hst-thm3} is given by
\begin{align}
\widetilde{\rho}^{(2)} =\frac{1}{2^{2n}} \widehat{\mathcal{Q}}^{AB}\left[ p^{(2)} s^{(2)} p^{(1)}  \sum_{(\vec{a}, \vec{b})\neq (\vec{0}, \vec{0})} \widetilde{\beta}^{AB}_{\vec{a}, \vec{b}}\big( \prod_{i}^{l} \widehat{p}^{(i)}_{\vec{a}, \vec{b}} \big)  WX_A^{\vec{a}}Z^{\vec{b}}_AW^{\dagger}\otimes X^{\vec{a}}_BZ^{\vec{b}}_B + \sum_{(\vec{g}, \vec{h})\neq (\vec{0}, \vec{0})}d_{\vec{g}, \vec{h}}  WX^{\vec{g}}_A Z^{\vec{h}}_AW^{\dagger}\otimes \id_B  \right],
\end{align}
where we used the definition of Pauli channels from \eqref{eq:Pauli-channel} and \eqref{eq:2n-qubit-Pauli-channel}.
By omitting the scaling factors, the relevant term after $\tau_3$ is given by 
\begin{multline}
\widetilde{\rho}^{(3)} = (\widetilde{\mathcal{E}}^{AB})^{\dagger}\circ\widehat{\mathcal{Q}}^{AB} \left( p^{(2)}  s^{(2)} p^{(1)}  \sum_{(\vec{a}, \vec{b})\neq (\vec{0}, \vec{0})} \widetilde{\beta}^{AB}_{\vec{a}, \vec{b}}\big( \prod_{i}^{l} \widehat{p}^{(i)}_{\vec{a}, \vec{b}} \big)  WX_A^{\vec{a}}Z^{\vec{b}}_AW^{\dagger}\otimes X^{\vec{a}}_BZ^{\vec{b}}_B \right)\\
+ (\widetilde{\mathcal{E}}^{AB})^{\dagger}\circ\widehat{\mathcal{Q}}^{AB}\left(\sum_{(\vec{g}, \vec{h})\neq (\vec{0}, \vec{0})}d_{\vec{g}, \vec{h}}  WX^{\vec{g}}_A Z^{\vec{h}}_AW^{\dagger}\otimes \id_B\right).
\end{multline}
Let 
$
\Ft_{HST}(V) \propto f(V) \coloneqq \Tr\left[ \widetilde{P}_{\vec{0}} \widetilde{\rho}^{(3)}\right]$. 
Then
\begin{multline}\label{def:falpha-hst-thm3}
f(V)  = \Tr\left[(\widehat{\mathcal{Q}}^{AB}\circ\widetilde{\mathcal{E}}^{AB}) (\widetilde{P}_{\vec{0}}) \left(p^{(2)}  s^{(2)} p^{(1)}  \sum_{(\vec{a}, \vec{b})\neq (\vec{0}, \vec{0})} \widetilde{\beta}^{AB}_{\vec{a}, \vec{b}}\big( \prod_{i}^{l} \widehat{p}^{(i)}_{\vec{a}, \vec{b}} \big)  WX_A^{\vec{a}}Z^{\vec{b}}_AW^{\dagger}\otimes X^{\vec{a}}_BZ^{\vec{b}}_B\right) \right]   \\ 
 +\Tr\left[ (\widehat{\mathcal{Q}}^{AB}\circ\widetilde{\mathcal{E}}^{AB}) (\widetilde{P}_{\vec{0}})  \left( \sum_{(\vec{g}, \vec{h})\neq (\vec{0}, \vec{0})}d_{\vec{g}, \vec{h}}  WX^{\vec{g}}_A Z^{\vec{h}}_AW^{\dagger}\otimes \id_B\right)\right]~.
\end{multline} 
Moreover, for simplicity we denote
\begin{align}
f_1(V) & \coloneqq \Tr\left[ (\widehat{\mathcal{Q}}^{AB}\circ\widetilde{\mathcal{E}}^{AB}) (\widetilde{P}_{\vec{0}}) \left(\sum_{(\vec{a}, \vec{b})\neq (\vec{0}, \vec{0})} \widetilde{\beta}^{AB}_{\vec{a}, \vec{b}}\big( \prod_{i}^{l} \widehat{p}^{(i)}_{\vec{a}, \vec{b}} \big)  WX_A^{\vec{a}}Z^{\vec{b}}_AW^{\dagger}\otimes X^{\vec{a}}_BZ^{\vec{b}}_B\right) \right] ~, \\
f_2(V) & \coloneqq\Tr\left[ (\widehat{\mathcal{Q}}^{AB}\circ\widetilde{\mathcal{E}}^{AB}) (\widetilde{P}_{\vec{0}})  \left( \sum_{(\vec{g}, \vec{h})\neq (\vec{0}, \vec{0})}d_{\vec{g}, \vec{h}}  WX^{\vec{g}}_A Z^{\vec{h}}_AW^{\dagger}\otimes \id_B\right)\right]~.
\end{align}
Let us focus on $f_1(V)$ and $f_2(V)$ individually. Consider the following:
\begin{align}
f_1(V)
& =  \Tr\left[\sum_{\substack{(\vec{a}, \vec{b})\neq (\vec{0}, \vec{0})\\\vec{a}',\vec{b}'}} \vartheta^{AB}_{\vec{a}, \vec{a}', \vec{b}, \vec{b}'}  Z^{\vec{b}'}_A X^{\vec{a}'}_A W X^{\vec{a}}_A Z^{\vec{b}}_AW^{\dagger} \otimes Z^{\vec{b}'}_BX^{\vec{a}'}_B X^{\vec{a}}_B Z^{\vec{b}}_B \right] = \Tr\left[\sum_{(\vec{a}, \vec{b}) \neq (\vec{0}, \vec{0})} \vartheta^{AB}_{\vec{a}, \vec{a}, \vec{b}, \vec{b}} Z^{\vec{b}}_A X^{\vec{a}}_A W X^{\vec{a}}_A Z^{\vec{b}}_A W^{\dagger} \right]\nonumber\\
& \leq \sum_{(\vec{a}, \vec{b}) \neq (\vec{0}, \vec{0})} \vartheta^{AB}_{\vec{a}, \vec{a}, \vec{b}, \vec{b}}\,. \label{eq:f1alpha-opt-proof-hst-thm2}
\end{align}
The first equality follows from \eqref{eq:EN-on-P0N1},
where 
$
\vartheta^{AB}_{\vec{a}, \vec{a}', \vec{b}, \vec{b}'} = (1/2^{2n})\widetilde{m}^{AB}_{\vec{a}', \vec{a}', \vec{b}', \vec{b}'}  \widehat{\widetilde{p}}^A_{\vec{a}', \vec{b}'} \widehat{q}^{AB}_{\vec{a}', \vec{a}', \vec{b}', \vec{b}'} \widetilde{\beta}^{AB}_{\vec{a}, \vec{b}}\big( \prod_{i}^{l} \widehat{p}^{(i)}_{\vec{a}, \vec{b}} \big)$.
The inequality follows from the arguments similar to \eqref{eq:arg-cs-ineq2}.
Here, the last inequality in \eqref{eq:f1alpha-opt-proof-hst-thm2} is saturated for any matrix $V$ in the set $\mathbb{V}_d^{\opt}$ of unitaries that optimize $F_{\HST}(V)$ (and hence $C_{L\HST}(V)$) given by \eqref{eq-setHST}.

On the other hand, 
\begin{align}
f_2(V)
& = \Tr\left[\sum_{\substack{(\vec{g}, \vec{h})\neq (\vec{0}, \vec{0})\\\vec{a}',\vec{b}'}} \varsigma^{AB}_{\vec{g}, \vec{a}', \vec{h}, \vec{b}'}   Z^{\vec{b}'}_A X^{\vec{a}'}_A W X^{\vec{g}}_A Z^{\vec{h}}_AW^{\dagger} \otimes Z^{\vec{a}'}_BX^{\vec{b}'}_B \right]\label{eq:f2alpha-1-thm3}\\
& = \Tr_A\left[ \sum_{\substack{(\vec{g}, \vec{h})\neq (\vec{0}, \vec{0})\\\vec{a}',\vec{b}'}} \varsigma^{AB}_{\vec{g}, \vec{a}', \vec{h}, \vec{b}'}   Z^{\vec{b}'}_A X^{\vec{a}'}_A W X^{\vec{g}}_A Z^{\vec{h}}_AW^{\dagger} \otimes \Tr_B\left(Z^{\vec{a}'}_BX^{\vec{b}'}_B\right) \right]\\
& = \sum_{(\vec{g}, \vec{h})\neq (\vec{0}, \vec{0})}\varsigma^{AB}_{\vec{g}, \vec{0}, \vec{h}, \vec{0}} \Tr_A\left(X^{\vec{g}}_A Z^{\vec{h}}_A\right)\\
&=0\label{eq:f2alpha-2-thm3}
\end{align}
where $\varsigma^{AB}_{\vec{g}, \vec{a}', \vec{h}, \vec{b}'} = (1/2^{2n})\widetilde{m}^{AB}_{\vec{a}', \vec{a}', \vec{b}', \vec{b}'}  \widehat{\widetilde{p}}^A_{\vec{a}', \vec{b}'}\widehat{q}^{AB}_{\vec{a}', \vec{a}', \vec{b}', \vec{b}'} d_{\vec{g}, \vec{h}}\widetilde{\beta}^{AB}_{\vec{a}, \vec{b}}\big( \prod_{i}^{l} \widehat{p}^{(i)}_{\vec{a}, \vec{b}} \big)$. From the last equality it follows that $f_2(V)$ is independent of $W$ (and hence of $V$) and thus does not affect the global optima. Therefore, from \eqref{eq:f1alpha-opt-proof-hst-thm2} it follows that the set of unitaries that optimize $\Ft_{\HST}(V)$ (and hence $\Ct_{\HST}(V)$) is $\widetilde{\VB}_d^{\opt}=\mathbb{V}_d^{\opt}$.
From Definition \ref{def:OPR} this  implies that $C_{\HST}$ exhibits strong-OPR to Noise Model 2 in Definition \ref{def:HST-noise-model-2} if we assume that the coefficients $\vartheta^{AB}_{\vec{a}, \vec{a}, \vec{b}, \vec{b}}$ characterizing the noise satisfy $\vartheta^{AB}_{\vec{a}, \vec{a}, \vec{b}, \vec{b}}\geq 0$.  

We now show that the cost function $C_{\LHST}$ exhibits strong-OPR to Noise Model 2. In particular, in the LHST we want to optimize the following function: 
\begin{align}
\Ft_{\LHST}(V) \propto g(V)  = \Tr\left[(\widehat{\mathcal{Q}}^{AB}\circ\widetilde{\mathcal{E}}^{'AB})(\widetilde{Q}_{00}) \left(p^{(2)}  s^{(2)} p^{(1)}  \sum_{(\vec{a}, \vec{b})\neq (\vec{0}, \vec{0})} \widetilde{\beta}^{AB}_{\vec{a}, \vec{b}}\big( \prod_{i}^{l} \widehat{p}^{(i)}_{\vec{a}, \vec{b}} \big)  WX_A^{\vec{a}}Z^{\vec{b}}_AW^{\dagger}\otimes X^{\vec{a}}_BZ^{\vec{b}}_B\right) \right]  \nonumber  \\ 
 +\Tr\left[ (\widehat{\mathcal{Q}}^{AB}\circ\widetilde{\mathcal{E}}^{'AB})(\widetilde{Q}_{00})  \left( \sum_{(\vec{g}, \vec{h})\neq (\vec{0}, \vec{0})}d_{\vec{g}, \vec{h}}  WX^{\vec{g}}_A Z^{\vec{h}}_AW^{\dagger}\otimes \id_B\right) \right]~,     
\end{align}
where we replaced the disentangling and measurement channels in \eqref{def:falpha-hst-thm3} with \eqref{def:eff-meas-op-LHST}. We now break up $g(V)$ into two different functions. 
\begin{align}
g_1(V) &\coloneqq \Tr\left[(\widehat{\mathcal{Q}}^{AB}\circ\widetilde{\mathcal{E}}^{'AB})(\widetilde{Q}_{00}) \left(\sum_{(\vec{a}, \vec{b})\neq (\vec{0}, \vec{0})} \widetilde{\beta}^{AB}_{\vec{a}, \vec{b}}\big( \prod_{i}^{l} \widehat{p}^{(i)}_{\vec{a}, \vec{b}} \big)  WX_A^{\vec{a}}Z^{\vec{b}}_AW^{\dagger}\otimes X^{\vec{a}}_BZ^{\vec{b}}_B \right)\right]  ~,\\
g_2(V) &\coloneqq \Tr\left[(\widehat{\mathcal{Q}}^{AB}\circ\widetilde{\mathcal{E}}^{'AB})(\widetilde{Q}_{00})\left(\sum_{(\vec{g}, \vec{h})\neq (\vec{0}, \vec{0})}d_{\vec{g}, \vec{h}}  WX^{\vec{g}}_A Z^{\vec{h}}_AW^{\dagger}\otimes \id_B\right)\right] ~.
\end{align}
By using  arguments similar to those used to derive Eqs.\ \eqref{eq:f2alpha-1-thm3}--\eqref{eq:f2alpha-2-thm3} and from \eqref{def:eff-meas-op-LHST}, it follows that $g_2(V)$ is independent of $W$ (and hence of $V$). Therefore, to prove the noise resilience of the LHST, we focus only on $g_1(V)$. We then get: 
\begin{align}\label{eq:f1alpha-thm4-1}
g_1(V) = \Tr\left[\sum_{j=1}^n \sum_{(\vec{a}, \vec{b})\neq (\vec{0}, \vec{0})} \sum_{a'_j, b'_j =0}^1 \tau^{(j)}_{\vec{a}, a'_j, \vec{b}, b'_j} (Z^{b'_j}_{A_j} X^{a'_j}_{A_j} \otimes \id_{\overline{A}_j})WX^{\vec{a}}_AZ^{\vec{b}}_AW^{\dagger} \otimes Z^{b'_j}_{B_j}X^{a'_j}_{B_j} X^{a_j}_{B_j} Z^{b_j}_{B_j} X^{\overline{a}_j}_{\overline{B}_j}Z^{\overline{b}_j}_{\overline{B}_j} \right]~,
\end{align}
where $\tau^{(j)}_{\vec{a}, a'_j ,\vec{b}, b'_j} =(1/4n) \widetilde{m}^{AB}_{a'_j,a'_j,b'_j, b'_j}  \widehat{\widetilde{p}}^{A_j}_{a', b'}
\widehat{q}^{AB}_{a'_j, a'_j, b'_j, b'_j}\widetilde{\beta}^{AB}_{\vec{a}, \vec{b}}\big( \prod_{i}^{l} \widehat{p}^{(i)}_{\vec{a}, \vec{b}} \big) $. We note that \eqref{eq:f1alpha-thm4-1} is similar to \eqref{eq:f1alpha-thm2-1}. Therefore, from the proof in Section \ref{sec:proof-thm-1} it follows that 
\begin{align}
g_1(V) &\leq \sum_{j=1}^n \sum_{(a_j, b_j)\neq (0,0)} \tau^{(j)}_{a_j, a_j, b_j, b_j}\,.
\end{align}
Where the inequality is saturated for unitaries $V'$ in the set $\mathbb{V}_d^{\opt}$ of unitaries that optimize $F_{\LHST}(V)$ (and hence $C_{\LHST}(V)$) given by \eqref{eq-setHST}. This further implies that 
\begin{align}
    g(V) \leq g(V'), \quad \text{for all} \quad V'\in \mathbb{V}_d^{\opt}=\widetilde{\VB}_d^{\opt}~.
\end{align}
Thus $C_{\LHST}$ exhibits strong-OPR to Noise Model 2 if we assume that the 
coefficients $\tau^{(j)}_{a_j, a_j, b_j, b_j}$ characterizing the noise satisfy $\tau^{(j)}_{a_j, a_j, b_j, b_j}\geq0$. 
\end{proof}

\section{Proof of Theorem~\ref{thm3}}\label{sec:proof-thm-3}

\begin{theoremApp}
The cost functions $C_{\LET}$ and $C_{\LLET}$ exhibit weak-OPR, as defined in Definition~\ref{def:OPR}, to Noise Model 3 in Definition \ref{defFISCnoise}.
\end{theoremApp}

\begin{proof}
Let us remark that in order to show weak-OPR to Noise Model 3 we just need to consider Pauli noise acting at $\tau_1$ and measurement noise, since noise resilience to global depolarizing noise follows from Lemma~\ref{lemmaGDN}.

We first consider the $C_{\LET}$ cost function. From Eqs.~\eqref{eq-allzero-Pauli} and \eqref{eq:Pauli-channel} we get that the action of the Pauli channel acting at time $\tau_1$ is given by
\begin{align}
    \mathcal{P}(\dya{\vec{0}})&=\sum_{\vec{l},\vec{k}}q_{\vec{l},\vec{k}}X^{\vec{l}}Z^{\vec{k}}\dya{\vec{0}} Z^{\vec{k}}X^{\vec{l}}=\sum_{\vec{l}} q_{\vec{l}} \dya{\vec{l}}\,, \label{eq-oxPC}
\end{align}
where $q_{\vec{l}}=\sum_{\vec{k}} q_{\vec{l},\vec{k}}$. Similarly, we can express the noisy measurement POVM from Definition \ref{def-noisy-POVM} as
\begin{equation}
\Pt_{\vec{0}} = \bigotimes_{j=1}^{n} \left( p_{0 0}^{(j)} \dya{0} + p_{0 1}^{(j)} \dya{1} \right)=\sum_{\vec{i}} p_{\vec{i}} \dya{\vec{i}}\,,    
\label{eq-oxMN}
\end{equation}
with $\vec{i}=i_1 i_2\ldots i_n$ a bit string and $p_{\vec{i}}=p_{0 i_1}^{(1)}p_{0 i_2}^{(2)}\ldots p_{0 i_n}^{(n)}$. For the present noise model we are interested in determining the optimum of the function
\begin{align}
     \widetilde{G}_{\LET}(V)&=\Tr \left[\widetilde{P}_{\vec{0}}(\mathcal{W}\circ\mathcal{P})(\dya{\vec{0}})\right],\label{eq-cost-global-ox}
\end{align}
with $\mathcal{W}=\mathcal{V}\ad \circ \mathcal{U}$ the channel that implements $U$ followed by $V\ad$. Then, by means of \eqref{eq-oxPC} and \eqref{eq-oxMN} we find
\begin{align}
     \widetilde{G}_{\LET}(V)&=\Tr \left[(\sum_{\vec{i}} p_{\vec{i}} \dya{\vec{i}})(\sum_{\vec{l}} q_{\vec{l}} W\dya{\vec{l}}W\ad)\right] =\sum_{\vec{i},\vec{l}} p_{\vec{i}} q_{\vec{l}} w_{\vec{i}\vec{l}}\,, \label{eq-PNt1MN}
\end{align}
where $w_{\vec{i}\vec{l}}=|\matl{\vec{i}}{W}{\vec{l}}|^2$ are the matrix elements of a  doubly stochastic matrix such that $     \sum_{\vec{i}} w_{\vec{i}\vec{l}}=\sum_{\vec{l}} w_{\vec{i}\vec{l}}=1$.

Let us now denote by $\vec{q}^{\downarrow}$ the vector with elements $q_{\vec{i}}$  ordered in decreasing order. 
Similarly, we denote  by $\vec{p}^{\downarrow}$ the  vector with elements  $p_{\vec{l}}$ ordered in decreasing order. Additionally, let  $\{\ket{q_r}\}$ and $\{\ket{p_s}\}$ be the basis in which $\vec{q}^{\downarrow}$ and $\vec{p}^{\downarrow}$ are ordered, respectively, i.e.,
\begin{equation}
  \mathcal{P}(\dya{\vec{0}})=\sum_{r} q^{\downarrow}_{r} \dya{q_{r}}\,,\qquad \text{and} \qquad \Pt_{\vec{0}} =\sum_{s} p^{\downarrow}_{s} \dya{p_{s}}\,.
\end{equation}
Then, from the permutation inequality (or the rearrangement inequality) \cite{hardy1952inequalities} we have
\begin{align}
     \widetilde{G}_{\LET}(V)=\sum_{\vec{i},\vec{l}} p_{\vec{i}} q_{\vec{l}} w_{\vec{i}\vec{l}}\leq\vec{p}^{\downarrow}\cdot\vec{q}^{\downarrow} \,. \label{eq-LETord}
\end{align}
The inequality in \eqref{eq-LETord} is saturated for matrices $W\in \mathbb{S}$, where $\mathbb{S}$ is  the subset of the Permutation Group which maps $\{\ket{p_s}\}$ to $\{\ket{q_r}\}$. We remark here that if the vector $\vec{q}^{\downarrow}$ (or $\vec{p}^{\downarrow}$) has components of equal magnitude, then the set $\mathbb{S}$ is degenerate.
Moreover, note  that
\begin{equation}
    p_{\vec{0}}\geq p_{\vec{i}}, \quad \text{and} \quad q_{\vec{0}}\geq q_{\vec{i}}, \quad \forall {\vec{i}} \neq {\vec{0}}\,, \label{eq-inequality-pq}
\end{equation} 
where the second inequality follows from Definition \ref{def-UPC}, while the first inequality always holds since  $p_{\vec{0}}=\prod_{j=1}^n p_{0 0}^{(j)}$, and since we have assumed that $p_{0 0}^{(j)} > p_{0 1}^{(j)} $ $\forall j$. 

We now recall that $\VB_d^{\opt}$  denotes the set of unitaries that optimize $C_{\LET}(V)$ and $C_{\LLET}(V)$, i.e., $\forall V' \in \VB_d^{\opt}$ we have $W'\ket{0}=(V')\ad U\ket{0}=\ket{0}$ (up to a global phase), which entails $w'_{\vec{i}\vec{0}}=w'_{\vec{0}\vec{i}}=\delta_{\vec{i},\vec{0}}$, and hence Eq.\ \eqref{eq-PNt1MN} becomes   
\begin{equation}\label{eqGvpr}
\widetilde{G}_{\LET}(V')=p_{\vec{0}} q_{\vec{0}} + \sum_{\vec{i},\vec{l}\neq {\vec{0}}} p_{\vec{i}} q_{\vec{l}} w'_{\vec{i}\vec{l}} \,.
\end{equation}
Since  $p_{\vec{0}}\geq p_{\vec{i}}$  and $q_{\vec{0}}\geq q_{\vec{i}}$ $\forall {\vec{i}}$ then the first term in \eqref{eqGvpr} corresponds to the first term in the summation $\vec{p}^{\downarrow}\cdot\vec{q}^{\downarrow}=\sum_r \vec{q}^{\downarrow}_r \vec{p}^{\downarrow}_r $. Hence, in order to saturate \eqref{eq-LETord} we now need that $W'\in \mathbb{S}$, i.e., the $(n-1)\times(n-1)$  principal submatrix of $W'$ with matrix elements  $\mted{\vec{z}}{W'}{\vec{z'}}$ (such that $\vec{z},\vec{z'}\neq\vec{0}$) must map $\{\ket{p_s}\}$ to $\{\ket{q_r}\}$ (where $s\neq 0$ and $r\neq 0$). Combining this result with \eqref{eq-LETord} we have that for any matrix $V$ in $\VB_d$ (the set of $d\times d$ unitary matrices)
\begin{align}
     \widetilde{G}_{\LET}(V)\leq\vec{p}^{\downarrow}\cdot\vec{q}^{\downarrow}=\widetilde{G}_{\LET}(V') \,, \label{eq-LETord2}
\end{align}
where $V'\in \widetilde{\VB}_d^{\opt}$ and where 
\begin{equation}
    \widetilde{\VB}_d^{\opt} = \{V' \in \VB_d : W=(V')\ad U \in \mathbb{S}\}.
\end{equation}
Evidently, not all matrices in $\VB_d^{\opt}$ are in $\mathbb{S}$, which then entails that $ \widetilde{\VB}_d^{\opt} \subseteq \VB_d^{\opt}$, and further means that $C_\LET$ exhibits weak-OPR to Noise Model 3 according to Definition \ref{def:OPR}.

Let us now consider the noise resilience of LLET to Noise Model 3  of Definition  \ref{defFISCnoise}. We are now interested in the optimum of
\begin{align}
     \widetilde{G}_{\LLET}(V)&=\frac{1}{n}\sum_{j=1}^n \Tr \left[\left((p_{0 0}^{(j)} \dya{0} + p_{0 1}^{(j)} \dya{1})\otimes\id^{\overline{A}_j}\right)(\mathcal{W}\circ\mathcal{P})(\dya{\vec{0}})\right]\\
     &=\frac{1}{n}\sum_{j=1}^n \Tr \left[\left((p_{0 0}^{(j)} \dya{0} + p_{0 1}^{(j)} \dya{1})\otimes\id^{\overline{A}_j}\right)(\sum_{\vec{l}} q_{\vec{l}} W\dya{\vec{l}}W\ad)\right]\,.
\end{align}
For any matrix  $ V' \in \VB_d^{\opt}$ we have $W'\ket{\vec{0}}=(V')\ad U\ket{\vec{0}}=\ket{\vec{0}}$ (up to global phase) and  $\sum_{\vec{l}} q_{\vec{l}} W'\dya{\vec{l}}(W')\ad=q_{\vec{0}}\dya{\vec{0}}+ \sum_{\vec{l}\neq\vec{0}} q_{\vec{l}} W'\dya{\vec{l}}(W')\ad$, which leads to
\begin{equation}\label{EqGLLETVpr}
\widetilde{G}_{\LLET}(V')=    \frac{1}{n}\sum_{j=1}^n p_{00}\jj q_{\vec{0}} + \frac{1}{n}\sum_{j=1}^n \Tr \left[\left((p_{0 0}^{(j)} \dya{0} + p_{0 1}^{(j)} \dya{1})\otimes\id^{\overline{A}_j}\right))(\sum_{\vec{l}\neq\vec{0}} q_{\vec{l}} W'\dya{\vec{l}}(W')\ad)\right]\,.
\end{equation}
On the other hand, for any unitary matrix $V\in \VB_d$  
\begin{align}
\widetilde{G}_{\LLET}(V)&= \frac{1}{n}\sum_{j=1}^n  \Tr \left[\left((p_{0 0}^{(j)} \dya{0} + p_{0 1}^{(j)} \dya{1})\otimes\id^{\overline{A}_j}\right)q_{\vec{0}}W\dya{\vec{0}}W\ad\right]\nonumber\\
&\qquad +\frac{1}{n}\sum_{j=1}^n \Tr \left[\left((p_{0 0}^{(j)} \dya{0} + p_{01}^{(j)} \dya{1})\otimes\id^{\overline{A}_j}\right)(\sum_{\vec{l}\neq\vec{0}} q_{\vec{l}} W\dya{\vec{l}}W\ad)\right] \nonumber\\
&\leq \frac{1}{n}\sum_{j=1}^n \left( \Tr \left[ p_{0 0}^{(j)} q_{\vec{0}} \id  W\dya{\vec{0}}W\ad\right] + \Tr \left[\left((p_{0 0}^{(j)} \dya{0} + p_{0 1}^{(j)} \dya{1})\otimes\id^{\overline{A}_j}\right)(\sum_{\vec{l}\neq\vec{0}} q_{\vec{l}} W\dya{\vec{l}}W\ad)\right] \right)\nonumber\\
&= \frac{1}{n}\sum_{j=1}^n  p_{0 0}^{(j)} q_{\vec{0}} + \frac{1}{n}\sum_{j=1}^n \Tr \left[\left((p_{0 0}^{(j)} \dya{0} + p_{0 1}^{(j)} \dya{1})\otimes\id^{\overline{A}_j}\right)(\sum_{\vec{l}\neq\vec{0}} q_{\vec{l}} W\dya{\vec{l}}W\ad)\right]\label{eqGLLETV}
\end{align}
where the inequality follows from  the fact that $p_{0 0}^{(j)}>p_{0 1}^{(j)}$, and hence
\begin{equation}
    (p_{0 0}^{(j)} \dya{0} + p_{0 1}^{(j)} \dya{1})\otimes\id^{\overline{A}_j}\leq (p_{0 0}^{(j)} \dya{0} + p_{0 0}^{(j)} \dya{1})\otimes\id^{\overline{A}_j}\leq p_{0 0}^{(j)} \id\,.
\end{equation}
We can then simplify Eq.\ \eqref{eqGLLETV} as 
\begin{align}
\widetilde{G}_{\LLET}(V)&\leq \frac{1}{n}\sum_{j=1}^n  p_{0 0}^{(j)} q_{\vec{0}} + \frac{1}{n}\sum_{j=1}^n \sum_{\vec{l}\neq\vec{0},\vec{k}\neq\vec{0}} q_{\vec{l}} p_{\vec{k}}\jj w_{\vec{k}\vec{l}}=  \frac{1}{n}\sum_{j=1}^n  p_{0 0}^{(j)} q_{\vec{0}} +  \sum_{\vec{l}\neq\vec{0},\vec{k}\neq\vec{0}} q_{\vec{l}} \tilde{p}_{\vec{k}} w_{\vec{k}\vec{l}}\label{eq-IneqW}\,,
\end{align}
where we have $p_{\vec{k}}\jj=p_{00}\jj$ if $k_j=0$, and $p_{\vec{k}}\jj=p_{01}\jj$ if $k_j=1$. On the the other hand, in the second equality of \eqref{eq-IneqW} we have defined $\tilde{p}_{\vec{k}}=\frac{1}{n}\sum_{j=1}^n p_{\vec{k}}\jj$. 
Finally,  the following inequality follows again from the rearrangement inequality
\begin{align}
\widetilde{G}_{\LLET}(V)
&\leq  \frac{1}{n}\sum_{j=1}^n  p_{0 0}^{(j)} q_{\vec{0}}+\sum_{\vec{l}\neq\vec{0},\vec{k}\neq\vec{0}} q_{\vec{l}}^\downarrow \tilde{p}_{\vec{k}}^\downarrow\label{eq-IneqW2}\,,
\end{align}
which is saturated for matrices $W\in \mathbb{S'}$, where $\mathbb{S'}$ is a subset of the Permutation Group such that  $\sum_{\vec{l}\neq\vec{0},\vec{k}\neq\vec{0}} q_{\vec{l}} \tilde{p}_{\vec{k}} w_{\vec{k}\vec{l}}=\sum_{\vec{l}\neq\vec{0},\vec{k}\neq\vec{0}} q_{\vec{l}}^\downarrow \tilde{p}_{\vec{k}}^\downarrow$. Here  $q^\downarrow$ and $\tilde{p}^\downarrow$ are vectors with components $q_{\vec{l}}$ and $\tilde{p}_{\vec{k}}$ in decreasing order, respectively. Hence, we can define the set of matrices which saturate \eqref{eq-IneqW2} as
\begin{equation}
    \widetilde{\VB}_d^{\opt} = \{V' \in \VB_d : W=(V')\ad U \in \mathbb{S'}\}.
\end{equation}

While any matrix in  $\VB_d^{\opt}$ saturates the inequality in \eqref{eqGLLETV}, only a subset will also  saturate \eqref{eq-IneqW2}. Hence, $ \widetilde{\VB}_d^{\opt} \subseteq \VB_d^{\opt}$, and $C_\LLET$ exhibits weak-OPR to Noise Model 3 according to Definition \ref{def:OPR}.
\end{proof}

\section{Proof of Corollaries \ref{cor:Pauli-conjugation}-\ref{cor:TensorProduct-FISC}}\label{sec:proof-corol}

\begin{corollaryApp}\label{cor1-app}
The cost functions $C_{\HST}$ and $C_{\LHST}$ exhibit strong-OPR to a noise model that includes the following: (1) all noise processes in  Noise Model 1, as well as
(2) a noise process during the implementation of $\mathcal{W} = \mathcal{W}_k \circ \cdots \circ \mathcal{W}_1 = \mathcal{V}^{\dagger}\circ\mathcal{U}$ (i.e., in the time interval between $\tau_1$ and $\tau_2$) in which global Pauli channels $\{\mathcal{P}^A_1, \dots, \mathcal{P}^A_k\}$ act on system $A$, such that the overall channel on $A$ is $\mathcal{P}^A_k\circ \mathcal{W}_k  \cdots \circ\mathcal{P}^A_1\circ \mathcal{W}_1$, provided that the following condition is satisfied:
\begin{align}\label{eq:pauli-channel-conjugation1}
(\mathcal{P}^A_k\circ \mathcal{W}_k  \cdots \circ\mathcal{P}^A_1\circ \mathcal{W}_1)(\cdot) = (\mathcal{W}_k \circ \mathcal{W}_{k-1}\cdots \circ \mathcal{W}_1\circ \widehat{\mathcal{P}}^{A}) (\cdot)~.
\end{align}
Here $\widehat{\mathcal{P}}^{A}$ is also a Pauli channel, and the channels $\UC$, $\VC\ad$, and $\WC$ correspond to conjugating the state by the unitaries $U$, $V\ad$, and $W$, respectively.
\end{corollaryApp}
\begin{proof}
This follows from the fact that the overall noisy channel acting during the implementation of $\WC$ is mathematically equivalent to a Pauli channel followed by the unitary $\WC$, as described in the condition \eqref{eq:pauli-channel-conjugation1} and by invoking Theorem \ref{thm1}, which allows for Pauli channel noise at time $\tau_1$.   
\end{proof}

\begin{corollaryApp}\label{cor2-app}
Let the $W=V\ad U$ gate sequence have the form $W = W^A_2 W^A_1$ with $W_1^A$ be composed only of Clifford gates. Then the cost functions $C_{\HST}$ and $C_{\LHST}$ exhibit strong-OPR to a noise model that includes the following: (1) all noise processes in Noise Model~1, as well as  
(2) a noise process during the implementation of $\mathcal{W}^A_1 = \mathcal{W}_{1,k} \circ \cdots \circ \mathcal{W}_{1,1}$, in which global Pauli channels $\{\mathcal{P}^A_1, \dots, \mathcal{P}^A_k\}$ act on system $A$, such that the overall channel on $A$ is $\mathcal{P}^A_k\circ \mathcal{W}_{1,k}  \cdots \circ\mathcal{P}^A_1\circ \mathcal{W}_{1,1}$.   
\end{corollaryApp}
\begin{proof}
From Lemma \ref{lem:clifford-conjugation} it follows that Clifford unitaries satisfy the condition in \eqref{eq:pauli-channel-conjugation1}. Therefore, Corollary \ref{cor2-app} is a special case of Corollary \ref{cor1-app}. 
\end{proof}

\begin{corollaryApp}\label{cor3-app}
Let the $W=V\ad U$ gate sequence have the form $W = W_2^A W_1^A$ with $W_1^A =  W^{A'}_1 \otimes W^{A''}_1$ being a tensor product, i.e., $W$ is a tensor product up to a particular time. Then the cost functions $C_{\HST}$ and $C_{\LHST}$ exhibit strong-OPR to a noise model that includes the following: (1) all noise processes in Noise Model 1, as well as (2) a noise process during the implementations of $\WC_1^{A'} = \WC_{1,k}^{A'} \circ \cdots \circ \WC_{1,1}^{A'}$ and $\WC_1^{A''} = \WC_{1,l}^{A''} \circ \cdots \circ \WC_{1,1}^{A''}$ in which local depolarizing channels $\{\DC^{A'}_{1,1}, \dots ,\DC^{A'}_{1,k}\}$ and  $\{\DC^{A''}_{1,1}, \dots ,\DC^{A''}_{1,l}\}$ act on subsystems 
$A'$ and $A''$, respectively, such that the overall channel on $A'A''$ is $(\DC^{A'}_{1,k}\circ \WC^{A'}_{1,k}\dots \DC^{A'}_{1,1}\circ \WC^{A'}_{1,1})\otimes(\DC^{A''}_{1,l}\circ \WC^{A''}_{1,l}\dots \DC^{A''}_{1,1}\circ \WC^{A''}_{1,1})$. 
\end{corollaryApp}
\begin{proof}
Let $\rho$ denote a quantum state. Consider the following chain of equalities:
\begin{align}
(\DC^{A'}_p\otimes \DC^{A''}_q)(\WC^{A'}\otimes \WC^{A''})(\rho) & =(\IC^{A'}\otimes \DC^{A''}_q) \left( p(\WC^{A'}\otimes \WC^{A''}(\rho)) +(1-p)\pi^{A'} \Tr_{A'}((\WC^{A'}\otimes \WC^{A''})(\rho)) \right)\\
&=(\IC^{A'}\otimes \DC^{A''}_q) \left( p(\WC^{A'}\otimes \WC^{A''}(\rho)) +(1-p)\pi^{A'} \Tr_{A'}((\IC^{A'}\otimes \WC^{A''})(\rho)) \right)\\
& = (\IC^{A'}\otimes \DC^{A''}_q)\left((\WC^{A'}\otimes \WC^{A''})(p\rho + (1-p)\pi^{A'}\Tr_{A'}(\rho)) \right)\\
& = (\IC^{A'}\otimes \DC^{A''}_q)(\WC^{A'}\otimes \WC^{A''})(\DC^{A'}_p(\rho))\\
& = (\WC^{A'}\otimes \WC^{A''})(\DC^{A'}_p\otimes \DC^{A''}_q)(\rho)~,\label{eq:depolarizing-conjugation}
\end{align}
where $\pi^{A'}$ is a maximally mixed state on system $A'$. 
Therefore, the result follows by applying \eqref{eq:depolarizing-conjugation} several times and invoking Corollary \ref{cor:Pauli-conjugation}.
\end{proof}

\begin{corollaryApp}
The cost functions $C_{\HST}$ and $C_{\LHST}$ exhibit strong-OPR to the following noise model: (1) all noise processes in  Noise Model 2, as well as 
(2) a noise process during the implementation of $\mathcal{W} = \mathcal{W}_k \circ \cdots \circ \mathcal{W}_1 = \mathcal{V}^{\dagger}\circ\mathcal{U}$ (i.e., in the time interval between $\tau_1$ and $\tau_2$) in which global non-unital Pauli channels $\{\mathcal{P}^A_{\text{NU},1}, \dots, \mathcal{P}^A_{\text{NU},k}\}$ act on system $A$  such that the overall channel on $A$ is $\mathcal{P}^A_{\text{NU},k}\circ \mathcal{W}_k  \cdots \circ\mathcal{P}^A_{\text{NU},1}\circ \mathcal{W}_1$, provided that the following condition is satisfied:
\begin{align}\label{eq:non-unital-pauli-channel-conjugation}
(\mathcal{P}^A_{\text{NU},k}\circ \mathcal{W}_k  \cdots \circ\mathcal{P}^A_{\text{NU},1}\circ \mathcal{W}_1)(\cdot) = (\mathcal{W}_k \circ \mathcal{W}_{k-1}\cdots \mathcal{W}_1\circ \widehat{\mathcal{P}}^A_{\text{NU}}) (\cdot)~,
\end{align}
where $\widehat{\mathcal{P}}^A_{\text{NU}}$ is also a Pauli channel.
\end{corollaryApp}
\begin{proof}
This follows from the fact that the overall noisy channel acting during the implementation of $\WC$ is mathematically equivalent to a non-unital Pauli channel followed by the unitary $\WC$, as described in the condition \eqref{eq:non-unital-pauli-channel-conjugation} and by invoking Theorem \ref{thm2}, which allows for non-unital Pauli noise at time $\tau_1$. 
\end{proof}

\begin{corollaryApp}
The cost functions $C_{\HST}$ exhibits strong-OPR to the following noise model: (1) global depolarizing noise acting continuously throughout the circuit, (2) global non-unital Pauli noise on system $A$ at a fixed time in between $\tau_1$ and $\tau_2$. 
\end{corollaryApp}
\begin{proof}
Let us decompose $\WC$ as $\WC = \WC_2\circ \WC_1$ such that the non-unital Pauli channel $\mathcal{P}^A_{\text{NU}}$ acts at time $\tau'$ between $\WC_1$ and $\WC_2$, with the overall channel between $\tau_1$ and $\tau_2$ given by $\mathcal{W}_2\circ\mathcal{P}^A_{\text{NU}}\circ\mathcal{W}_1$.  The state at time $\tau_1$ is
\begin{align}
\rho^{(1)} = p^{(1)} \dya{\Phi^+} + (1-p^{(1)}) \id/d~,
\end{align}
where  $p^{(1)} = p^{(k, 1)} \cdots p^{(1,1)}$ corresponds to the continuous depolarizing channel as discussed in Appendix~\ref{sec:proof-thm-1}. We break up the time interval in between $\tau'$ and $\tau_1$ into $l$ steps. The state at time $\tau'$ is given by 
\begin{align}
\widetilde{\rho}^{(2)} &= \mathcal{P}^A_{\text{NU}} \circ \mathcal{D}^{AB}_{q^{(2,l)}} \circ \mathcal{W}_1^l \cdots \circ \mathcal{D}^{AB}_{q^{(2,1)}}\circ \mathcal{W}_1^1(\rho^{(1)})     \\
& = \mathcal{P}^A_{\text{NU}}(p^{(1)}q^{(2)} \mathcal{W}_1(\dya{\Phi^+}) + (1-p^{(1)}q^{(2)})\id/d)\\
& = p^{(1)}q^{(2)}\mathcal{P}^A_{\text{NU}}(\mathcal{W}_1(\dya{\Phi^+})) + (1-p^{(1)}q^{(2)})\id/d + (1-p^{(1)}q^{(2)})\frac{1}{d} \sum_{(\vec{g}, \vec{h})\neq (\vec{0}, \vec{0})} d_{\vec{g}, \vec{h}} X_A^{\vec{g}}Z^{\vec{h}}_A \otimes \id_B~,
\end{align}
where $q^{(2)} = q^{(2, k)}\cdots q^{(2,1)}$ and $\mathcal{W}_1 = \mathcal{W}_1^l \cdots \mathcal{W}_1^1$. Similarly, we break up the the time interval between $\tau_2$ and $\tau'$ into $m$ steps. The term that depends on $\mathcal{W}$ at time $\tau_2$ is given by 
\begin{align}
\widetilde{\sigma}^{(2)} = p^{(1)}q^{(2)}r^{(2)} \mathcal{W}_2 \circ \mathcal{P}_{\text{NU}}^A\circ \mathcal{W}_1 (\dya{\Phi^+}) + r^{(2)}(1-p^{(1)}q^{(2)})\frac{1}{d}\sum_{(\vec{g}, \vec{h})\neq (\vec{0}, \vec{0})}d_{\vec{g},\vec{h}} W_2 X^{\vec{g}}_A Z^{\vec{h}}_A W^{\dagger}_2 \otimes \id_B~.
\end{align}
Let 
\begin{align}
    \Ft_{\HST}(V) \propto f(V) \coloneqq \Tr[\dya{\Phi^+}\widetilde{\sigma}^{(2)}]~.
\end{align}
Moreover, for simplicity we denote
\begin{align}
    f_1(V)& \coloneqq \Tr\left[\dya{\Phi^+}(\mathcal{W}_2 \circ \mathcal{P}^A_{\text{NU}}\circ\mathcal{W}_1 )(\dya{\Phi^+}) \right]~,\\
    f_2(V) & \coloneqq \Tr\left[ \dya{\Phi^+}( W_2 X^{\vec{g}}_A Z^{\vec{h}}_A W^{\dagger}_2 \otimes \id_B) \right]~.
\end{align}
Consider the followings:
\begin{align}
f_1(V)  & =    \Tr\left[\dya{\Phi^+}(\mathcal{W}_2^A \circ \mathcal{P}^A_{\text{NU}})((\mathcal{I}_A \otimes (\mathcal{W}^{T}_1)^B)(\dya{\Phi^+}_{AB}))\right]\\
& = \Tr\left[ (\mathcal{I}_A \otimes (\mathcal{W}^{*}_1)^B)(\dya{\Phi^+}) (\mathcal{W}_2^A \circ \mathcal{P}^A_{\text{NU}})(\dya{\Phi^+}) \right]\\
& = \Tr\left[((\mathcal{W}^{\dagger}_1)^A \otimes \mathcal{I}_B)(\dya{\Phi^+}) (\mathcal{W}_2^A \circ \mathcal{P}^A_{\text{NU}})(\dya{\Phi^+})  \right]\\
& = \Tr\left[\dya{\Phi^+} (\mathcal{W}_1^A\circ \mathcal{W}_2^A \circ \mathcal{P}^A_{\text{NU}})(\dya{\Phi^+}) \right]\\
& \leq f_1(V')~\,,
\end{align}
where $V'\in \VB_d^{\opt}$, and where $\VB_d^{\opt}$ denote the sets of unitaries that optimize $F_{\HST}(V)$ (and hence $C_{\HST}(V)$) as defined in \eqref{eq-setHST}.
The first and third equalities follow from the ricochet property. The last equality corresponds to the case when there is non-unital Pauli noise at time $\tau_1$ and no other noise in the HST circuit, which is a special case of Theorem \ref{thm2}. Therefore, the inequality follows from Theorem \ref{thm2}. 
Moreover, by  using the arguments similar to \eqref{eq:f2alpha-1-thm3}--\eqref{eq:f2alpha-2-thm3}, we find that $f_2(V)$ is independent of $W$. This completes the proof. 
\end{proof}

\begin{corollaryApp}\label{cor6-app}
The cost functions $C_{\LET}$ and $C_{\LLET}$ exhibit weak-OPR to a noise model that includes the following: (1) all noise processes in  Noise Model 3, as well as
(2) a noise process during the implementation of $\WC = \mathcal{W}_k \circ \cdots \circ \mathcal{W}_1 = \mathcal{V}^{\dagger}\circ\mathcal{U}$ in which global Pauli channels $\{\mathcal{P}_1, \dots, \mathcal{P}_k\}$ act, such that the overall channel is $\mathcal{P}_k\circ \mathcal{W}_k  \cdots \circ\mathcal{P}_1\circ \mathcal{W}_1$, provided that the following condition is satisfied:
\begin{align}\label{eq:pauli-channel-conjugation-LET2_App}
(\mathcal{P}_k\circ \mathcal{W}_k  \cdots \circ\mathcal{P}_1\circ \mathcal{W}_1)(\cdot) = (\mathcal{W}_k \circ \mathcal{W}_{k-1}\cdots \circ \mathcal{W}_1\circ \widehat{\mathcal{P}}) (\cdot)~.
\end{align}
where $\widehat{\mathcal{P}}$ is also a Pauli channel.
\end{corollaryApp}
\begin{proof}
This follows from arguments similar to Corollary \ref{cor1-app} and by invoking Theorem \ref{thm3}.  
\end{proof}

\begin{corollaryApp}\label{cor7-app}
Let the $W=V\ad U$ gate sequence have the form $W = W^A_2 W^A_1$ with $W_1^A$ be composed only of Clifford gates. Then the cost functions $C_{\LET}$ and $C_{\LLET}$ exhibit weak-OPR to a noise model that includes the following: (1) all noise processes in  Noise Model 3, as well as  
(2) a noise process during the implementation of $\mathcal{W}^A_1 = \mathcal{W}_{1,k} \circ \cdots \circ \mathcal{W}_{1,1}$, in which global Pauli channels $\{\mathcal{P}^A_1, \dots, \mathcal{P}^A_k\}$ act on system $A$, such that the overall channel on $A$ is $\mathcal{P}^A_k\circ \mathcal{W}_{1,k}  \cdots \circ\mathcal{P}^A_1\circ \mathcal{W}_{1,1}$. 
\end{corollaryApp}
\begin{proof}
This corollary is a special case of Corollary \ref{cor6-app}, since Lemma \ref{lem:clifford-conjugation} implies that Clifford unitaries satisfy \eqref{eq:pauli-channel-conjugation-LET2_App}. 
\end{proof}

\begin{corollaryApp}
Let the $W=V\ad U$ gate sequence have the form $W = W_2^A W_1^A$ with $W_1^A =  W^{A'}_1 \otimes W^{A''}_1$ being a tensor product, i.e., $W$ is a tensor product up to a particular time. Then the cost functions $C_{\LET}$ and $C_{\LLET}$ exhibit weak-OPR to a noise model that includes the following: (1) all noise processes in Noise Model 3, as well as (2) a noise process during the implementations of $\WC_1^{A'} = \WC_{1,k}^{A'} \circ \cdots \circ \WC_{1,1}^{A'}$ and $\WC_1^{A''} = \WC_{1,l}^{A''} \circ \cdots \circ \WC_{1,1}^{A''}$ in which local depolarizing channels $\{\DC^{A'}_{1,1}, \dots ,\DC^{A'}_{1,k}\}$ and  $\{\DC^{A''}_{1,1}, \dots ,\DC^{A''}_{1,l}\}$ act on subsystems 
$A'$ and $A''$, respectively, such that the overall channel on $A'A''$ is $(\DC^{A'}_{1,k}\circ \WC^{A'}_{1,k}\dots \DC^{A'}_{1,1}\circ \WC^{A'}_{1,1})\otimes(\DC^{A''}_{1,l}\circ \WC^{A''}_{1,l}\dots \DC^{A''}_{1,1}\circ \WC^{A''}_{1,1})$. \end{corollaryApp}
\begin{proof}
This follows from arguments similar to the proof of Corollary \ref{cor3-app} and by invoking Corollary \ref{cor6-app}. 
\end{proof}

\end{document}